\documentclass[a4paper,twocolumn,11pt,unpublished]{article}
\pdfoutput=1
\usepackage[utf8]{inputenc}
\usepackage[english]{babel}
\usepackage[T1]{fontenc}
\usepackage{amsmath}
\usepackage{hyperref}
\usepackage{amssymb}
\usepackage{amsfonts}
\usepackage{amsbsy}
\usepackage{latexsym}

\usepackage{tikz}
\usepackage{lipsum}
\usepackage[numbers,sort&compress]{natbib}

\newcommand{\nbar}{{\overline n}}
\newcommand{\NN}{\mathbb{N}}
\newcommand{\RR}{\mathbb{R}}
\newcommand{\CC}{\mathbb{C}}

\newcommand{\Chi}{\mathbb{\chi}}
\newcommand{\refeq}[1]{(\ref{#1})}
\newcommand{\dist}{\mbox{dist}}
\newcommand{\Un}{{\mbox{U}(n)}}
\newcommand{\un}{{\mathfrak{u}(n)}}
\newcommand{\unbar}{{\mathfrak{u}(\nbar)}}
\newcommand{\Unbar}{{\mbox{U}(\nbar)}}
\newcommand{\Image}{\mbox{\rm Im}}

\newenvironment{proof}{\mbox{\textbf{Proof.}}}{ \hfill $\Box$}
\newcommand{\pdist}{\mbox{\textrm{pdist}}}
\newcommand{\optgoal}{\mbox{\textrm{optgoal}}}

\newcommand{\trace}{\mbox{\textrm{trace}}}
\newcommand{\diag}{\textrm{diag}}
\newcommand{\Xgoalbar}{{\overline X}_{{goal}_\ell}}
\newcommand{\Xgoalbarellminusone}{{\overline X}_{{goal}_{\ell-1}}}
\newcommand{\Window}{$\mathbb{W}$}
\newcommand{\BallV}[1]{{\overline B}_{#1}(I)} 
\newcommand{\BallZ}[1]{{{\overline B}}^Z_{#1}(0)} 

\newtheorem{theorem}{Theorem}

\newtheorem{problem}{Problem}
\newtheorem{corollary}{Corollary}
\newtheorem{lemma}{Lemma}
\newtheorem{proposition}{Proposition}
\newtheorem{definition}{Definition}

\newtheorem{remark}{Remark}

\begin{document}

\title{On encoded quantum gate generation by iterative Lyapunov-based methods}

\author{Paulo Sergio Pereira da Silva~\thanks{{Polytechnic School -- PTC,
		University of S\~ao Paulo (USP), S\~ao
		Paulo, Brazil}\\
	{\tt paulo@lac.usp.br}\\
	{\em orcid: 0000-0002-5458-182X}
	}
\and Pierre Rouchon~\thanks{{Laboratoire  de Physique de l’Ecole Normale Supérieure, Mines Paris-PSL, Inria,  ENS-PSL, Université PSL, CNRS, Paris, France.}\\
		{\tt pierre.rouchon@minesparis.psl.eu}\\
		{\em orcid:0000-0001-6160-5634}}
	}

\maketitle

\begin{abstract}
  The problem of encoded quantum gate generation is studied in this paper. The idea is to consider a quantum system of higher dimension $n$ than the dimension $\nbar$ of the quantum gate to be synthesized. Given two orthonormal subsets $\mathbb{E} = \{e_1, e_2, \ldots, e_{\nbar}\}$ and $\mathbb F = \{f_1, f_2, \ldots, f_{\nbar}\}$ of $\CC^n$, the problem  of encoded quantum gate generation consists
in obtaining an open loop control law defined in an interval $[0, T_f]$ in a way that all initial states $e_i$ are steered to $\exp(\jmath \phi) f_i, i=1,2, \ldots ,\nbar$ up to some desired precision and to some global phase $\phi \in \RR$.
This problem includes the classical (full) quantum gate generation problem, when $\nbar = n$,  the state preparation problem, when $\nbar = 1$, and finally the encoded gate generation when $ 1 < \nbar < n$.  Hence, three problems are unified here within a unique common approach. The \emph{Reference Input Generation Algorithm (RIGA)} is generalized in this work for considering the encoded gate generation problem for closed quantum systems. A suitable Lyapunov function is derived   from  the orthogonal projector on the support  of the encoded gate. Three  case-studies of physical interest indicate the potential interest of such numerical  algorithm: two coupled transmon-qubits,   a cavity mode coupled to a transmon-qubit, and  a chain of $N$ qubits, including a large dimensional case for which $N=10$.

\end{abstract}

\section{Introduction}\label{sec1}

State preparation and quantum gate generation are important  quantum control problems.
Optimal control could be applied in order to solve these problems
\cite{PalaoK2002PRL,PalaoK2003PRA}. However,
this last approach is restricted to low orders due to complexity issues \cite{SchirF2011NJoP}. Lyapunov stabilization \cite{L1,L2,L3,L4,L5,L6,LS1,SilPerRou14,SilPerRou16} may be considered for large orders. The drawback is that Lyapunov techniques produces slow solutions in general. The Krotov method  \cite{SchirF2011NJoP}, GRAPE \cite{KHANEJA2005,SecondGRAPE}, and  RIGA \cite{PerSilRou19,CODE_OCEAN_CONSTANT} in the piecewise-constant setting, and  CRAB \cite{CRAB}, GOAT \cite{MacShaTanFra15} or RIGA \cite{PerSilRou19,CODE_OCEAN_SMOOTH}  in the smooth setting, are efficient methods for solving the
quantum gate generation problem. A comparison between Krotov, GRAPE, CRAB and GOAT  methods for small dimensions (two and three qubits) is presented in \cite{CRAB_GOAT_GRAPE}.
A comparison between GRAPE and RIGA in the piecewise-constant setting can be found in \cite{PerSilRou19}. In the present paper, some comparisons between RIGA and GRAPE are presented
first in a theoretical context (see section \ref{sComparison}). A second set of comparisons is presented in a more practical context, by considering  two examples that were also studied in \cite{LeuAbdKocSch17} and also another example that was considered
in \cite{HeeEtAl15}. This comparison between RIGA and GRAPE that is based on the results of the numerical experiments is included in the conclusions section of the paper (Section \ref{sConclusions}).It is worth noting that RIGA is able to consider large Hilbert space dimensions with a performance that seems to be superior to GRAPE (see for instance the third example, where $n = 1024$), although a complete comparison must be done in several case studies in the same computational platform\footnote{The authors think that the runtime of their MATLAB implementation of  \cite{CODE_OCEAN_SMOOTH} may be improved a lot, for instance by compilation in language $C$.}.

This paper considers the problem of encoding a quantum gate of dimension $\nbar$
inside a quantum system whose propagator $X(t)$ evolves on $\Un$, with $n \geq \nbar$.
An encoded gate corresponds to the case  for which $n$ is greater than  $\nbar$, which is a common situation of quantum control. For instance, in \cite{CavityTransmonGrape} the authors  want to encode a Hadamard gate
inside a coupled cavity-transmon qubit. A strong reason to do this is related to the robustness
of quantum information, that can be improved by the encoding process, since one may recover the lost quantum information after some decoherence events
\cite{Maz14}.  It is worth to mention that many authors have considered different approaches of encoding
quantum information and error correction strategies \cite{Jin12,Nig14,Got01,Mic16,Chi04,Cra16,Fu17,Ris15,Mag15}.

RIGA was introduced in \cite{PerSilRou19} for the case where $n$ coincides to $\nbar$, as well as
the proofs of convergence of this algorithm along with  several numerical experiments with the piecewise-constant\footnote{It means that the control pulses are assumed to be piecewise-constant,
whereas the ``smooth version '' means the the control pulses are assumed to be smooth.}
implementation of RIGA \cite{CODE_OCEAN_CONSTANT}. The proofs of RIGA's convergence for  the case
where $n$ is equal to $\nbar$ can be found in \cite{PerSilRou19}. The authors have generalized those proofs for the encoded case that is presented in this paper, that is, when $\nbar$ is less than $n$. Those proofs are quite similar to the ones of  \cite{PerSilRou19}, but they are rather long ones, and so they are deferred to the Appendix \ref{aMathematical}. 

The authors have already implemented the smooth version of RIGA \cite{CODE_OCEAN_SMOOTH}, that considers any possible values of $\nbar$. The description of this implementation as well as a set of numerical experiments for all cases are included in this work, even if the main contributions regards the encoded case.

In this paper one considers a closed quantum system (Shr\"odinger picture) that is modelled by:
\begin{equation}\label{cqs}
    \dot{X}(t)  =  -\iota \left( H_0 + \sum_{k=1}^{m}u_k(t) H_k \right) X(t),  X(0)= X_0
\end{equation}
where $X \in \mbox{U}(n)$ is the propagator, $S_k = -\iota H_k \in \mathfrak{u}(n), k= 0, 1 \ldots$, where  $\mbox{U}(n)$ and  $\mathfrak{u}(n)$ are respectively the Lie-group  of $n \times n$ unitary matrices and its Lie Algebra\footnote{Recall that $\mathfrak{u}(n)$ is the set of anti-hermitian $n \times n$ matrices.} , and  $u_k(t) \in \RR, k = 1, \ldots , m$ are the controls
and the initial condition $X_0$ is the identity matrix $I$.
One defines two $\nbar$-dimensional subspaces $\mathcal E$ and $\mathcal F$ of $\CC^n$ spanned respectively by the orthonormal sets $\mathbb{E} = \{e_1, e_2, \ldots, e_{\nbar}\}$ and $\mathbb {F} = \{f_1, f_2, \ldots, f_{\nbar}\}$. The subspace $\mathbb E$  is called by \emph{ decoded space} and the subspace $\mathbb F$ is the \emph{ encoded space}. An encoded quantum gate is the linear map $Z_{goal} : \mathcal E \rightarrow \mathcal F$
defined by $Z_{goal} (e_i) = \exp (\jmath \phi) f_i, i=1,2, \ldots, n$, where $\phi \in \RR$ is a global phase.
In some situations, the encoded and the decoded subspaces, may coincide and this fact will be transparent in the context of the mathematical setting that is proposed here (in this case the gate is an endomorphism).

From a control perspective, since the propagator is a representation of a linear map, the studied control problem may be stated as follows:
\begin{problem} \label{Prob1} \emph{(Encoded Quantum Gate Generation Problem)}
Fix orthonormal bases $\mathbb{E} = \{e_1, e_2, \ldots, e_{\nbar}\}$ and $\mathbb {F} = \{f_1, f_2, \ldots, f_{\nbar}\}$, with  $\mathbb E, \mathbb F \subset \CC^n$. Fix a final time $T_f$.
The  \emph{Encoded Quantum Gate Generation Problem} consists in finding an open loop control $u : [0, T_f] \rightarrow \RR^m$ that steers the initial conditions $e_i$ at $t=0$ to $\exp (\jmath \phi) f_i$ at $t=T_f$, for $i=1, \ldots, \nbar$
up to some desired precision and to some global phase $\phi \in \RR$.
In particular, defining the matrices $E$ and $F$ whose columns are respectively given 
by $\mathbb{E}$ and $\mathbb{F}$, then 
the control $u :  [0, T_f] \rightarrow \RR^m$ must steer the system \refeq{cqs} from $X(0) = I$ to $X(T_f) = X_f$ in a way that $X_f E = \exp (\jmath \phi) F$ up to some desired precision.
\end{problem}

The precision of a quantum gate is usually expressed by a function $\mathcal F : \Un \rightarrow [0, 1]$ called \emph{Fidelity}. The function $\mathcal I(\cdot) = 1 - \mathcal F (\cdot)$ is called \emph{Infidelity} function. 
\begin{definition}
\label{dInfidelity}
In this paper, the infidelity function $\mathcal I : \Un \rightarrow \RR$ of an encoded gate is defined by \cite{LeuAbdKocSch17}:
\[
 \mathcal I( X) = 1 - \left( \frac{1}{\nbar}  \left\| \trace \left( F^\dag  X E \right) \right\| \right)^2
\]
where $E$ and $F$ are the matrices defined in Prob. \ref{Prob1}.
\end{definition}
Smaller is $\mathcal I (X)$, smaller is the error
of the quantum gate corresponding to $X$. It can be shown that $\mathcal I (X) = 0$ if and only if there exists a global phase $\phi \in \RR$ such that $X e_i = \exp(\imath \phi) f_i, i=1, \ldots, \nbar$.

The paper is organized as follows. Section \ref{sResults} presents RIGA, and the main objective is to be readable for the interested user of the algorithm, that is available as a public domain software \cite{CODE_OCEAN_SMOOTH}. Three interesting examples are studied in section \ref{sExamples}:   an encoded C-NOT gate generation and state preparation for two coupled transmon-qubits;  an encoded Hadamard gate generation for for a coupled cavity-transmon qubit system;   the generation of a Hadamard gate for a chain of $N$-coupled qubits (the unique example for which $n= \nbar$), including a high dimensional case ($ n= 1024$). 
The reader that is interested in more details about the algorithm may refer to Section \ref{s:Refinements}, including the main properties of the Lyapunov function, the strategy in order to avoid its singular and or critical points, the choice of the seed aof the algorithm and other related questions. 
Section \ref{sImplementation} is devoted to the computational implementation
of RIGA. The main implementation considers that the control pulses $u_k(t)$ are smooth functions for $k=1, \ldots, m$, but the piecewise-constant implementation  is also considered.  In fact, it  allows
to establish a direct connection with GRAPE in sub-section \ref{sComparison} which is summarized by the following remark: the piecewise-constant version RIGA is a kind of closed loop version of GRAPE.  Some conclusions are stated in section \ref{sConclusions}. Finally, the reader that interested in the mathematical details is invited to refer to the Appendices. 

\paragraph{Acknowledgments. }
This project has received funding from the European Research Council (ERC) under the European Union’s Horizon 2020 research and innovation program (grant agreement No. [884762]).

\section{The Reference Input Generation Algorithm (RIGA)}
\label{sResults}

The method called \emph{Reference Input Generation Algorithm (RIGA)} was introduced in \cite{PerSilRou19}
as an efficient method for solving the quantum gate generation problem. It was originally conceived for tackling the full quantum gate generation problem, that is, the case for which the dimension $\nbar$ of the quantum gate coincides with the full dimension $n$ of the quantum system. A convenient choice of Lyapunov function will be the key for adapting RIGA for the case where $\nbar$ is less than $n$. Firstly, a simplified description of RIGA will be given. A more detailed description will be given in section \ref{s:Refinements} for the interested reader.
The main ingredients of the RIGA are  \textbf{(i) a reference trajectory to be tracked}, and \textbf{(ii) an adequate Lyapunov function}. Let us describe all the features of RIGA in the sequel.

\subsection{The reference and the error systems}

Given a reference input $\overline u : [0, T_f] \rightarrow \RR^m$, consider the reference system
 \begin{equation}\label{reference}
    \dot{\overline X}(t)= S_0 \overline X(t) + \sum_{k=1}^{m}{\overline u}_k(t) S_k \overline X(t), \quad {\overline X}(0)={\overline X}_0
\end{equation}
Define the (tracking) error matrix $\widetilde{X}(t)=\overline{X}^{\dag}(t) X(t) \in \Un$.
The dynamics of $\widetilde{X}(t)$ is given by \cite{SilPerRou16}:
\begin{equation}\label{eErrorSystem}
  \dot{\widetilde{X}}(t) = \sum_{k=1}^{m} \widetilde{u}_k(t) \widetilde{S}_k(t)
\widetilde{X}(t), \quad \widetilde{X}(0) = \widetilde{X}_0,
\end{equation}
where $\widetilde{u}_k(t) = u_k(t) - \overline{u}_k(t)$ and\footnote{Recall that $S_k = \jmath H_k$, $k=1, \ldots, n$, where $H_k$ are the control Hamiltonians of the system \ref{cqs}.} 
$ \widetilde{S}_k(t)  =   {\overline X}^\dag (t) S_k {\overline X}(t)$.

In principle, a goal matrix ${X}_{goal}$ is any unitary matrix such that 
\begin{equation}
\label{eXgoal}
    {\overline X}_{goal} E = \exp(\jmath \phi) F
\end{equation}
for a given global phase $\phi \in \RR$, and the reference trajectory
of a step $\ell$ of RIGA is computed in a way that its final condition is a goal matrix. 

\begin{remark}  A trick for avoiding critical and or singular points
of the Lyapunov function will be explained in section \ref{sFirstStrategy}, and in this case the chosen goal matrix ${\overline X}_{goal}$ may not obey the condition \eqref{eXgoal}, at least during 
in certain steps of the algorithm. Then RIGA may be slightly modified in order to include this trick. This is not important for a first reading of this work. The interested reader may refer to Section \ref{s:Refinements}.
\end{remark}

\subsection{The Lyapunov-based tracking control}

Let $\mathcal V : \Un \rightarrow \RR$ be a Lyapunov function. It will be assumed that $\mathcal V$ is well defined and smooth on a open set $G \subset \CC^{n \times n}$ containing $\Un$, and so it admits a gradient $\nabla_X \mathcal V$ for all  $X \in G$. Furthermore
\[
\nabla_X \mathcal V \cdot D
\]
denotes the directional derivative
of $\mathcal V$ at $X$ in the direction $D$ (a complex $n\times n$ matrix). The different choices of the Lyapunov function $\mathcal V$ will determine if RIGA will consider state preparation or (encoded) quantum gate generation. The value of  $\mathcal V(\widetilde X(t))$ is related to a notion of distance between $X(t)$ and $\overline X(t)$.
The tracking control is meant to assure that $\frac{d}{dt} \mathcal V(\widetilde X(t)) \leq 0$.
 Computing $\frac{d}{dt} \mathcal V(\widetilde X(t))$, the linearity of the directional derivative and \refeq{eErrorSystem} gives:
\[
 \dot{\mathcal V}(t) =    \sum_{k=1}^{m} \widetilde{u}_k(t) {\left\{ \nabla_{\widetilde X}{\mathcal V} \cdot ({\widetilde S}_k \widetilde X) \right\} }
\]
 So, choosing some gain $K>0$,  one may define the Lyapunov-based control:
\begin{equation} \label{eUtilk}
 \widetilde{u}_k(t) =  -K \left\{ \nabla_{\widetilde X}{\mathcal V} \cdot ({\widetilde S}_k \widetilde X) \right\} , k=1, \ldots, m
\end{equation}
Then, for the closed loop system, the derivative of the Lyapunov function will be:
\begin{equation}
\label{eVdot}
\dot{\mathcal V}(t) = - \sum_{k=1}^{m} \frac{\widetilde{u}_k(t)^2}{K} =- K~\left\{ \nabla_{\widetilde X}{\mathcal V} \cdot ({\widetilde S}_k \widetilde X) \right\}^2\leq 0
\end{equation}
This will assure that the Lyapunov function is always nonincreasing for the closed loop system. As $\mathcal V(\widetilde X(t))$ is related to a notion of distance
of $\overline X(t)$ and $X(t)$, such distance will be non-increasing on the interval $[0, T_f]$, also.

 \subsection{The Closed Loop System and RIGA}
 \label{sRigaDescription}

 Consider that a reference trajectory  $\overline X$ with final condition  $\overline X (T_f) = {\overline X}_{goal}$ is chosen. Recall that $\overline X : [0, T_f] \rightarrow \Un$ is a solution of \refeq{reference}.
 Define the closed loop system by:
\begin{subequations}
\label{eClosedLoopComplete}
\begin{eqnarray}
  \dot{{\widetilde X}}(t) & = & \left( \sum_{k=1}^{m} \widetilde{u}_k(t) {\widetilde S}_k(t) \right)
{\widetilde X}(t), \\
  {\widetilde  X}(0) & = &  {\overline X}^\dag (0),\\
\label{Sktilde} \widetilde{S}_k(t) &  =  &  {\overline X}^\dag (t) S_k {\overline X}(t)\\
\label{utilde2} { \widetilde u}_k(t) &  =  & - K \left\{ \nabla_{\widetilde X}{\mathcal V} \cdot ({\widetilde S}_k \widetilde X(t)) \right\}
\end{eqnarray}
\end{subequations}
Note that an equivalent way for describing the closed loop system is:
\begin{subequations}
\label{eClosedLoopComplete2}
\begin{eqnarray}
  \dot{{ X}}(t) & = & \left( S_0 + \sum_{k=1}^{m} { u}_k(t) { S}_k(t) \right)
{ X}(t), \\
  { X}(0) & = & I,\\
\label{utilde} { \widetilde u}_k(t) &  =  & - K~\left\{ \nabla_{\widetilde X}{\mathcal V} \cdot ({\widetilde S}_k \overline{X}(t)  X(t)) \right\},\\
{ u}_k(t) &  =  & {\overline u}_k(t) +   {\widetilde u}_k(t)
\end{eqnarray}
\end{subequations}

A simplified description of RIGA that is useful for understanding its main features is the following algorithm:
\begin{itemize}
 \item[$\sharp 1.$] Choose the \textbf{seed input} $\overline u^0: [0, T_f] \rightarrow \RR^m$.\\
 Execute the steps $\ell=1, 2, 3, \ldots$\\[0.2cm]
  \textbf{BEGIN STEP $\ell$}.\\
  \begin{itemize}
 \item[$\sharp 2.$] Set ${\overline u}(\cdot) = {\overline u}^{\ell-1}(\cdot)$. \\
             Set ${\overline X}(T_f) = X_{goal}$. \\
             Integrate numerically backwards the  \textbf{(open loop) reference system}~\eqref{reference} from $T_f$ to $0$. \\
               Obtain ${\overline X}^{\ell}(t)  = \overline X(t)$,  for $t \in [0, T_f]$. \\
\item [$\sharp 3.$]   Set $X(0) = I$. Integrate numerically  forward the \textbf{closed loop system }~\eqref{eClosedLoopComplete2} from $0$ to $T_f$.\\
               Obtain $X^{\ell}(t) = X(t)$ for $t \in [0, T_f]$ . \\
              Set ${\overline u}^{\ell}(\cdot) = {\widetilde u}(\cdot) +  {\overline u}(\cdot)$ (closed loop input)\\
\item [$\sharp 4.$] If the final infidelity $\mathcal I(X^{\ell}(T_f))$  is acceptable, \\
             then terminate RIGA. Otherwise, execute step $\ell+1$.
 \end{itemize}
 \item[]
 \textbf{END STEP $\ell$}
 \end{itemize}
 
 It is clear from \eqref{eVdot} that the Lyapunov function $\mathcal V(\widetilde X(t))$ is non-increasing during a step $\ell$ of RIGA. A main feature of RIGA is the fact that the Lyapunov function is non-increasing  along all the steps of the algorithm (see Appendix \ref{sNonIncreasing} for a complete discussion about this question). In other words, one may expect monotonic convergence of the gate fidelity along the steps of RIGA. This is not the case for gradient based algorithms like GRAPE.
 
\subsection{The partial trace Lyapunov Function}

In \cite{PerSilRou19}, only the case where $\nbar$ is equal to $n$ was considered, that is, the case  for which the dimension of the quantum gate coincides with the dimension of the system. In this work we are also interested in the encoded gates, that is, the case where $\nbar$ is less than $n$. Recall also that $E, F$ are complex $n \times \nbar$  matrices with
columns respectively given by  $\{e_1, e_2, \ldots, e_{\nbar}\}$ and $ \{f_1, f_2, \ldots, f_{\nbar}\}$. Let $X_E, X_F \in \Un$ be any pair of matrices such
that $X_E = [ E \; \widehat E]$ and $X_F = [F \; \widehat F]$. The quantum gate generation problem is then equivalent
to find $u : [0, T_f] \rightarrow \RR^m$ that steers the system from the initial condition $X_E$ to the final condition $X_F$. By right-invariance, the problem could be solved by steering
the system from the identity to $X_{goal} = X_F X_E^\dag$. In principle, this problem could be tackled considering the same approach of
 \cite{PerSilRou19}, but it is clear that the choices of $\widehat E$ and $\widehat F$ are transparent to the problem, and then these choices will generate artificial restrictions in this context.
Define the partial trace function $\mathcal V: \Un \rightarrow \RR$ given by
 \begin{equation}
 \label{eLyap_function}
\mathcal V (\widetilde X) = 2 \nbar - 2 \Re \left[ \trace \left( E^\dag \widetilde X E \right) \right],
 \end{equation}
where $\Re (z)$ denotes the real part of a complex number $z$.
Then Prop. \ref{pCritical} of section \ref{s:Refinements} will show that, if $X, \overline X \in \Un$, and $\widetilde X = {\overline X}^\dag X$, then $ \mathcal V(\widetilde X) = \| \overline X E - X E\|^2$.
If a reference trajectory $\overline X(t)$ is such that $\overline X(t) = X_{goal}$, from Prop. \ref{pCritical} it will be clear that this function is a good candidate to measure the relevant notion of distance between  $\overline X(t)$  and the trajectory $X(t)$. It is easy to show from \refeq{eClosedLoopComplete} and \refeq{eLyap_function}, that the feedback law \refeq{utilde} that is related to this choice of Lyapunov function is given by:
\begin{equation}
\label{eFeedbackLaw}
 \widetilde u_k(t) = 2 K \Re \left[ \trace \left( E^\dag {\widetilde S}_k \widetilde X E \right) \right].
\end{equation}

\subsection{A Lyapunov function for the case where $n$ is equal to $\bar n$}

When $\nbar$ coincides with $n$, a convenient choice of the Lyapunov function is the one defined in  \cite{PerSilRou19}.
Let 
\begin{equation}
\label{eDefine_W}
\mathcal W = \{ X \in \Un ~|~\det(X+I) \neq 0\}.
\end{equation}
Then define $\mathcal V : \mathcal W \rightarrow \RR$ by\footnote{This Lyapunov function $\mathcal V( \widetilde X)$ coincides with
the square of the Frobenius norm of the Cayley transformation $\mathfrak{W}(\widetilde X) = (\widetilde X - I)(\widetilde X + I)^{-1}$ (see section \ref{sSmooth}). In other words, $\mathcal V( \widetilde X) = \|\mathfrak{W}(\widetilde X)\|^2$.}
 \begin{equation}
 \label{eLyap2}
 \begin{array}{rcl}
 \mathcal V(\widetilde X) &  = & \trace[ (\widetilde X-I)^2 (\widetilde X+I)^{-2}] \\
  & = & \| (\widetilde X-I)(\widetilde X+I)^{-1}\|^2
  \end{array}
 \end{equation}
It is shown in that paper that $\mathcal V(\widetilde X) = \sum_{i=1}^n \tan(\frac{\theta_i}{2})^2$, where $\exp(\imath \theta_i)$ are the eigenvalues of $\widetilde X$.
The unique critical point of this function is the identity matrix. However $\mathcal V$ is unbounded in $\Un$.
 Its singular points form the set $\mathcal S = \{ X \in \Un ~|~\det(X+I) =0\}$, corresponding to the existence of at least an eigenvalue
 of $X$ that is equal to minus one. The corresponding (unbounded) feedback law is given by (see \cite{PerSilRou19}):
\begin{equation}
 {\widetilde u}_k(t) =   K \trace [ Z(\widetilde X) {\widetilde S}_k]
\end{equation}
where  $Z({\widetilde X}) = {\widetilde X} ({\widetilde X}-I)({\widetilde X}+I)^{-3}$ is also unbounded.

Note that critical points must be avoided for the partial trace Lyapunov function. However, for the Lyapunov function
\refeq{eLyap2}, the singular points must be avoided. The same strategy that avoids eigenvalues at $-1$ (that may be critical points of the
partial trace  \refeq{eLyap_function}, or singular points of \refeq{eLyap2},
will work for both Lyapunov functions (see Section \ref{sFirstStrategy}).

\subsection{Window and saturation functions}
\label{sWindow}

In some cases it will be useful to replace \refeq{eFeedbackLaw} by:
\begin{equation}
\label{eFeedbackWindow}
 { \widetilde u}_k(t)  =   - K~\mbox{{\Window}(t)} \Re \left[ \trace \left( E^\dag {\widetilde S}_k \widetilde X E \right) \right],
\end{equation}
 where the function $\mbox{{\Window}}(t)$ may be a
\emph{Hamming-like} window function:
\begin{equation}
\label{eHamming}
\mbox{{\Window}}(t)  = \frac{1}{2} \left[ 1 - cos \left(2 \pi \frac{t}{T_f} \right)\right]
\end{equation}
Note that $\mbox{\Window}^{(j)}(0)= \mbox{\Window}^{(j)}(T_f) = 0$, for $j \in \{0, 1\}$, and $\mbox{\Window}(T_f/2) = 1$.
The effect of such window function will be discussed in the examples, but the main issue is to reduce the bandwidth of the
control pulses $u_k(t)$ by assuring that $u_k^{(j)}(0)=u_k^{(j)}(T_f) =0$  for $j \in \{0, 1\}$, with an smooth variation between the endpoints. Those restrictions must hold for the seed input also,
as described in section \ref{sSeed}. Taking the window function identically equal to one could be acceptable in some situations.

It is also useful to consider a saturation function for the control pulses.
The classical saturation function is not smooth.
It is much more convenient to define a smooth saturation function of the form
\[
\phi(x) = \frac{2~x}{\pi} \arctan(\frac{\pi~x}{2})
\]
A smooth saturation between $u^*$ and $-u^*$ can be obtained by the function  
\begin{equation}
\label{eSat}
\mbox{sat}(x) = u^* \phi (\frac{x}{u^*}).
\end{equation}
Figure \ref{fSaturation} compares the traditional saturation and the smooth saturation.
\begin{figure*}[t]
    \centering{\includegraphics[scale=0.80]{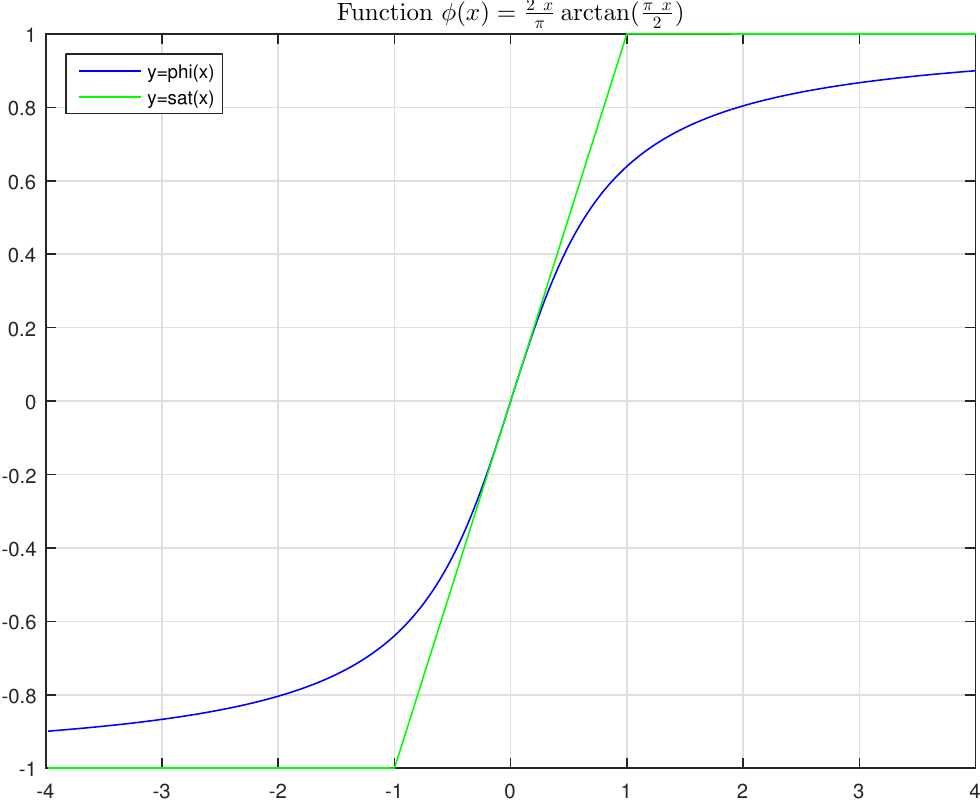}}
    \caption{Traditional saturation $\mbox{sat}(x)$ and smooth saturation $\phi(x)$.
    }
    \label{fSaturation}
\end{figure*}

Let ${ \widetilde u}_k(t)$ be given by \refeq{eFeedbackWindow}.
Let $u_k^{max}$ be the maximum admissible absolute value of the input.  The final saturated input to be applied to the system is given by
\begin{subequations}
\begin{equation}
\label{eSATu}
 u_k^{sat} =
 \left\{
\begin{array}{l}
{\overline u}_k  + u^* \phi (\frac{{\widetilde u}_k}{u^*}), \mbox{if} \;{\widetilde u}_k \geq 0,\\
{\overline u}_k  + u_* \phi (\frac{{\widetilde u}_k}{u_*}), \mbox{if}\; {\widetilde u}_k < 0,
\end{array}
\right.
\end{equation}
where
\begin{equation}
\begin{array}{l}
 u^* = u_k^{max} - {\overline u}_k,\\
 u_* = u_k^{max} +{\overline u}_k
\end{array}
\end{equation}
\end{subequations}
Figure \ref{Saturation_policy} illustrate the action of the saturation policy. At an instant $t$, the value of  ${\overline u}_k (t)$ is
marked as a point in the (vertical) interval $[ -u_k^{max},  u_k^{max}]$. If ${\widetilde u}_k \geq 0$ the distance of this point to the maximum accepted value
of the input is $u^* = u_k^{max} - {\overline u}_k$. Otherwise, if ${\widetilde u}_k \geq 0$ the distance of this point to the minimum value
$-u_k^{max}$ is $u_* = u_k^{max} + {\overline u}_k$. This explains why the saturation policy \refeq{eSATu} is defined in that way.
The function \refeq{eSATu} is smooth, and the sign of
feedback is preserved  as in the case of the traditional saturation function.
In particular it is easy to see that this saturation policy preserves the non-positiveness of $\dot {\mathcal V}$. However, this signal invariance holds only if,
in each step of RIGA, the value of ${\overline u}_k (t)$ is
 in the interval $[-u_k^{max},  u_k^{max}]$, that is, all the inputs, including the seed, must respect
 the saturation restrictions.
\begin{figure*}[t]
    \centering{\includegraphics[scale=0.80]{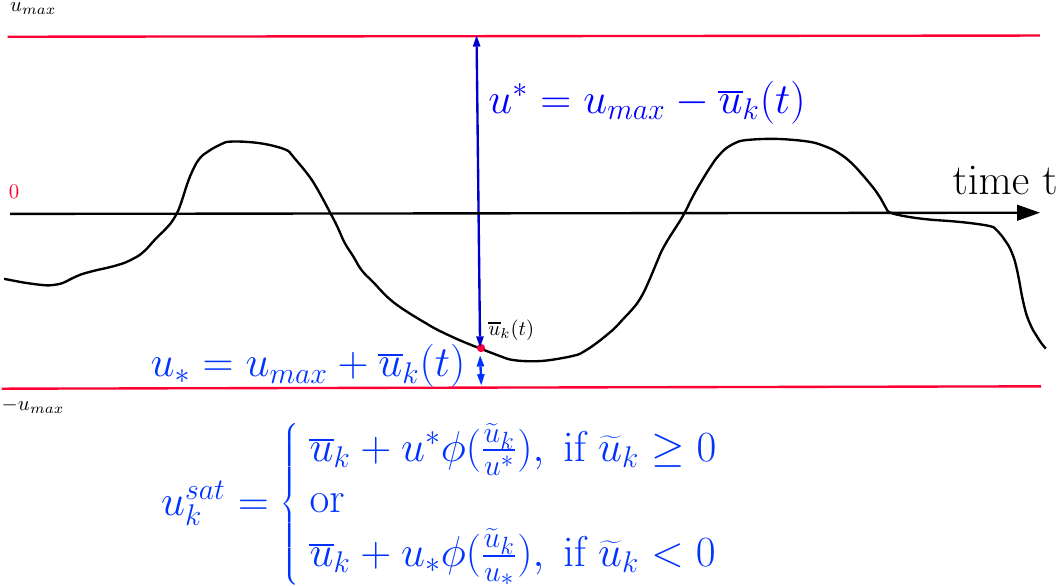}}
    \caption{Feedback saturation: one will always have $-u_{max} \leq u_k^{sat}(t)  \leq u_{max}$}.
     \label{Saturation_policy}
\end{figure*}

\subsection{Choice of the seed input ${\overline u}^0(t)$}
\label{sSeed}

Let $\mbox{{\Window}}(t)$ be a chosen window function as defined in section \ref{sWindow}.
Fix some $T>0$ and consider  an input of the form
\begin{subequations}
  \label{refcon}
 \begin{eqnarray}
  \label{refconab1}
  \overline{u}_k (t) = \mbox{{\Window}}(t) \overline{u}_{k_1} (t),
 \end{eqnarray}
 where the ${\overline u}_{k_1}(t)$ are given by
 \begin{equation}
 \label{refconab}
 \sum_{\ell = 1}^{M} \left[ a_{k \ell} \sin(2\ell \pi t / T) + b_{k \ell} \cos(2\ell \pi t / T) \right],
 \end{equation}
  \end{subequations}
 for $k=1, \ldots, m$,
 which are a sum of a finite number $M$ of harmonics of $\sin(2\pi t / T)$ and $\cos(2\pi t / T)$, and whose amplitudes are parameterized by a pair of a randomly chosen vectors $(\mathbf{a},\mathbf{b}) \in  \RR^{m M} \times \RR^{m M}$, where:
 \[
 \begin{array}{l}
  \mathbf a  =  (a_{11}, a_{12}, \ldots, a_{1M}, \ldots, a_{m1}, a_{m2}, \ldots, a_{mM})\\
   \mathbf b  =  (b_{11}, b_{12}, \ldots, b_{1M}, \ldots, b_{m1}, b_{m2}, \ldots, b_{mM})
   \end{array}
 \]
 In \cite{PerSilRou19} it is shown for the case which $n = \nbar$, that such choice assures the convergence of RIGA with probability one in a precise sense that can be found in that paper\footnote{Without much loss of generality,
 those proofs are done for the case where the window function is absent (or identically equal to one).}. This result
 is supported from previous ideas of \cite{SilPerRou14} and \cite{SilPerRou16} that are mainly based on Coron return's method \cite{Cor07}. These proofs  are generalized for the case where $\nbar < n$, for considering the partial trace Lyapunov function of the present paper (see Appendix \ref{aMathematical}. The examples of section \ref{sExamples} presents further discussions on the choice of
 $M$, $T$, and the vector $(\mathbf{a},\mathbf{b})$. If one wants to consider input saturation, it is clear that the seed input $\overline{u}_k (t)$ must be inside the interval  $[ -u_k^{max},  u_k^{max}]$
 for all $t \in [0, T_f]$. Executing RIGA without such a restriction would imply the positiveness of $\dot {\mathcal V}$.

\section{Examples}
\label{sExamples}

\subsection{Coupled Transmon-qubits}
  \label{ssec:ex1}

This system consists in two coupled transmon-qubits. The aim is to implement an  encoded
C-NOT gate, and a state preparation, as well. This system was also considered in  \cite{LeuAbdKocSch17} for testing an implementation of GRAPE.
It is considered to be a benchmark because of the nature of control problem itself.
The Hamiltonians of the system
 are given by:
\begin{eqnarray*}
H_0 & = & J (b_1 + b_1^\dag)(b_2+b_2^\dag)\\
    & + & \left[ \sum_{i=1}^2 \omega_j b_j^\dag b_j + \frac{1}{2} \alpha_j b_j^\dag b_j (b_j^\dag b_j -1)\right]\\
H_{u_1} & = & \beta (b_1+ b_1^\dag)\\
H_{u_2} & = & \beta (b_2+ b_2^\dag)\\
H_{u_3} & = & \beta b_2^\dag b_2
\end{eqnarray*}
In the simulations, the model is truncated. Only  $n_c=7$ levels are considered for both transmons. So $n= n_c^2$.
The parameters of the system $\frac{\omega_1}{2 \pi} = 3.5 GHz$, $\frac{\omega_2}{2 \pi} = 3.9 GHz$, $\frac{\alpha_j}{2 \pi} = -225 MHz, j=1,2$, $\frac{J}{2 \pi} = 100 MHz$, $\frac{\beta}{2 \pi} = 1 GHz$
This value of $\beta$ means that the control inputs $u_1$, $u_2$ and $u_3$ are given in $GHz$. A fourth  input  $u_4$  (global phase) is included, and it corresponds to the Hamiltonian $H_{u_4} = I_{n_c} \otimes I_{n_c}$.
 The aim is to encode a C-NOT gate in the first two levels of the transmon cavities. The control qubit is the first transmon qubit. In this case the C-NOT flips the second qubit when the first one is set to one.
For the C-NOT gate, the corresponding matrices are $E = [ |0~0\rangle, |0~1\rangle, |1~0\rangle, |1~1\rangle ]$ and  $F =[ |0~0\rangle, |0~1\rangle, |1~1\rangle, |1~0\rangle ]$.
As the quantum models of both transmons are truncated to the levels $0,1,2, \ldots, n_c-1$. The  level $n_c-1$ will be called \emph{forbidden} level.
The population of the $2 n_c - 1$ states  $\{|i~j\rangle~|~ i, j \in \{0, 1, \ldots, n_c-1\}, ~ \mbox{where} ~ i = n_c-1~\mbox{or}~ j = n_c-1\}$
is called \emph{forbidden population}. Let $\Pi_{forb}$  be the orthogonal projector onto the ``\emph{forbidden space}''.
In the simulations we have defined the bad population  $\mathcal B(t)$ as the maximum norm of the column vectors of $B(t) = \Pi_{forb} X(t) E$. As $F^\dag$ is the projector
into the encoded space, the ``good population'' $G(t)$
is defined as the smaller complex norm of an element of the diagonal of $G(t) = F^\dag X(t) E$.
The main RIGA parameters are $T_f=10 ns$, $K = \frac{1}{\omega_1}$, $u_{max}=0.5 GHz$. The desired final infidelity
is $0.001$. We have chosen a number of simulation points $N_{sim} = 4000$ (see section \ref{sImplementation}).
The parameters of the seed of RIGA are $M=3$, $T=19 M \frac{2 \pi}{\omega_1}$, $A_m = \frac{0.2}{M}$. No window function
was used in simulation.
The  entries of the vectors $\mathbf{a}$, $\mathbf{b}$ defining the amplitudes of the harmonics of the seed are available in \cite{CODE_OCEAN_SMOOTH} (file \verb"ab.mat"
in the data directory). The figures \ref{Control_CNOT}, \ref{Spectra_CNOT}, \ref{Population_CNOT} illustrates the obtained results for the C-NOT
gate generation. After the computation of RIGA, the re-simulation of the system with $n_c = 10$ levels for both transmon cavities presented the infidelity 1.0004e-03 (is not changed {up} to a precision of 1e-6) , showing that the original truncation of the system is a good approximation.

\begin{figure*}[t]
    \centering{\includegraphics[scale=0.7]{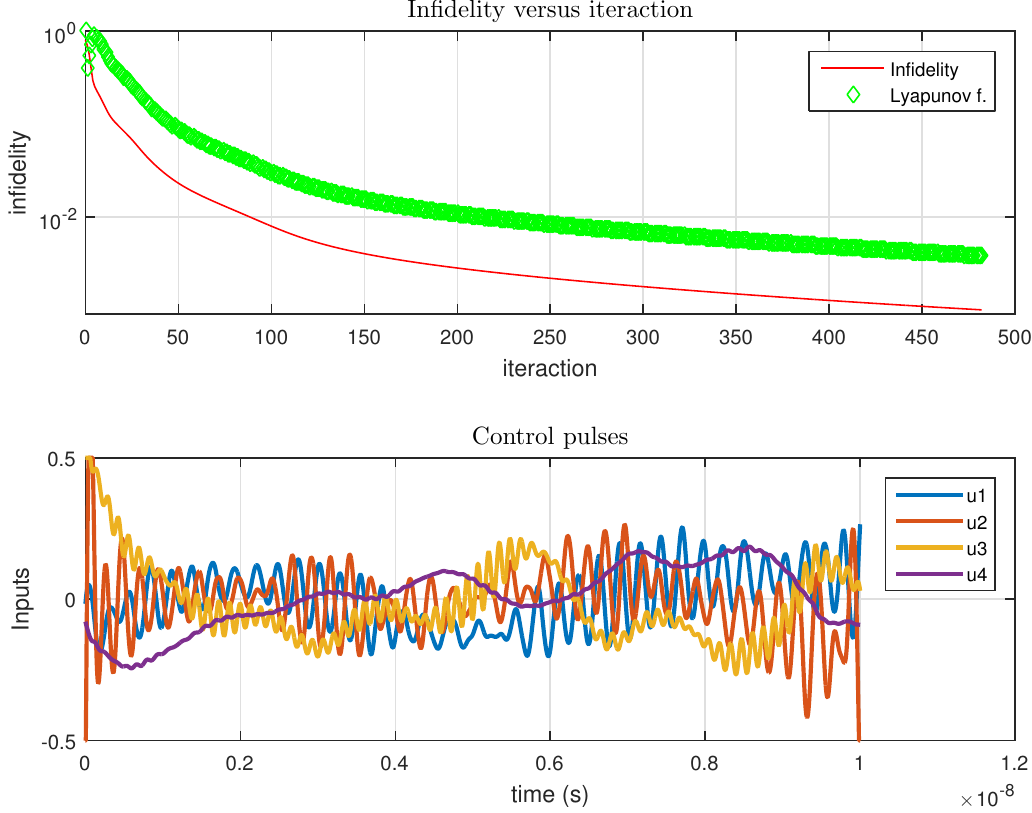}}
    \caption{
    \textbf{Results of RIGA for the Hadamard gate for system of sub-section~\ref{ssec:ex1}.
    Top: evolution of infidelity along the steps of RIGA.
    Bottom:  control pulses generated by  RIGA. }.}
     \label{Control_CNOT}
\end{figure*}

\begin{figure*}[t]
    \centering{\includegraphics[scale=0.7]{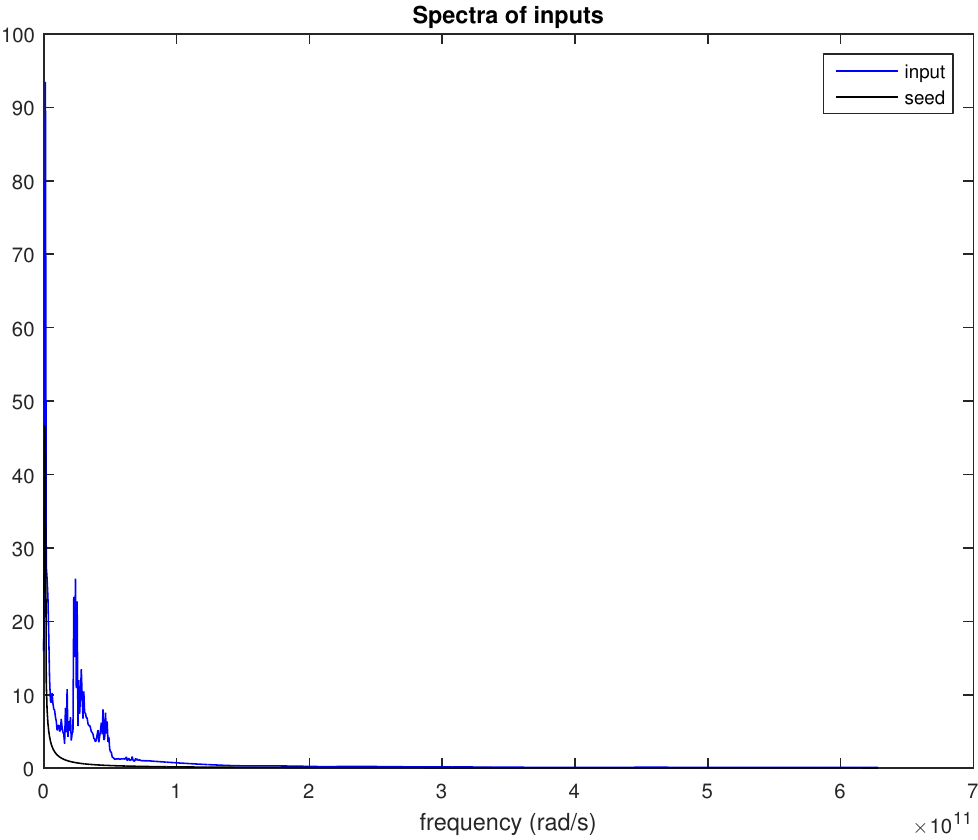}}
    \caption{
    \textbf{Spectra of the control signals and of the seed input for the Hadamard gate for the system of sub-section~\ref{ssec:ex1}.}}
     \label{Spectra_CNOT}
\end{figure*}

\begin{figure*}[t]
    \centering{\includegraphics[scale=0.7]{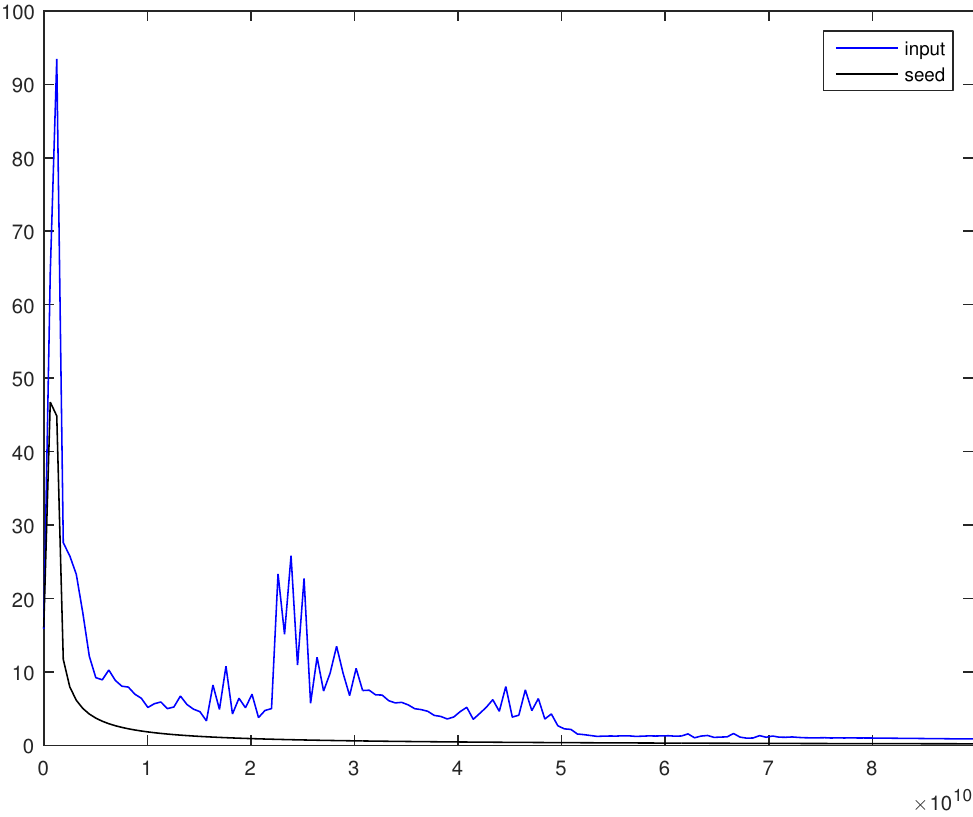}}
    \caption{
    \textbf{Results of RIGA for the Hadamard gate for system of sub-section~\ref{ssec:ex1}. 
    Zoom of the spectra of Figure \ref{Spectra_CNOT}}.}
     \label{Spectra_zoomCNOT}
\end{figure*}

\begin{figure*}[t]
    \centering{\includegraphics[scale=0.7]{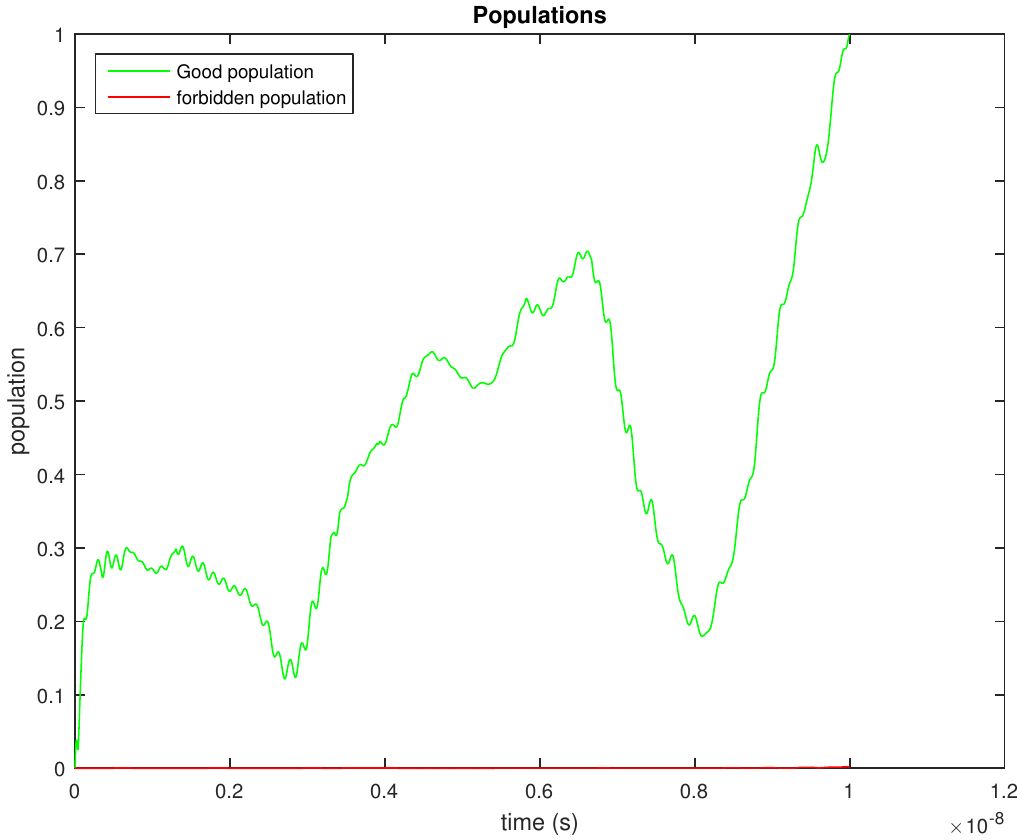}}
    \caption{
    \textbf{Populations for the Hadamard gate for system of sub-section~\ref{ssec:ex1}} corresponding to the re-simulation of the system with $n_c = 10$ levels for both transmon systems.}
     \label{Population_CNOT}
\end{figure*}

\begin{figure*}[t]
    \centering{\includegraphics[scale=0.70]{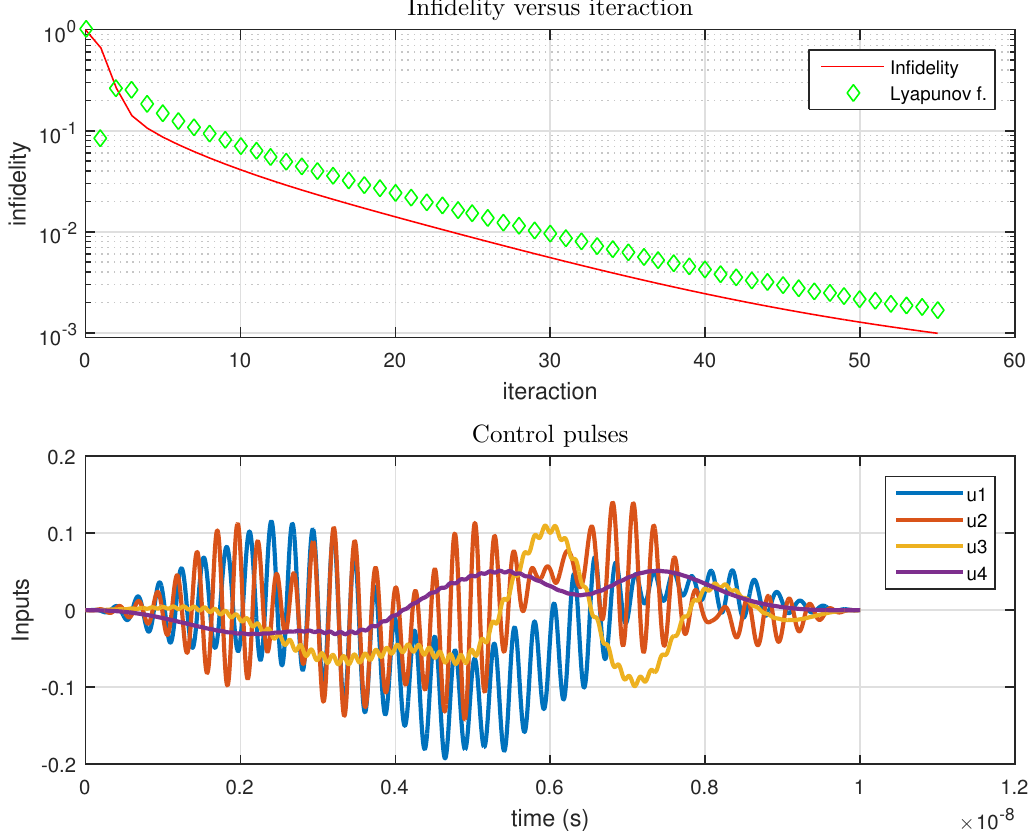}}
    \caption{
    \textbf{Results of RIGA for the state preparation for system of sub-section~\ref{ssec:ex1}.
    Top: evolution of infidelity along the steps of RIGA.
    Botton:  control pulses generated in the last step of RIGA. }.}
     \label{Control_State}
\end{figure*}

The implementation of RIGA in \cite{CODE_OCEAN_SMOOTH} presents the generated control pulses
and their spectra. This helps the selection of a good gain $K$, because a high gain tends to generate artificial high frequencies
in the control pulses. Furthermore, the spectra of the seed and of the feedback law is also shown. A good tip for choosing the parameters
$M$ and $T$ of the seed is to have in mind that it is not desirable to have artificial high frequencies in the seed, and at the same time, the seed must
``excite sufficiently'' the closed loop system in order to assure the convergence of RIGA \cite{PerSilRou19}. So studying the spectra of the generated feedback and comparing with the spectra
of the seed could give good insight to the choice of these parameters. After that, the implementation \cite{CODE_OCEAN_SMOOTH} produces also a simulation of the system with half step ($\delta=T_f/N_{sim}$, and the
half step is $\delta/2$), showing the infidelity between the final condition of the solutions with half step $\delta/2$ and the solution with step $\delta$. High gains $K$ tends to need
smaller steps (or larger $N_{sim}$) to assure a good precision. In the end of RIGA, the information of the ``forbidden population'' is a good measure for
estimating if the number of levels $n_c$ of the models were sufficiently large for assuring a good precision of the numerical simulation. In any case, the re-simulation  of the open-loop system
with a model with more levels will give a final answer to this question. High gains tends to excite the high levels of the cavities, and the simulation of more levels may be needed to assure a good precision.

For the state preparation, we have considered the same parameters, but now $E = [|0~0\rangle]$ and  $F=[\frac{\sqrt{2}}{2} \left( |1~0\rangle + |0~1\rangle \right)]$. Note that we have used
the window function \eqref{eHamming}.  The figure \ref{Control_State} illustrates the obtained results for such a  state preparation.

After the computation of RIGA, the re-simulation of the system with $n_c = 10$ levels for both transmon cavities presented the infidelity 9.9036e-04 (the difference between the infidelity that was obtained with $n_c=7$ indicates a precision of 1e-6) , showing that the original truncation of the system is indeed a good approximation.

It is important to observe in Figure \ref{Spectra_zoomCNOT} that the spectra of the the generated inputs includes the frequencies of the first three levels  of both trasmon systems (counting the level zero). When the control gains are well tuned, RIGA is able to generate the control pulses without generating spurious frequencies without the need of any penalty function\footnote{As a mather of fact, in the implementation of GRAPE described in \cite{LeuAbdKocSch17}, some penalty functions on the control pulses must be included in order to assure that the energy and the frequency band of the generated control pulses is adequate.}


\subsection{Cavity coupled to a transmon qubit}
This second system considered in~\cite{CavityTransmonGrape} consists in 
a cavity-mode dispersively coupled to a transmon qubit.  The aim  is to implement an encoded Hadamard gate for the transmon-qubit.
In~\cite{CavityTransmonGrape}  the authors have generated the control pulses using GRAPE. In this work
the control pulses are  smooth and generated using an implementation of RIGA
available at \cite{CODE_OCEAN_SMOOTH}, which implements the theory of the present article.

The Hamiltonian of the model is of the form:
\[
H_{\mbox{\tiny cavity}} + H_{\mbox{\tiny transmon}} + H_{\mbox{\tiny int}} + H_{\mbox{\tiny drive}}
\]
where:
\begin{eqnarray*}
\frac{H_{\mbox{\tiny cavity}}}{\hbar} & = & \omega_c a^\dag a+ \frac{\kappa}{2} (a^\dag)^2 a^2\\
\frac{H_{\mbox{\tiny transmon}}}{\hbar} & = & \omega_T b^\dag b + \frac{\alpha}{2} (b^\dag)^2 b^2\\
\frac{H_{\mbox{\tiny int}}}{\hbar} & = & \chi (a^\dag a) (b^\dag b) + \frac{\Chi^\prime}{2} [(a^\dag)^2 a^2] [(b^\dag)^2 b^2]\\
\frac{H_{\mbox{\tiny drive}}}{\hbar} & = & u_1 (a + a^\dag) + \frac{u_2}{ \jmath} (a - a^\dag) \\
& + &  u_3 (b + b^\dag) + \frac{u_4} {\jmath} (b - b^\dag) \\
& + & u_5 I_{n_c} \otimes  I_{n_T}
\end{eqnarray*}
The controls $u_1, u_2$ (resp., $u_3, u_4$ are local controls of the cavity (resp., of the transmon). The control $u_5$ is an artificial global phase control, added
because our approach is in $\Un$ (and not in $\mbox{SU}(n)$). The parameters of the system are the following:  cavity frequency $\frac{\omega_c}{2 \pi} = 4452.6$ MHz, transmon frequency  $\frac{\omega_T}{2 \pi} = 5664.0$ MHz,
dispersive shift  $\frac{\chi}{2 \pi} = 2194.0$ kHz, transmon anharmonicity $\frac{\alpha}{2 \pi} = -236$ MHz, Kerr effect  $\frac{\kappa}{2 \pi} = -3.7$ kHz, 2\emph{nd} order dispersive shift  $\frac{\chi^\prime}{2 \pi} = 19.0$ kHz. As usual, a rotating frame was adopted in order to eliminate the terms  $\frac{H_{\mbox{cavity}}}{\hbar}$ and $\frac{H_{\mbox{transmon}}}{\hbar}$ of the whole Hamiltonian. In this way, the control pulses that are presented here are in fact
modulating respectively the amplitudes of the cavity frequency $\frac{\omega_c}{2 \pi}$ for $u_1$ and $u_2$  (resp. the transmon frequency $\frac{\omega_T}{2 \pi}$ for $u_3$ and $u_4$).

Standard notations of the states $|k\rangle, k \in \NN$ of the cavity
and for the transmon-qubit are considered here. One shall define the states $| \pm Z_L \rangle$ of the cavity by (see \cite{CavityTransmonGrape}):
\begin{eqnarray*}
|+Z_L \rangle & = & \sum_{n} \frac{\alpha^{4 n}} { \sqrt{4 n}} | 4 n \rangle,\\
|-Z_L \rangle & = & \sum_{n} \frac{\alpha^{4 n+2}} { \sqrt{4 n+2}} | 4 n +2 \rangle
\end{eqnarray*}
The state $|0\rangle$ of the transmon will denoted by $g$ and the state $|1\rangle$ will be denoted by $e$.
In the simulations, the states  $| \pm Z_L \rangle$ are truncated up to $n_c$ levels and then renormalized.

In order to define the encoded gate to be implemented, define the vectors $e_1 = |0 \rangle \otimes g$, $e_2 = |0 \rangle \otimes e$, $f_1= |+Z_L \rangle \otimes g$ and $f_2 = |-Z_L \rangle \otimes g$. In this example,
an encoded
 \emph{Hadammard Gate} is defined by $e_i \mapsto h_{1i} f_1 +  h_{2i} f_2, i=1,2$
 where $h_{ij}= \{H\}_{ij}$ is the Hadamard  matrix:
 \[
 H = \frac{1}{\sqrt{2}} \left[
 \begin{array}{cc}
  1 & 1\\
  1 & -1
 \end{array}
  \right]
 \]
The main RIGA parameters are $N_{sim} =500$, $T_f =1.1 \mu s$ (which is approximately $2.4 \left[ \frac{2 \pi}{\chi} \right]$),
$K=$ 0.5e-6, $u_{max} = 5$, using a Hamming-like window function given by  \refeq{eHamming}. The parameters
of the seed are $M=3$, $T=M~T_f/(2 \pi)$, $A_m = K/2$. The vector of coefficients $\mathbf{a}$  and $\mathbf{b}$ are
the same of the file \verb"ab.mat" of the directory data of \cite{CODE_OCEAN_SMOOTH}.
The desired infidelity is $0.001$.
 The cavity model is truncated to $n_c = 20$ levels, and the transmon model is truncated to $n_T = 4$ levels.
The figures \ref{Control_Hadamard}, \ref{Spectra_Hadamard} illustrate the obtained results. A  re-simulation of the system with higher truncation up to   $n_c = 25$ and $n_T = 6$ for the cavity/transmon
provide an infidelity 9.8577e-04, showing that the original truncation is  adapted. Figure  \ref{Population_Hadamard} show the populations for the resimulation, showing that the "forbidden population" is almost absent, indeed\footnote{The forbidden population is the  population of the levels $25$ and $6$ respectively for the cavity and the transmon.}. 

\begin{figure*}[t]
    \centering{\includegraphics[scale=0.70]{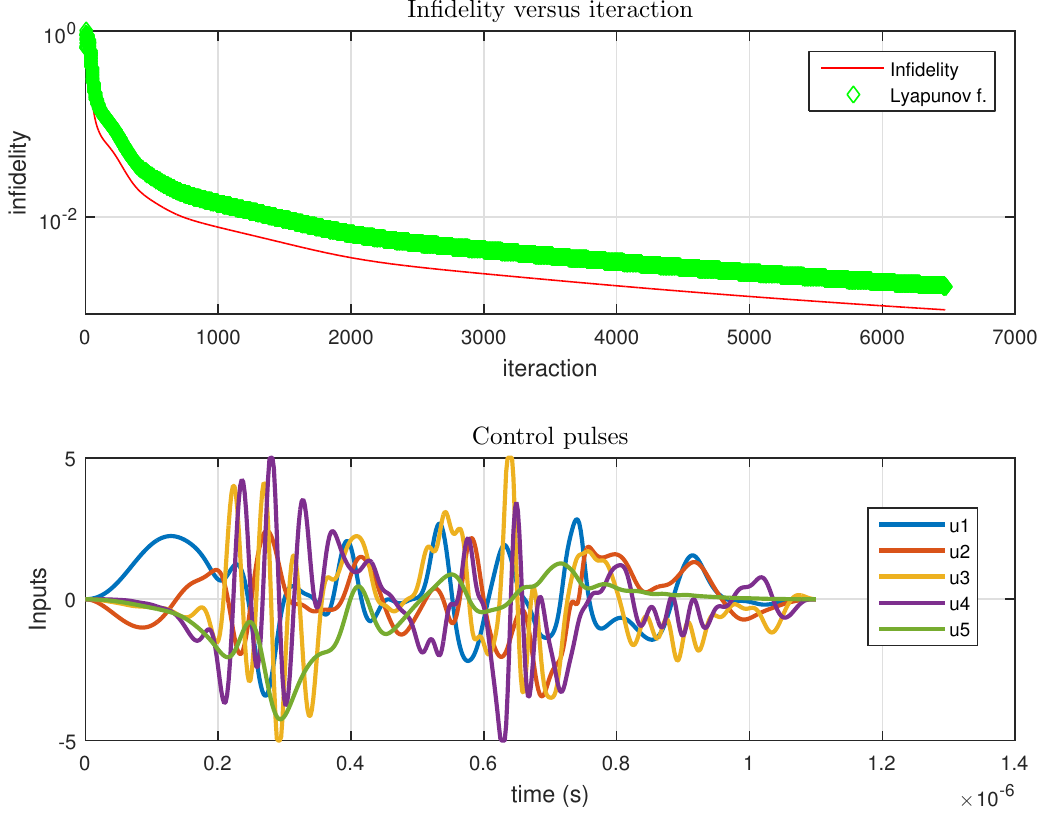}}
    \caption{
    \textbf{Results of RIGA for the Hadamard gate for the second example.
    Top: evolution of infidelity along the steps of RIGA.
    Botton:  control pulses generated in the last step of RIGA. }.}
     \label{Control_Hadamard}
\end{figure*}

\begin{figure*}[t]
    \centering{\includegraphics[scale=0.70]{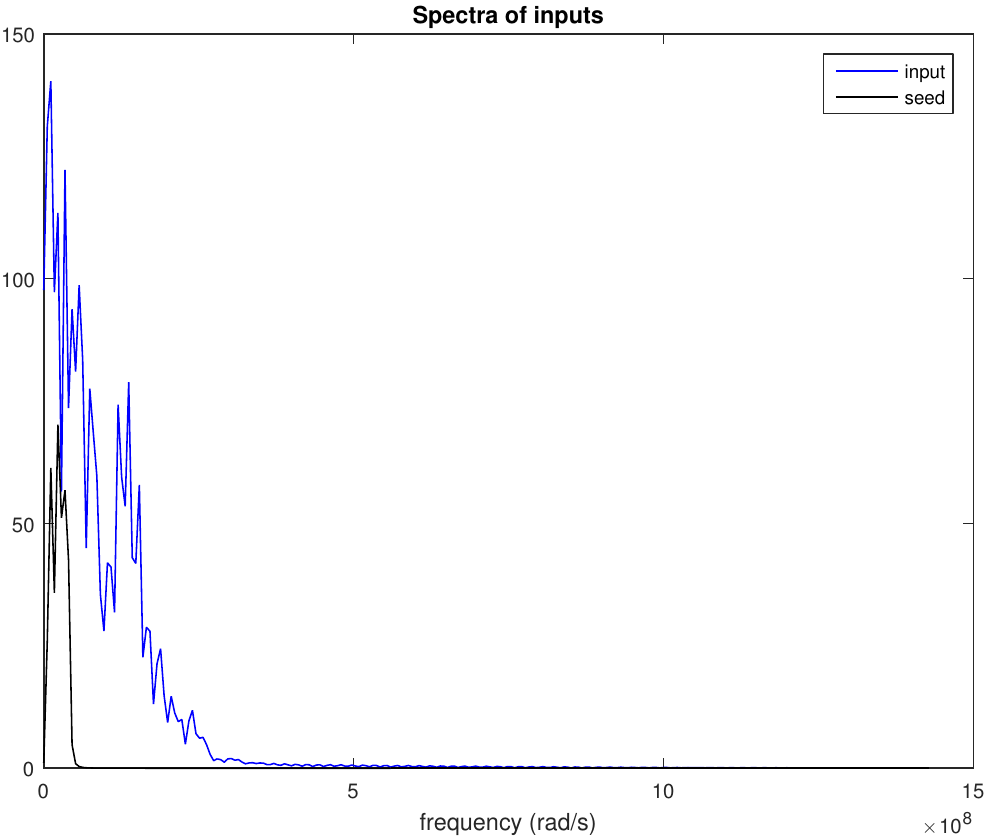}}
    \caption{
    \textbf{Spectra of control inputs for the Hadamard gate for the second example.}.}
     \label{Spectra_Hadamard}
\end{figure*}

\begin{figure*}[t]
    \centering{\includegraphics[scale=0.70]{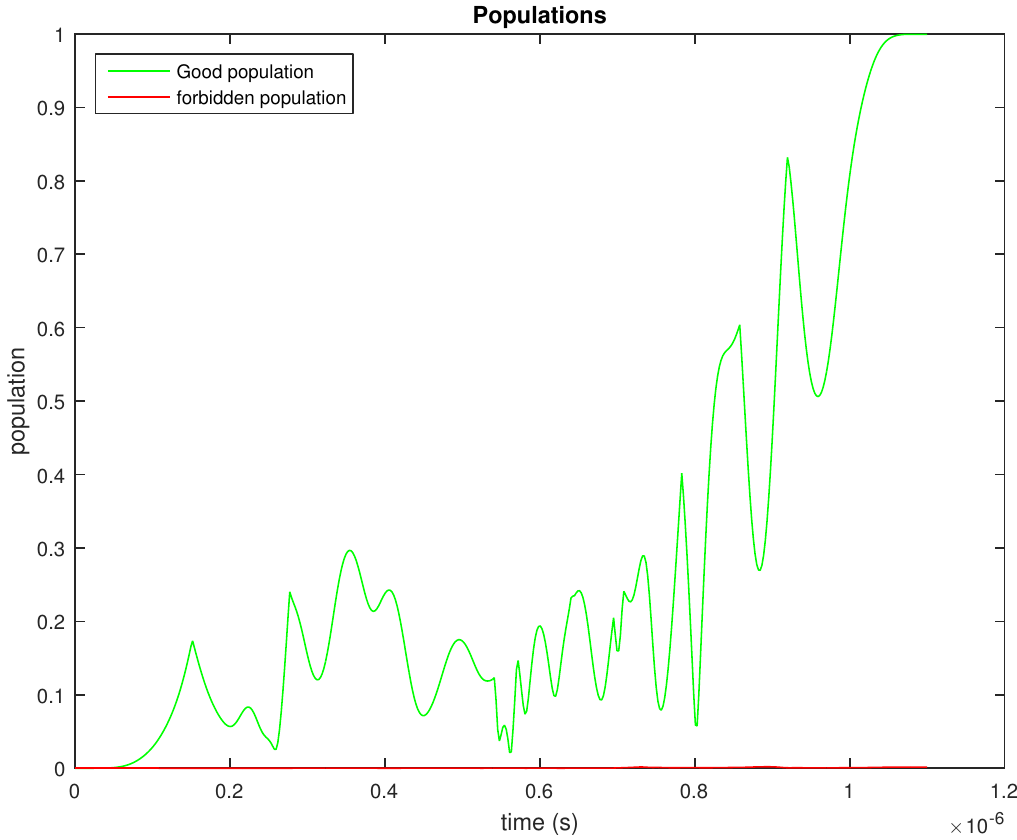}}
    \caption{
    \textbf{Populations for the Hadamard gate for the resimulation of the second example with a truncation of $25$ and $6$ levels respectively of the cavity and the transmon.}.}
     \label{Population_Hadamard}
\end{figure*}

\subsection{$N$-coupled qubits}
\label{sN_coupled}

\begin{figure*}[t]
    \centering{\includegraphics[scale=0.70]{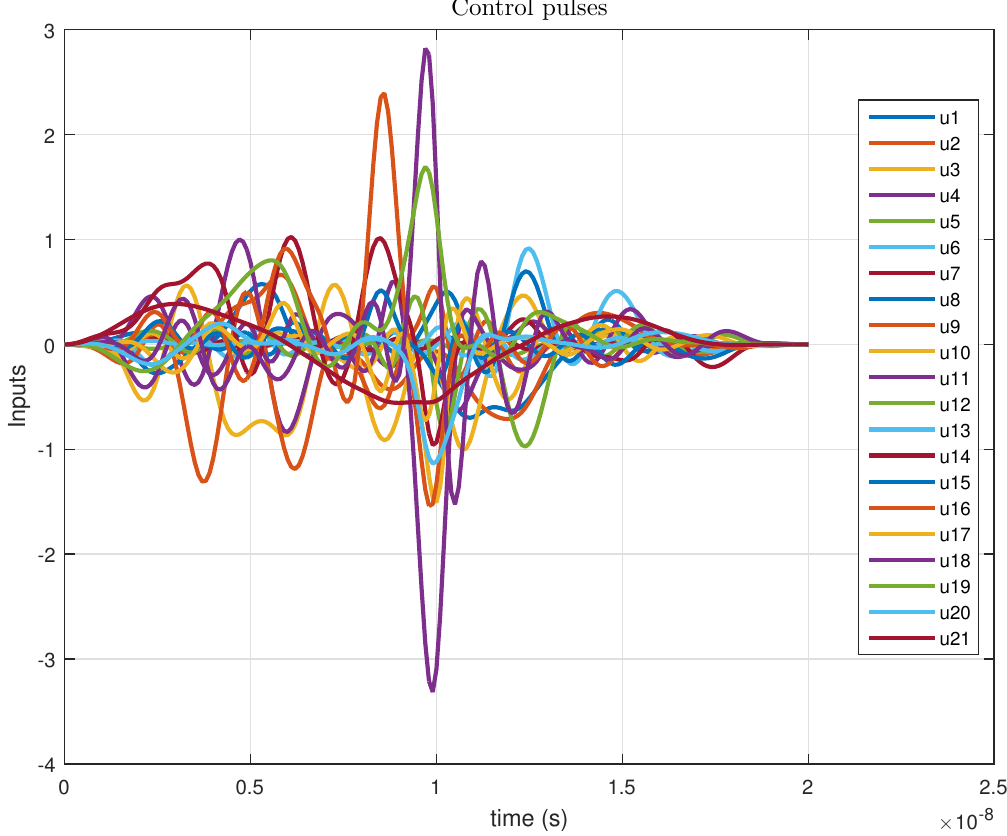}}
    \caption{Third Example -- Numerical experiment for $N=10$ qubits: {Generated control inputs (with the window option)}.}
     \label{fControl}
\end{figure*}

\begin{figure*}[t]
    \centering{\includegraphics[scale=0.70]{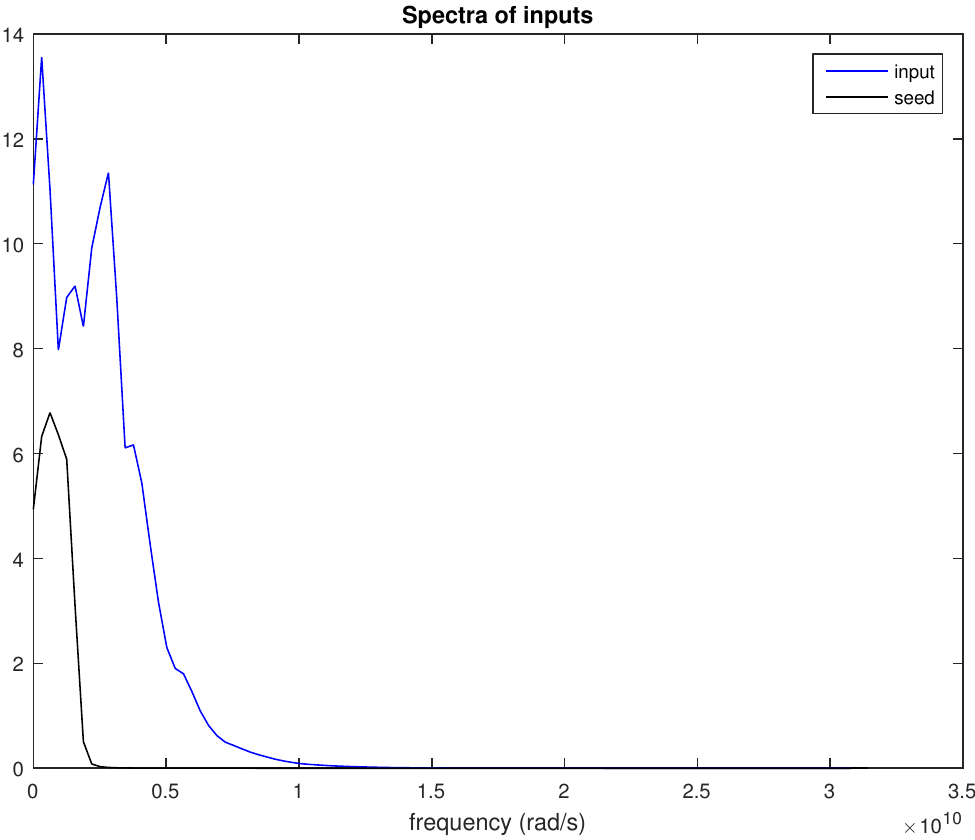}}
    \caption{Third Example -- Numerical experiment for $N=10$ qubits: \textbf{Average (in $k$) of the spectra of the seed ${\overline u}^0_k(t), k=1, \ldots , m$ of RIGA and of the generated control inputs}.}
     \label{fSpectra}
\end{figure*}

This example is  a benchmark proposed  in  \cite[section D]{LeuAbdKocSch17} for testing an
implementation of GRAPE for a large dimensional Hilbert space. One considers here  such a composite system made  of  $N$ qubits where $N$ varies from 3 to 10  with   Hilbert space dimensions from $2^3=8$ to  $2^{10}=1024$. The Hamiltonian  reads
\begin{eqnarray*}
 H(t) & = & J_0 \left[\sum_{s=1}^{N-1} \sigma_z^{(s)} \sigma_z^{(s+1)} \right] \\
      & + & J \left[\sum_{k=1}^{N} u_{x_k}(t) \sigma^{(k)}_x +  u_{y_k}(t) \sigma^{(k)}_y \right] + J_g u_g(t) I
\end{eqnarray*}
where $\sigma_x, \sigma_y, \sigma_z$ are the usual  Pauli matrices.
The artificial input $u_g(t)$ is only a global phase input. The goal is to generate the "large Hadamard gate"  formed by  the tensor product of  $N$  Hadamard qubit-gates. In the present example, $\nbar$ is equal to $n$. The dimension of the Hilbert space is $n= 2^N$ and the number of control inputs is $m = 2 n +1$.
The system parameters are $J_0 = J = J_g = 2 \pi 100 MHz$.
The main RIGA parameters are   $T_f = (2 N)~ns$, $N_{sim}=20 N$, $u_{max} = 5$.
The feedback gains are given by $K = \frac{1}{J} \mathcal K(N)$, where $\mathcal K(N)$ is the $N$-th entry of the vector
$[10, 10, 10, 2, 2, 1, 0.5, 0.25, 0.125, 0.0625]$. The values of gains $K$ were chosen with the aid of an option of the piecewise-constant  implementation of RIGA
\cite{CODE_OCEAN_CONSTANT}. The  values of the entries of $\mathbf a , \mathbf b$  are available in the file \verb"ab.mat" of the directory data
of \cite{CODE_OCEAN_SMOOTH}. The value of  $M$  is the $N$-th entry of the vector $[10~10~11~10~14~14~14~14~14~14~]$, $T = \pi T_f$ and $A_m =2/M$.
The desired infidelity is $0.001$.
The numerical experiments with this system are shown in the next table:

\begin{tabular}{|c|c|c|c|}
  \hline
  N         & Infidelity &  steps  &  Runtime\\
  (qubits)   &           &              &   $(s)$ \\
  \hline
  $3.0 $&$ 9.99\cdot 10^{-4} $&$ 55.0 $&$ 2.59  $\\ $
  4.0 $&$ 9.87\cdot 10^{-4} $&$ 77.0  $&$  8.17 $ \\ $
  5.0 $&$ 9.95\cdot 10^{-4} $&$ 477.0 $&$ 135.0 $\\ $
  6.0 $&$ 9.87\cdot 10^{-4} $&$ 84.0 $&$  90.8 $ \\ $
  7.0 $&$ 9.73\cdot 10^{-4} $&$ 130.0 $&$ 632.0 $ \\ $
  8.0 $&$ 9.98\cdot 10^{-4} $&$ 148.0 $&$ 5166.0 $ \\ $
  9.0 $&$ 9.97\cdot 10^{-4} $&$ 362.0 $&$ 1.24\cdot 10^5 $ \\ $
  10.0 $&$ 9.95\cdot 10^{-4} $&$ 229.0 $&$ 5.08\cdot 10^5 $\\
 \hline
\end{tabular}

The results for $N=10$ qubits are shown in figures \ref{fControl} and \ref{fSpectra}. It must be pointed out that, for $N=10$ qubits, the propagator $X(t)$ is a $1024\times 1024$ matrix. This shows that RIGA is able to tackle high-dimensional systems.

\section{RIGA refinements and some mathematical details}
\label{s:Refinements}

The description of RIGA of section \ref{sRigaDescription} is very compact, but it cannot include the scenario that is faced by the more complete version of RIGA. In fact, as the control is Lyapunov-based, and the Lyapunov function contains singular and/or critical points, the goal  matrix $X_{goal}$ must be chosen step-wise in order to avoid such undesirable points. A fundamental property of RIGA to be shown is the fact the Lyapunov function of the final error $\widetilde {\overline X}(T_f)^\dag X(T_f) = X_{goal}^\dag X(T_f)$ is always a non-increasing sequence along the steps of RIGA.  This nice property does not hold any more during the steps for which the application of a strategy for avoiding singular and/or critical points of the Lyapunov function are necessary.

The implemented version of RIGA is given by the following algorithm\footnote{
This description of RIGA is a little bit different from the one of \cite{PerSilRou19}, but it is easy to show their equivalence from the right-invariance of the quantum system.}:
\begin{center}
 \textbf{RIGA (ALGORITHM 1)}
\begin{tabbing}
 \=123 \=123 \=123 \=123 \kill\\
 $\sharp 1.$ Choose the seed input $\overline u^0: [0, T_f] \rightarrow \RR^m$.\\
 Execute the steps $\ell=1, 2, 3, \ldots$\\
  \textbf{BEGIN STEP $\ell$}.\\
 $\sharp 2.$ \> \>  Set ${ u}(\cdot) = {\overline u}^{\ell-1}(\cdot)$. Set $  X(0) = I$.\\
             \> \>  Integrate numerically the  open loop system \\
             \>  \>   \refeq{cqs}. Obtain $X^{\ell-1}(t)  =X(t)$, for $t \in [0, T_f]$, \\
             \> \> and ${X}_f^{\ell-1} = { X}^{\ell-1}(T_f)$.\\
$\sharp 3.$ \> \>  If the final infidelity $\mathcal I(X_f^{\ell-1})$  is acceptable, \\
            \> \> then terminate RIGA. Otherwise, continue.\\
$\sharp 4.$ \> \>  Construct an adequate $\Xgoalbar$.\\
             \> \> Compute $R_\ell = \left( X_f^{\ell-1} \right)^\dag \Xgoalbar$.\\
 $\sharp 5.$  \> \> Define ${\overline X}^{\ell-1} (t) = \overline X(t) = X^{\ell-1}(t) R_\ell$ \\
            \> \>  (new reference defined by right translation)\\
            \> \> Define ${\overline  u}(\cdot) = {\overline u}^{\ell-1}(\cdot)$.\\
$\sharp 6.$ \> \>   Integrate numerically the closed loop \\
            \>  \>  system \refeq{eClosedLoopComplete2} from $X(0) = I$.\\
             \>  \>  Obtain $X^{\ell}(t) = X(t)$, for $t \in [0, T_f]$. \\
            \>  \>   Set ${\overline u}^{\ell}(\cdot) = {\widetilde u}(\cdot) +  {\overline u}(\cdot)$\\
 \textbf{END STEP $\ell$}\\
\end{tabbing}
\end{center}

\begin{remark}
\label{rRIGA} For further reference, the specific operations  of RIGA are numbered from $\sharp 1$ to $\sharp 6$.
 This will avoid confusion with the steps $\ell =1, 2, \ldots$ of RIGA.
From a pure mathematical point of view, operation $\sharp 6$, that regards the integration of system \refeq{eClosedLoopComplete2}
with initial condition ${\widetilde X}_0 = \overline X(0)^\dag X(0) = R^\dag_\ell$,
is equivalent
to the integration of \refeq{eClosedLoopComplete}
in closed loop with the input \refeq{eUtilk} with initial condition $X_0 = I$. However, from the numerical point of view,
different computation errors may be obtained (see  section \ref{sImplementation} for the details about the algorithm implementation).
 The choice of $\Xgoalbar$ in each step $\ell$
is done in order to avoid the critical and or singular  points of the Lyapunov function, if they do exist. The process of avoiding singular and critical
points (see Section \ref{sFirstStrategy}) may produce $\Xgoalbar$ such that $\Xgoalbar E $ is not equal to $\exp(\jmath \phi) F$.  Proposition \ref{pMonotonous} shows
that, if $\Xgoalbar E = \Xgoalbarellminusone E$, then the value of Lyapunov function $\mathcal V(\widetilde X(T_f))$ in the end  of step $\ell-1$ is coincident with
the  value  $\mathcal V(\widetilde X(0))$ of the beginning of step $\ell$. This is a key property of RIGA, as  discussed in Proposition  \ref{pMonotonous} in the Appendix.
From part (a) of Prop. \ref{pMonotonous}, the operations $\sharp 2$, $\sharp 4$, and $\sharp 5$ are shown to be equivalent to integrating system \refeq{reference} backwards from $\overline X(T_f) = \Xgoalbar$. Nevertheless the determination of $\Xgoalbar$ by a certain optimization process (to be discussed in section \ref{sFirstStrategy}) depends on $X_f^{\ell-1}$, which is computed by a forward integration of $\overline X  (t)$ with $X(0) = I$. Finally, in theory, the computation of $\sharp 2$ for step $\ell+1$ could be obtained from $\sharp 6$ in the previous step $\ell$. However, from a numerical point of view, the result is not the same as explained in Remark \ref{rSharps} of Section \ref{sSmooth}. The reader will find an explanation in that remark of the apparently unnecessary repetition of the computations  of $\sharp 2$ and $\sharp 6$.
 \end{remark}

The following notations and definitions are considered in the sequel:
\begin{definition}
\label{dFund}
Consider the following notations in the context of RIGA:\\
(a) $X^{\ell-1}(t)$  is the trajectory that is obtained in operation $\sharp 2$ of step $\ell$ of RIGA (which coincides
       \footnote{The algorithm is executed in this way for the sake of numerical issues (see the section \ref{sImplementation} regarding
the numerical implementation of RIGA).} with the trajectory $X(t)$ that is obtained in operation $\sharp 6$ of step $\ell-1$).
Note that $X_f^{\ell-1} = X^{\ell-1}(T_f)$.\\
(b) $X^{\ell}(t)$ is the trajectory (with $X(0) = I$)  that is obtained in operation $\sharp 6$ of step $\ell$  of RIGA.\\
(c) The goal matrix that is constructed in operation $\sharp 4$ of step $\ell$ of RIGA is denoted by  $\Xgoalbar$.\\
(d) The reference trajectory ${\overline X}^{\ell-1}(t)$ of step $\ell$ is the one that is obtained
in operation $\sharp 5$ of step $\ell$ of RIGA, and is such that\footnote{This is shown in part (a) of Prop. \ref{pMonotonous}.}  ${\overline X}^{\ell-1}(T_f) = \Xgoalbar$
and ${\overline X}^{\ell-1}(t) = X^{\ell-1}(t) R_\ell$.\\
(e) $\widetilde X^\ell(t) = \left( {\overline X}^{\ell-1}  (t) \right)^\dag X^{\ell}(t)$.\\
(f) $R_\ell = \left( X_f^{\ell-1} \right)^\dag \Xgoalbar$.\\
(g) ${\overline u}_k(t)$ is the reference input constructed  in operation $\sharp 6$ of step $\ell$ of RIGA.
\end{definition}

\subsection{The critical points of the \emph{partial trace} Lyapunov function}

The partial trace Lyapunov function is bounded with a bounded gradient, and so the resulting feedback law \refeq{eFeedbackLaw} is always bounded.
This is not the case  for the unbounded Lyapunov function that is used in \cite{PerSilRou19} in the context that considers that $n$ coincides with $\bar n$.
That Lyapunov function presents a set $\mathcal W = \{ W \in \Un ~|~\det(W+I)=0\}$ of singular points.
So an strategy to avoid the set $\mathcal W$ which is used in that paper is essentially to avoid eigenvalues at $-1$.
On the one hand, the partial trace Lyapunov function
does not admit singular points, but on the other hand, it admits a set $\mathcal G$ of nontrivial critical points\footnote{Recall that a critical point $\widetilde X \in \Un$ of
$\mathcal V$ is a point such that $\nabla_{\widetilde X} \mathcal V \cdot D = 0$ for all $D \in \un$.} that must be avoided by RIGA.
Given a  $n \times n$  complex matrix ${\widetilde W}$, it can be decomposed in four submatrices $ {\widetilde W}_{11},  {\widetilde W}_{12},  {\widetilde W}_{21},  {\widetilde W}_{22}$
where ${\widetilde W}_{11}$ is $\bar n$-square and
  ${{\widetilde W}} = \left[
 \begin{array}{cc}
 {\widetilde W}_{11} & {\widetilde W}_{12}\\
 {\widetilde W}_{21} & {\widetilde W}_{22}
 \end{array}
 \right]$.
 If $\widetilde W \in \Un$ it is straightforward to show that ${\widetilde W}_{21}$ is the null matrix if and only if  ${\widetilde W}_{12}$ is also null.
 The next proposition justifies why this Lyapunov function is called\emph{ partial trace } and why $\sqrt {\mathcal V( \widetilde X)} = \| (X - \overline X) E \|$ is
 is called\emph{ partial distance} between $X$ and $\overline X$.
 \begin{proposition}
 \label{pCritical}
 Let $X, \overline X \in \Un$ such that $\widetilde X = {\overline X}^\dag X$. Let  $X_E = [ E ~\; \widehat E]$ be a unitary  matrix, where the matrix $E$ is defined in Prob. \ref{Prob1} in the introduction of this work..
 Let $ W = X_E^\dag X  X_E$, $\overline W = X_E^\dag \overline X  X_E$ and $\widetilde W = {\overline W}^\dag W$.
 Then\\
(a) $\mathcal V (\widetilde X) = \mathcal V ({\widetilde X}^\dag) =$ $2 \nbar - 2 \Re \left[ \trace \left( {\widetilde W}_{11}  \right)  \right] =$ $\| (\widetilde X - I) E\|^2$.\\
(b)  $\| \overline X E - X E\|^2 = \mathcal V(\widetilde X)$.\\
(c) $\widetilde X$ is a critical point of the partial trace Lyapunov function if and only ${\widetilde W}_{21}$ is a null matrix
  and ${\widetilde W}_{11} =  \bar V^\dag D \bar V$ with $V \in \Unbar$, where $D$ is a diagonal matrix of the form
$D = \mbox{diag} [\pm 1, \pm 1, \ldots, \pm 1]$.\\
(d) Let   $\Pi : \CC^n \rightarrow \CC^{\nbar}$ be the canonical projection, represented by the matrix $\Pi = [ I_{\nbar} \; 0 ]$. Then
  $\mathcal V (\widetilde X) = 2 \nbar - 2 \Re \left[ \trace \left( \Pi \widetilde W \Pi^\dag \right) \right] = \| (\widetilde W - I) \Pi^\dag\|^2$.
 \end{proposition}
\begin{proof} See appendix \ref{aCritical}.
 \end{proof}

  \begin{definition}
 \label{dPartial}
 Given $X, \overline X \in \Un$, the partial distance is defined by $\pdist(X, \overline X) = \| (X - \overline X) E\|$.
 Clearly, the partial distance is a semi-norm in $\Un$.
 \end{definition}

The proof of the last proposition also shows that $\widetilde X$ is a critical point of $\mathcal V$ if and only if
$E^\dag \widetilde X E$ can be written in the same form of  ${\widetilde W}_{11} = \bar U^\dag D \bar U$ that appears in that proposition, which is a more intrinsic characterization
of critical points.  Furthermore, if $\widetilde X$ is a critical point of $\mathcal V$, it is easy to show that, if $n_c$ is the number of diagonal elements of $D$
that are equal to $-1$, then $\mathcal V (\widetilde X) = 4  n_c$, and $\widetilde X$ admits the eigenvalue $-1$
 with multiplicity $n_c$.
 Let ${\mathcal S}_4 = \{ \widetilde X \in \Un ~|~\mathcal V(\widetilde X(0)) <4\}$.
 If the initial error matrix in the execution of RIGA is inside the open set ${\mathcal S}_4$, the fact that $\mathcal V (\widetilde X) = 4  n_c$
 at a critical point implies that no nontrivial critical point is contained in ${\mathcal S}_4$.
 As the Lyapunov function is always decreasing along the steps of RIGA (at least while $\Xgoalbar E = \Xgoalbarellminusone E$)
 then $\widetilde X(t)$ will never meet a nontrivial critical point of the Lyapunov function. If $\mathcal V(\widetilde X(0)) \geq 4$,
 a strategy for avoiding nontrivial critical points would be necessary.  Essentially, this would be accomplished by ``saturating''
 the eigenvalues of the error matrix to a region that does not contain $-1$. Before talking about these strategies, we present a remark about the invariance of  RIGA with respect to the matrix $E$ (that is defined in Prob. \ref{Prob1} in the introduction):
 
 \begin{remark}
\label{rY}
If $X(t)$ and $\overline X(t)$ are respectively solutions of \refeq{cqs} and \refeq{reference} will be useful to consider
the trajectories $Y(t) = X(t) E$ and $\overline Y(t) = \overline X(t) E$. Then, from \refeq{eFeedbackLaw}, from the fact that
${\widetilde S}_k = {\overline X}^\dag S_k \overline X$, and $\widetilde X = {\overline X}^\dag X$, it follows that
\begin{equation}
\label{eFeedbackLaw2}
 \widetilde u_k(t) = K \Re \left[ \trace \left( {\overline Y}(t)^\dag  S_k Y(t) \right) \right].
\end{equation}
Note that $Y(t)$ and $\overline Y(t)$ are the relevant parts of the encoded information to be tracked and the Lyapunov function is given by $\mathcal V ( \overline X(t)^\dag X(t) ) = \|Y(t) - \overline Y(t)\|^2$.
\end{remark}

From this last remark it is clear that RIGA could be implemented by simulating only $Y(t)$ and not the entire propagator $X(t)$. However, the entire propagator is needed for implementing the strategy for avoiding singular or critical points of the Lyapunov function.

\subsection{Avoiding critical and/or singular points}
\label{sFirstStrategy}

Recall that the natural choice of $X_{goal}$ in each step of RIGA is te one of the next remark.
 \begin{remark} \label{NoStrategy}
 Construct a fixed $X_{goal}$ such that $X_{goal} E = F$ and then
 choose the same  $\Xgoalbar = X_{goal}$ for all $\ell =1, 2, \ldots$. Then one cannot avoid the critical and/or singular points
 of the Lyapunov function and so RIGA may not converge globally to a solution of the encoded gate generation problem. 
 \end{remark}

 The choice $\Xgoalbar$ in each step $\ell$ of RIGA is done for avoiding critical and/or singular points of the Lyapunov function.
 From Proposition \ref{pCritical}, this would be accomplished by ``saturating''
 the eigenvalues of the error matrix to a region that does not contain the point -1, for instance, accepting eigenvalues of the form $\exp(\jmath {\overline \theta})$
 with $\overline\theta \in [-\pi/4, \pi/4]$ as done in \cite{PerSilRou19}.

 Let $X_f^{\ell-1} = X^{\ell-1}(T_f)$ be the final propagator of operation $\sharp 4$ in a step $\ell$ of RIGA.
 The algorithm for choosing  ${\overline X}_{goal}$ in each step $\ell$ of RIGA consists in two steps:\\
 \textbf{(a) Optimizing  $X_{goal}$ (only when $\nbar < n$).} Given  $X_{goal} \in \Un$  find an optimal $X_{goal}^{*}$ and an optimal phase $\phi \in \RR$ in order to minimize the Frobenius norm $\|X_f - X_{goal}^{*}\|$ under the restriction $X_{goal}^{*} E = \exp(\jmath \phi) X_{goal} E$.
 Another option is to take $\phi$ always equal to zero, considering the restriction $X_{goal}^{*} E = X_{goal} E$. The solution of this problem is based on singular value decompositions and it is presented in Theorem \ref{tOpt} of Appendix \ref{aOpt}. A solution  of this optimization problem will be denoted by $X_{goal}^{*} = \optgoal (X_{goal}, X_f)$. Note that
 the optimization problem has no sense when $\nbar = n$. Much more information about this optimization 
 problem can be found in the Appendix \ref{aOpt}.\\
 \textbf{(b) (Saturating eigenvalues).} Consider the usual saturation function $\mbox{sat} : \RR \rightarrow  \RR$ such that $\mbox{sat}(x) = x$, if $x \in [-\pi/4, \pi/4]$, $\mbox{sat}(x) = \pi/4$
 if $x > \pi/4$ and $\mbox{sat}(x) = -\pi/4$,  if $x < -\pi/4$.
  Compute $R = X_f^\dag X_{goal}^{*}$ and its eigenstructure\footnote{The Schur decomposition is indicated in this case}
  \[
  R = U^\dag \mbox{diag}[\exp(\imath \theta_1), \ldots,  \exp(\imath \theta_n)] U.
  \]
  Then compute
  \[
  R^{sat}(R) = U^\dag \mbox{diag} [\exp(\imath {\overline \theta}_1), \ldots,  \exp(\imath {\overline \theta}_n)] U,
  \]
where ${\overline \theta}_i = \mbox{sat}(\theta_i), i=1, \ldots, n$.
After that, define $\Xgoalbar = X_f R^{sat}(R)$.

Summarizing and using the notation defined above, the algorithm for choosing $\Xgoalbar$ in each step reads:
\[
\begin{array}{c}
\mbox{$R = X_f^\dag [\optgoal(X_{goal}, X_f)]$ (optimize first)}\\
\mbox{$\Xgoalbar = X_f R^{sat} (R)$ (saturate eigenvalues)}
\end{array}
\]
 By Remark \ref{rRIGA}, it follows that error matrix $\widetilde X(0)$ of step $\ell$ is given by $R_\ell^\dag = (R^{sat})^\dag$, which has also the same ``saturated'' eigenvalues.
 
 It is interesting to stress that the Theorem \ref{tEigenOpt} that is shown in Appendix \ref{aOpt} proves that this optimization process restricts the saturation of the eigenvalues to the sum of the spaces $\mathcal Y = \Image (X_f E)$ and $\mathcal F = \Image F$.  In fact, Theorem \ref{tEigenOpt} shows that, after the optimization process,  at least $n - 2 \nbar$ eigenvalues of $R$ will coincide with $1$, and so one needs to saturate at most $2 \nbar$ eigenvalues angles $\theta_i$ in this algorithm, the other are left invariant (and equal to zero).

Although this strategy of avoiding critical/singular points works well in numerical experiments of RIGA, convergence proofs for RIGA equipped with this strategy are very difficult to obtain. The authors have obtained analogous convergence proofs for the encoded case and for another strategy for avoiding singular and or critical points, that is the strategy that is analogous of \cite[Algorithms B and C]{PerSilRou19}. This different strategy and the corresponding convergence proofs for the encoded case are available in Appendix \ref{aMathematical}.

\section{Numerical Implementations of RIGA}
\label{sImplementation}

The implementations of RIGA may consider that the inputs $u_k(t)$ are piecewise-constant functions (as in GRAPE) or that they are smooth functions.
An executable code of a  MATLAB$^\circledR$ implementation of RIGA is available in \cite{CODE_OCEAN_CONSTANT}
for the piecewise-constant case (only for $\nbar = n$) and in  \cite{CODE_OCEAN_SMOOTH} for the smooth case (for all
values $\nbar$). The authors believe that the smooth implementation is more relevant than the piecewise-constant one for future applications.
Hence, all the numerical experiments of this paper are performed with the smooth implementation of RIGA.
However, the piecewise-constant implementation is also presented here because it is possible to establish a direct comparison between GRAPE and RIGA,
showing that GRAPE is a kind of open loop version of (the piecewise-constant) RIGA, in some sense to be precised in this section.
For the numerical implementations of RIGA, the interval $[0, T_f]$ will be divided in $N_{sim}$ equal parts,
 and  time $t$ will be discretized at instants $t_s = s \delta, s=0, 1 \ldots, N_{sim}$, where $\delta = T_f/N_{sim}$.
Consider that the inputs $u(\cdot)$ and the reference inputs ${\overline u}(\cdot)$ are piecewise-constant in the intervals $(t_{s}, t_{s+1}]$.
We shall denote:
\begin{subequations}
\label{eSamp}
\begin{eqnarray}
t_s & = & s \delta, s \in \{0, 1, \ldots, N_{sim}\}\\
{\overline u}_k (t_s)& = & {\overline u}_{k_s}\\
u_k (t_s)  & = & u_{k_s}.
\end{eqnarray}
\end{subequations}
Similarly, one will denote
$X_s = X(t_s)$, ${\overline X}_s = {\overline X}(t_s)$.
In the piecewise-constant case
it is assumed that
\begin{equation}
\label{ePC}
\begin{array}{l}
u_k(\tau) = u_{k_1}, \mbox{for}~\tau \in [0, \delta],\\
u_k(t_s + \tau) = u_{k_{s+1}}, \mbox{for}~\tau \in (0, \delta], \\
\mbox{for}~s=1, \ldots, N_{sim},\\
k = 1, 2, \ldots, m.
\end{array}
\end{equation}
The same properties hold by replacing $u_k(\cdot)$ by ${\overline u}_k$ in \eqref{ePC}.
In this way, for instance, the  input and the reference input may be represented by the $m  N_{sim}$-dimensional vectors:
\begin{equation}
\label{eUbarPiecewise}
\begin{array}{c}
 {\mathfrak U}  =  \left( {u}_{k_s} : k=1, \ldots, m, s=1, \ldots, N_{sim} \right)\\
\overline {\mathfrak U}  =  \left(  {\overline u}_{k_s} : k=1, \ldots, m, s=1, \ldots, N_{sim} \right)
\end{array}
\end{equation}

In the smooth case,  the following linear interpolation is used as
 an approximation for the open-loop simulations:
\begin{equation}
\label{ePL}
\begin{array}{l}
{\overline u}_k(t_s + \tau) = \left( \frac{\tau - \delta}{\delta} \right) {\overline u}_{k_s} + \left( \frac{\tau}{\delta} \right) {\overline u}_{k_{s+1}},  \\
\mbox{for}~s=0, 1 \ldots, N_{sim}-1\\
k = 1, 2, \ldots, m.
\end{array}
\end{equation}
 This means that
an input may be represented by $m  (N_{sim}+1)$-dimensional vector, for instance
\begin{equation}
\label{eUbarSmooth}
\overline {\mathfrak U} =  \left(  {\overline u}_{k_s} : k=1, \ldots, m, s=0, \ldots, N_{sim} \right)
\end{equation}
The details of each implementation are discussed in specific subsections in the sequel.

\subsection{Piecewise-constant implementation of RIGA}
\label{sPiecewise}

 For the sake of comparison between RIGA and GRAPE, it will be considered that
$\Xgoalbar$ is a fixed matrix $X_{goal}$ for all the steps $\ell =1, 2, \ldots $ of RIGA. This corresponds to the simplified version of RIGA of Section \ref{sRigaDescription}.
Note that a reference control is represented by a $m  N_{sim}$-vector given by \refeq{eUbarPiecewise}.
In this implementation, the reference system is simulated  backwards from $X_{goal}$, including operations $\sharp 2, \sharp 4$ and $\sharp 5$
in a single operation, without the need of computation $R_\ell$ (operation $\sharp 5$). The Lyapunov based feedback that is related to a Lyapunov function,
when computed at $t= t_{s-1}$
will be given by (see equation \refeq{eUtilk})
\[
 - K \nabla_{{\widetilde X}_{s-1}} {\mathcal V} \cdot ({\widetilde S}_{k_{s-1}} {\widetilde X}_{s-1})
\]
where ${\widetilde S}_{k_{s-1}} = {\overline X}_{s-1}^\dag {S}_k {\overline X}_{s-1}$ and
${\widetilde X}_{s-1} = {\overline X}_{s-1}^\dag { X}_{s-1}$.
Hence, at $t=t_s$, the feedback that will be applied to the system in a zero order approximation will
be
\[
 u_{k_s} = {\overline u}_{k_s} - K  \nabla_{{\widetilde X}_{s-1}} {\mathcal V} \cdot ({\widetilde S}_{k_{s-1}} {\widetilde X}_{s-1})
\]
Then, consider the following zero order pseudocode implementation of RIGA\footnote{The exponential of matrices may be implemented as a Pad\'{e} approximation \cite{PerSilRou19}.}:
\begin{center}
 \textbf{RIGA (Piecewise-constant implementation)}
\begin{tabbing}
 \=123 \=123 \=123 \=123 \kill\\
 $\sharp 1.$ Choose the seed input \\
$
\overline {\mathfrak U}^0 =  \left(  {\overline u}_{k_s}^0 : k=1, \ldots, m, s=1, 2, \ldots, N_{sim} \right)
$\\
 \textbf{Execute the steps $\ell=1, 2, 3, \ldots$}\\
  \textbf{BEGIN STEP $\ell$}.\\
            $\sharp 2.$ 
             \> \> REFER. SYST. SIMULATION (Backwards)\\
             \> \>  ${\overline X}_{Nsim} = X_{goal}$\\
             \> \> FOR $s=N_{sim}:-1:1$ \\
            \> \> \> ${\overline X}_{s-1}  = \exp \left[ -\delta  (S_0 + \sum_{k=1}^m {\overline u}_{k_s}^{\ell-1} S_k)\right] {\overline X}_{s}$\\
            \> \> END\\
$\sharp 3.$ \> \>  CLOSED-LOOP SYST. SIMULATION  \\
            \> \> $X_0 = I$\\
            \> \> FOR $s=1:N_{sim}$\\
            \> \>  \> ${\widetilde X}_{s-1} = {\overline X}_{s-1}^\dag X_{s-1}$\\
             \>  \> \> ${\widetilde S}_{k_{s-1}} =  {\overline X}_{s-1}^\dag { S}_k X_{s-1}$\\
            \> \> \> FOR $k=1:m$\\
            \>  \> \> \> ${\overline u}_{k_s}^\ell = {\overline u}_{k_s}^{\ell-1} - K \nabla_{{\widetilde X}_{s-1}} \mathcal V \cdot ({\widetilde S}_{k_{s-1}} {\widetilde X}_{s-1})$\\
            \> \> \> END\\
$\sharp 3^\prime. $ \> \> \> ${X}_{s}  = \exp \left[ \delta  (S_0 + \sum_{k=1}^m {\overline u}_{k_s}^\ell S_k)\right] { X}_{s-1}$\\
            \> \> END\\
            \> \>  Next (reference) input is:\\
            \> \> $\overline {\mathfrak U}^\ell =  \left(  {\overline u}_{k_s}^\ell : k=1, \ldots, m, s=1, 2, \ldots, N_{sim} \right)$\\
              \> \> $X_f^{\ell} = X_{N_{sim}}$.\\
 $\sharp 4.$   \> \> If the final infidelity $\mathcal I (X_f^{\ell})$ is acceptable\\
             \> \>  then terminate RIGA. Otherwise continue.\\
 \textbf{END STEP $\ell$}\\
\end{tabbing}
\end{center}

\subsection{A comparison with GRAPE}
\label{sComparison}

Consider the evolution map 
\[{\mathcal X} : \RR^{m N_{sim}} \rightarrow \Un\]
such that, given 
\begin{small}
\[
{\mathfrak U} =  \left(  { u}_{k_s} : k=1, \ldots, m, s=1, 2, \ldots, N_{sim} \right)
\in \RR^{m N_{sim}}\] 
\end{small}
then ${\mathcal X}(\mathfrak U)$ is defined by:
\begin{equation}
\label{eXs}
\left\{
\begin{array}{l}
X_0 = I\\
X_s   = \exp \left[ \delta  (S_0 + \sum_{k=1}^m { u}_{k_s} S_k)\right] { X}_{s-1}, \\
s=1, \ldots, N_{sim}\\
\mathcal X(\mathfrak U) = X_{Nsim}
\end{array}
\right.
\end{equation}

The first order GRAPE \cite{KHANEJA2005} considers an objective function $\Omega : \RR^{m N_{sim}} \rightarrow \RR$ such that
$\Omega (\mathfrak U) = \mathcal V (X_{goal}^\dag \mathcal X( \mathfrak U))$, where
$\mathcal V$ is some fidelity (or infidelity) measure that it will be optimized by a gradient
ascent (or descent) method. Here it will be considered that $\mathcal V$ is our Lyapunov function,
and so it is an infidelity measure.

 Recall that the first  order version of GRAPE \cite{KHANEJA2005} is essentially
the gradient descent (or ascent) algorithm for this objective function:
\begin{center}
 \textbf{(GRAPE - first order)}
\begin{tabbing}
 \=123 \=123 \=123 \=123 \kill\\
 $\sharp 1.$ Choose $\Omega_*$, the acceptable value of \\
 the objective function.  Choose the seed input \\
$
{\mathfrak U}^0 =  \left(  { u}_{k_s}^0 : k=1, \ldots, m, s=1, 2, \ldots, N_{sim} \right)
$\\
 Execute the steps $\ell=1, 2, 3, \ldots, \ell^*$\\
  \textbf{BEGIN STEP $\ell$}.\\
$\sharp 2.$   \> \> Compute $\Omega_\ell = \Omega ({{\mathfrak U}^{\ell-1}})$.  \\
              \> \> If $\Omega_\ell \leq \Omega_*$, then stop. Otherwise continue.\\
$\sharp 3.$   \> \> Compute the gradient $\nabla_{{\mathfrak U}^{\ell-1}} \Omega$ \\
$\sharp 4.$   \> \> Set ${\mathfrak U}^{\ell} = {\mathfrak U}^{\ell-1} - K \, \nabla_{{\mathfrak U}^{\ell-1}} \Omega$\\
   \textbf{END}
 \end{tabbing}
 \end{center}

The piecewise-constant implementation of RIGA can be regarded as a closed-loop version of GRAPE, as stated in the following result.
\begin{theorem}
\label{tGRAPE_RIGA}
If one replaces ${\overline u}_{k_s}^\ell$ by ${\overline u}_{k_s}^{\ell-1}$ in the operation $\sharp 3^\prime$ of the piecewise-implementation
of RIGA, then one obtains a particular implementation of GRAPE.
\end{theorem}

\begin{proof}
See Appendix \ref{aGRAPE_RIGA}.
\end{proof}

 It is clear from  figure \ref{fGrapeandRiga} that GRAPE may be interpreted as being an open loop
version of RIGA. In the computation of  step $\ell$ of RIGA, the input ${\overline u}_{k_s}^\ell$ is applied ``on the fly'', that is,
the feedback is dynamically computed inside a step $\ell$. So, what is being simulated is really the closed loop system.
On the contrary, for GRAPE, what is being simulated in $\sharp 3$ is the open-loop system. The input is updated only
in the end of each step $\ell$.

\begin{figure*}[t]
    \centering{\includegraphics[scale=0.70]{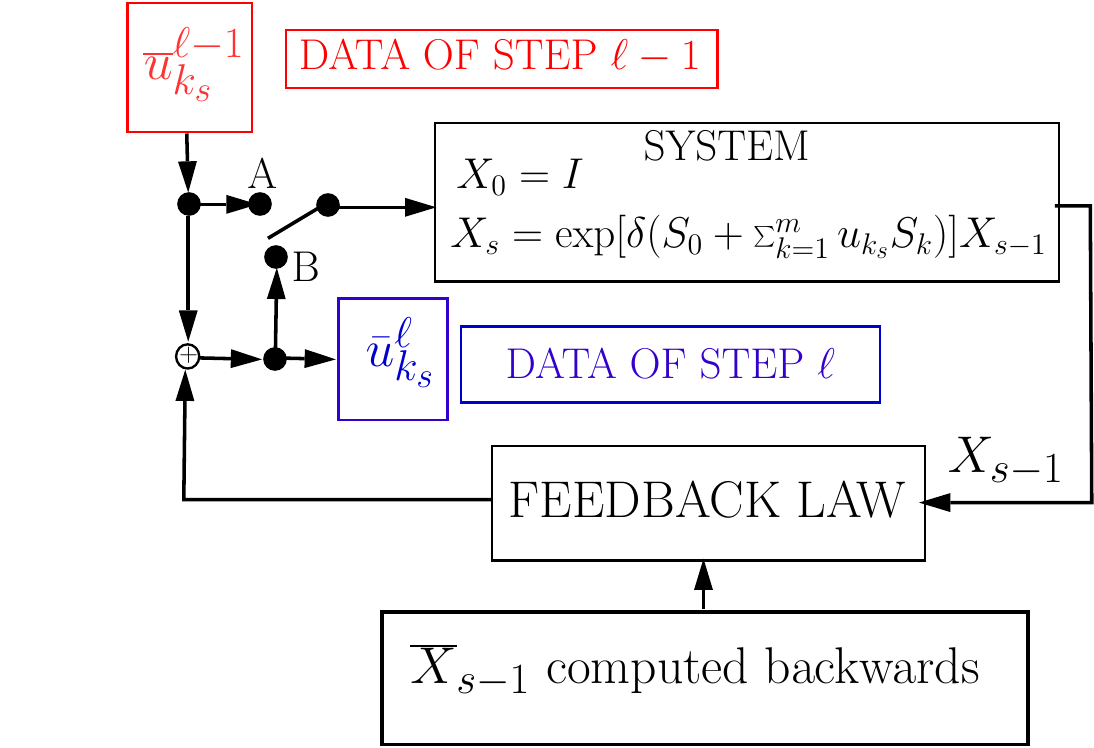}}
    \caption{The block diagram represents the statement of Theorem \ref{tGRAPE_RIGA}.  When the switch is at the position A (open loop), the block diagram corresponds to an implementation of GRAPE. When the switch is at the position  B (closed
     loop), the block diagram corresponds to a piecewise-constant implementation of RIGA.}
    \label{fGrapeandRiga}
\end{figure*}

\subsection{The smooth implementation of RIGA}
\label{sSmooth}

The smooth implementation of RIGA is based on the Cayley transformation \cite{Die98} and a standard 4th-order Runge-Kutta scheme.
A complete code for a MATLAB$^\circledR$ implementation can be found in \cite{CODE_OCEAN_SMOOTH}.
Using the same notation \refeq{eSamp}, recall that the control inputs are supposed to be smooth and in the open loop integration they will be approximated to piecewise-linear functions, which means that \refeq{ePL} holds. The idea of this implementation is to develop a 4th-order Runge-Kutta integration scheme for both closed-loop and open-loop cases. Since $\Un$ is not an Euclidean space, there is no sense in applying a Runge-Kutta method directly to the dynamics\footnote{The Runge Kutta method in this case will generate
non-unitary matrices.} \refeq{reference}--\refeq{eClosedLoopComplete} (or even to \refeq{reference}--\refeq{eClosedLoopComplete2}). The idea is based on the method that is proposed in \cite{Die98}, and relies on  a\footnote{The Cayley transformation considered in \cite{Die98} is given by $-\mathfrak W(\cdot)$.} smooth-map $\mathfrak W$ that is similar to the homographic function that was considered in \cite{SilPerRou14} for defining the
Lyapunov function $\mathcal V$ that is used in \cite{PerSilRou19}.

Let $\mathfrak{W} : \mathcal W \subset \Un \rightarrow \Sigma \subset \un$ defined by 
\begin{equation}
    \label{eCayley}
\mathfrak{W}(\widetilde X) = (\widetilde X - I)(\widetilde X + I)^{-1},
\end{equation}
where $\Sigma \subset \un$ is the set of anti-hermitean  complex matrices
$\sigma$ such that $(\sigma - I)$ is invertible. Recall that $\mathcal W$ is the set of complex matrices $\widetilde X$ such that $(\widetilde X + I)$ is invertible.  It is easy to show that $(\mathfrak W( \widetilde X) - I) = -2 (\widetilde X + I)^{-1}$ and so, the following result is straightforward to be shown: 
\begin{proposition}
Let $\widetilde X \in \mathcal W$. Then,
 the matrix $(\mathfrak W ( \widetilde X) - I)$ is always invertible, with inverse $-\frac{1}{2} (\widetilde X + I)$. Furthermore, the inverse of the map $\mathfrak W$ is the smooth map $\mathfrak X :  \Sigma \rightarrow \mathcal W$ such that  
 \begin{equation}
 \label{eInvCayley}
 \mathfrak X(W) = -(W-I)^{-1} (W+I).
 \end{equation}
\end{proposition}

\begin{remark} 
The computation of the maps $\mathfrak W$ and $\mathfrak X$ could be done by using the Schur decomposition. For instance, since $\Un$ is a set of normal matrices\footnote{Recall that a square complex matrix $U$ is normal if $U^\dag U = U U^\dag$.}, the Schur decomposition $\widetilde X  = U D U^\dag$ coincides with the eigenstructure of $\widetilde X$, where $U$ is a unitary matrix and  $D$ is the diagonal matrix whose entries are the eigenvalues of $\widetilde X$. Then is easy to show that  
$\mathfrak W(\widetilde X) = U (D-I)(D+I)^{-1} U^\dag$, which relies on the inversion of a diagonal matrix. In the same way, as $\un$ is also a set of normal matrices, if $W = V D_1 V^\dag$ is the Schur decomposition of $W \in \Sigma$, then $\mathfrak X(W) = - V (D_1 - I)^{-1} (D_1 + I) V$. However, the implementation of the 4th order Runge-Kutta integration scheme that is described in the seqquel is done in a way that the matrix inversions of the computations are always well conditioned. Hence the numerical results using matrices inversions have better precision a faster runtime than the ones that have used  Shur decomposition\footnote{Using the MATLAB\texttrademark function  $\mbox{mldivide}$ is better than using division and multiplication. For instance, the MATLAB\texttrademark implementation of $\mathfrak X(W)$ is $\mbox{mldivide}(W-I,W+I)$. }.
\end{remark} 

Some standard computations show easily that, if $\widetilde X(t) \in \Un$ is a solution of a differential equation $\frac{d}{dt} \widetilde X(t) = \widetilde S(t)  \widetilde X(t)$, with $\widetilde S(t) \in \un$,then:
\begin{equation}
 \label{Wdynamics}
\dot W(t) = - \frac{1}{2} (W(t)-I) \widetilde S(t) (W(t) + I)
\end{equation}
The equation \eqref{Wdynamics} may be integrated numerically, instead of \refeq{cqs}, with the advantage that $\un$ is an Euclidean space, and so the Runge-Kutta method may be applied in a natural way.
Define the map $F  : \Sigma \times \un \rightarrow \Sigma$ by:
\begin{subequations}
\label{Fdynamics}
\begin{equation}
 \label{FW}
   F (W, \widetilde S) = - \frac{1}{2} (W-I) \widetilde S (W + I)
\end{equation}
then equation \refeq{Wdynamics} reads 
\begin{equation}
\label{WdynamicsF}
\dot W(t) = F(W(t), \widetilde S(t)). 
\end{equation}
\end{subequations}
Each 4th order Runge-Kutta step \cite{BreCamPet95} of the open loop integration of \refeq{Wdynamics} considers the interval $[t_s, t_{s+1}]$
with $W(t_s) = 0$ (corresponding to the identity in $\Un$, which is far from the frontier of $\Sigma$), with posterior correction by right-invariance. Note that $W(t)$ may not be too close to the frontier of the region $\Sigma$, otherwise
 a numerical problem will certainly occur. By linear interpolation, define the linear interpolation
\begin{equation}
\label{eLinearInterpolation}
\begin{array}{l}
{\overline u}_k^{\ell-1} (t_s+\tau) =\left(  \frac{\delta-\tau}{\delta} {\overline u}_{k_s}^{\ell-1} + \frac{\tau}{\delta} {\overline u}_{k_{s+1}}^{\ell-1} \right)\\
s=0, 1, \ldots, N_{sim}-1\\
k=1, \ldots, m\\
\tau \in \{0, \delta/2, \delta\}
\end{array}
\end{equation}
Define $\Sigma (s, \tau) = S_0 + \sum_{k=1}^m {\overline u}_k^{\ell-1} (t_s+\tau) S_k$ and $G(W, s, \tau) = F(W, \Sigma(s, \tau))$ where the map $F$ is defined in \eqref{Fdynamics}.
Then, for the dynamics \eqref{Fdynamics}, each step $s$ of the 4th-order Runge-Kutta for the open loop system for $s=1, 2, \ldots , N_{sim}$, the implementation for the case where $\nbar$ is less than $n$ reads:
\begin{tabbing}
1234 \= 1234 \= 1234 \= 1234 \kill
    \>  \> $\framebox{  $ W_0 = 0$}$\\
    \>    \> $k_1 = \delta G(W_0, s, 0)$\\
    \>    \> $k_2 = \delta G(W_0+ \frac{k1}{2}, s, \frac{\delta}{2}))$\\
    \>    \> $k_3 = \delta G(W_0+ \frac{k2}{2}, s, \frac{\delta}{2})$\\
    \>    \> $k_4 = \delta G(W_0+ {k_3}, s, \delta)$ \\
    \>    \> ${\overline W}_{s} = W_0 + \frac{1}{6} (k_1 + 2 k_2 + 2 k_3 + k_4)$ \\
    \>    \> $X_{s+1} = \mathfrak X({\overline W}_{s}) X_s$ \\
    \>    \> $\framebox{${\overline X}^{\ell-1}_{s+1} =  X_{s+1} R_\ell$}$ \\
    \> \> $\%$  ${\overline X}^{\ell-1} =$ Ref. trajec. of  step $\ell$ \%\%\\
\end{tabbing}
For the simulation of the closed loop system \refeq{eClosedLoopComplete2}, note from  \refeq{eFeedbackLaw2} that we need to
compute  ${\mu}_k(s, \tau, X) = {\overline u}_k^{\ell-1}(t_s+\tau) + K \Re [\trace(E^\dag {\overline X}^{\ell-1}(t_s + \tau) S_k X E)]$
as function of $s$, $\tau$ and $X \in \Un$  for $\tau \in \{0, \delta/2, \delta\}$. The values of ${\overline u}_k^{\ell-1}(t_s+\tau)$
are computed by linear interpolation \refeq{eLinearInterpolation}.
The values of ${\overline X}^{\ell-1}(t_s) = {\overline X}_s$, ${\overline X}^{\ell-1}(t_{s+1})= {\overline X}_{s+1}$ are available
but the value of ${\overline X}(t_s+\delta/2)$ is interpolated. This is done in the transformed space $\un$ by the expression:
\begin{equation}
\label{eInterpol}
{\overline X}(t_s+\tau)=
\left\{
\begin{array}{l}
 {\overline X}_s,~\tau=0,\\
 \mathfrak X (\frac{{\overline W}_0 + {\overline W}_s}{2}) {\overline X}_s, ~\tau=\delta/2,\\
 {\overline X}_{s+1},~\tau=\delta.
 \end{array}
 \right.
\end{equation}
Note that, for $\tau = \delta/2$, the interpolation corresponds to  the middle point between  ${\overline W}_0=0$ and ${\overline W}_s$ in the Euclidean space $\un$.
Define the map ${\Sigma}(s, \tau, X) = S_0 + \sum_{i=1}^m \mu_k (s, \tau, X) S_k$. Then
define $ G_1(W,s,\tau) = F(W, { \Sigma}(s, \tau, \mathfrak X(W))$, where $\tau \in \{0, \delta/2, \delta\}$. Then each step $s \in \{0, 1, \ldots, N_{sim}-1 \}$
of the 4th-order Runge-Kutta \cite{BreCamPet95} for the closed loop
dynamics simulation reads:
\begin{tabbing}
1234 \= 1234 \= 1234 \= 1234 \kill
    \>    \> ${\overline u}_{k_s}^\ell = \mu_k^{\ell} (s, 0, X_s)$ \\
    \>   \> $\framebox{  ${\overline W}_0 = 0$}$\\
    \>    \> $k_1 = \delta  G_1({\overline W}_0, s, 0)$\\
    \>    \> $k_2 = \delta  G_1({\overline W}_0+ \frac{k1}{2}, s, \frac{\delta}{2}))$\\
    \>    \> $k_3 = \delta  G_1({\overline W}_0+ \frac{k2}{2}, s, \frac{\delta}{2})$\\
    \>    \> $k_4 = \delta  G_1({\overline W}_0+ {k_3}, s, \delta)$ \\
    \>    \> ${ W}_{s} = {\overline W}_0 + \frac{1}{6} (k_1 + 2 k_2 + 2 k_3 + k_4)$ \\
    \>    \> $X_{s+1} = \mathfrak X({ W}_{s}) X_s$ \\
    \>  \>  IF $s == N_{sim}-1$\\
    \>  \> \>    ${\overline u}_{k_{s+1}}^\ell = \mu_k^{\ell} (s, \delta, X_{s+1})$. \\
    \>  \>  END
\end{tabbing}

\begin{remark} 
\label{rSharps} 
Consider now the complete description of RIGA of section \ref{s:Refinements}.
We explain here why $\sharp 2$ and $\sharp 6$ are both necessary in RIGA.
The nature of interpolations of the closed loop system and the open loop system are different in the sense that the interpolation of the reference trajectory of the closed loop systems
does not correspond to the input interpolation of the open loop system. The interpolation of the 
closed loop system considers the feedback law, whereas the interpolation of the open loop system
considers a piecewise linear interpolation.
So, if the control pulses will be generated by linear  interpolation \refeq{eLinearInterpolation},
a repetition of the simulation, as done in the operation $\sharp 2$ of RIGA is justified. If the control pulses are generated by other interpolation method, then the open
loop simulation of operation $\sharp 2$ of RIGA must replace \refeq{eLinearInterpolation} by the same interpolation method. However, there is a more profound and important reason in order to justify the apparently repetition of tasks of $\sharp 2$ and $\sharp 6$. In a single step, the difference between 
such computations could be neglected. However, the difference between the computation $\sharp 2$ and $\sharp 6$ will be cumulative, since several steps of integration considering a different interpolation method will produce an important final error in the end of RIGA. This fact was confirmed by several numerical experiments that shows that using only $\sharp 6$ for computing RIGA produces an algorithm whose final error increases with the number of steps, while the implementation that considers both $\sharp 2$ and $\sharp 6$ produces an algorithm that maintains the same numerical error of a single step, that is, the error is not cumulative along the computation of several steps.
\end{remark}

In the case which $n$ coincides with $\nbar$, the implementation of RIGA of \cite{CODE_OCEAN_SMOOTH} considers the Lyapunov function \refeq{eLyap2} of \cite{PerSilRou19}.
The dynamics of the closed loop system that is implemented is \refeq{eClosedLoopComplete} (with state $\widetilde X(t)$, instead of $X(t)$ of \refeq{eClosedLoopComplete2}).
The advantage is that it is not necessary to choose $W_0$ to be equal to identity, and to correct the value by right-invariance (which demands the calculation of  $\mathfrak X({ W}_{s})$
in every step $s = 0, 1, \ldots, N_{sim}-1$). In fact, the domain of the Lyapunov function \refeq{eLyap2}
is the set $\mathcal W$ and furthermore, if ${\widetilde X}_0 \in \mathcal W$, then for the closed loop system $\widetilde X(t) \in \mathcal W$ for $t \in [0, T_f]$. This is implied by the closed loop
monotonic behavior of the closed loop system, that is, $\mathcal V ({\widetilde X(t)}) \leq \mathcal V(\widetilde X(0))$ for all $t \in [0, T_f]$ (see \cite{PerSilRou19}). This will assure that the Lyapunov function is always well defined and so $\mathfrak(\widetilde X(t))$ is also always well defined.
Some simple computations shows that $\mathfrak Z (W) =Z(\mathfrak X( W)) = \frac{1}{4} W (W+I)(W-I)$, which simplifies the computation of the feedback law as a function of
$W$. Then define ${\widetilde \mu}_k (s, \tau, W) = {\overline u}_{k_s}^{\ell-1} +  K \trace[ \mathfrak Z (W) {\overline X}(t_s+\tau)^\dag S_k {\overline X}(t_s+\tau)]$,
where  ${\overline X}(t_s+\tau)$ is given by \refeq{eInterpol}. Then define ${\widetilde \Sigma}(s, \tau, W) = \sum_{i=1}^m {\widetilde \mu}_k^ (s, \tau, X) {\overline X}(t_s+\tau)^\dag S_k {\overline X}(t_s+\tau))$
and $\widetilde G(W, s, \tau) = F(W, {\widetilde \Sigma}(s, \tau, W))$. Then, the implementation of each step $s$ of the closed loop system for $n = \nbar$ reads:
\begin{tabbing}
1234 \= 1234 \= 1234 \= 1234 \kill
     \>   \> ${\overline u}_{k_s}^\ell = {\widetilde \mu}_k^{\ell} (s, 0, W_s)$ \\
    \>    \> $k_1 = \delta \widetilde G(W_s, s, 0)$\\
    \>    \> $k_2 = \delta \widetilde G(W_s+ \frac{k1}{2}, s, \frac{\delta}{2}))$\\
    \>    \> $k_3 = \delta \widetilde G(W_s+ \frac{k2}{2}, s, \frac{\delta}{2})$\\
    \>    \> $k_4 = \delta \widetilde G(W_s+ {k_3}, s, \delta)$ \\
    \>    \> ${ W}_{s+1} = W_s + \frac{1}{6} (k_1 + 2 k_2 + 2 k_3 + k_4)$ \\
    \>    \> $\% \% \%$ ${\widetilde X}^\ell_{s+1} = \mathfrak X({ W}_{s+1})$\\
     \>  \>  IF $s == N_{sim}-1$\\
    \>  \> \>      ${\overline u}_{k_{s+1}}^\ell = {\widetilde \mu}_k (s, \delta, W_{s+1})$. \\
    \>  \>   END
\end{tabbing}
Note that the computation of ${\widetilde X}^\ell_{s+1}$ is unnecessary, since the feedback may be computed directly a  function of $W \in \un$,
the transformed state by the (minus) Cayley-transformation. Furthermore, $X^{\ell}(t)$ will be computed in operation $\sharp 2$ of the next
step of RIGA.

\subsection{Error analysis of the smooth case}

The present error analysis regards the open-loop simulation with 4th-order Runge-Kutta (RK4).
For this error analysis,  we have chosen the third  example, that is, the $N$-qubit system with $N=6$ qubits.
 We have considered a fixed choice the control inputs \eqref{eUbarSmooth} that are obtained in the end of the execution of RIGA in that case.
 We have re-simulated  the open-loop 4th order Runge-Kutta integration with smaller steps $\frac{\delta}{r}$.
 Let $X_f(r)$ be the final  propagator that is obtained with time-step $\frac{\delta}{r}$ considering $X_0=I$.
 The control pulses are always computed from a expression that that is similar to \eqref{eLinearInterpolation} but
 now $\tau$ may take the values in the set $\{0, \frac{\delta}{r},   \frac{2 \delta}{r}, \ldots , \frac{(r-1) \delta}{r}, \delta\}$.

For instance $X_f(8)$ is the propagator that is obtained by RK4 with a step $\frac{\delta}{8}$.
The infidelity between $X_f(8)$ and $X_f(r)$, $r=1,2,3$ gives a first measure of the precision.
 Let $X_{fexp} (r)$ be the propagator obtained by assuming that the input is piecewise-constant in intervals
      $\frac{\delta}{r}$.
 The matrix $X_{fexp} (r)$ will be computed with the MATLAB$^{\circledR}$ function \textbf{expm} similarly to the linear piecewise-constant implementation;
The piecewise-constant control inputs are obtained by sampling the linearly interpolated inputs at $t_s = s \frac{\delta}{r}$.
The infidelity between $X_f(8)$  and $X_{fexp}(r)$ gives a second measure of the precision.
\begin{figure*}[t]
    \centering{\includegraphics[scale=0.7]{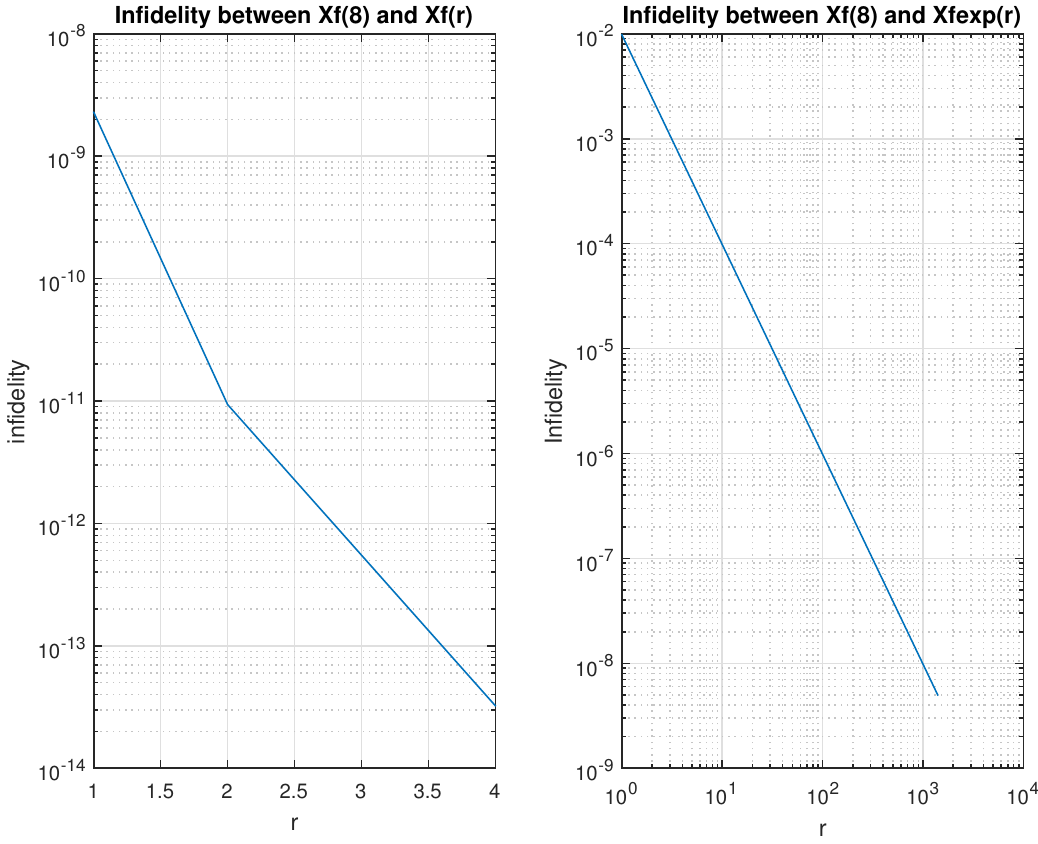}}
    \caption{LEFT PLOT: Infidelity between $X_f(8)$ and $X_f(r)$, where $X_f(r)$ = integration with RUNGE-KUTTA with step is $\frac{\delta}{r}$.
    RIGHT PLOT : Infidelity between $X_f(8)$ and $X_{fexp}(r)$, where $X_{fexp}(r)$ = piecewise-constant computation with ``expm'' and step $\frac{\delta}{r}$.}
     \label{fXfexp}
\end{figure*}

What  figure \ref{fXfexp} indicates is the following. The precision RK4 of this implementation
increases very fast with $r$ (recall that the step is $\frac{\delta}{r}$). The approximation by piecewise-constant
inputs tends to the smooth case when $r$ tends to infinite, but one needs very small steps for approximating the piecewise-constant case
to the solution that is obtained with RK4. For instance, the right plot of Figure \ref{fXfexp} shows that the precision
corresponding to a step $\delta/r$ with $r=1000$ is needed for the piecewise-constant case to recover the same precision of RK4 with
step $\delta$.

Now, one studies what happens when the integration is done in $\Un$ directly, that is, without the Cayley transformation,
 As it produces a non-unitary matrix at every step,  another possible option is to \emph{correct} the result with a projection into
  $\Un$ at every step of the 4th-order Runge-Kutta. In fact, if $X= U^\dag \Sigma V^\dag$,
  is the singular value decomposition of $X$,
  then $W^*= U^\dag V^\dag$ is  called ``unitary projection''. It is well known that $W^*$ is the closest unitary matrix, in the sense
  that the Frobenius norm $\|W - X\|$ for $W \in \Un$ is minimal for  $W=W^*$ (it is easy to show this from the results of \cite{Kel75}).
 The precision of the integration with the Cayley transformation is compared with the one of the integration in $\Un$,
 with or without correction by the projection into $\Un$. Figure \ref{fPrecision} summarizes the obtained results. The legend ``Infidelity W'' corresponds
 to the integration in $\un$ with the Cayley transformation, ``Infidelity in $\Un$'' corresponds to the integration directly in $\Un$,
 and the legend ``Infidelity in $\Un$ corr'' corresponds to the integration in $\Un$ equipped with the correction by the projection into $\Un$.

\begin{figure*}[t]
    \centering{\includegraphics[scale=0.7]{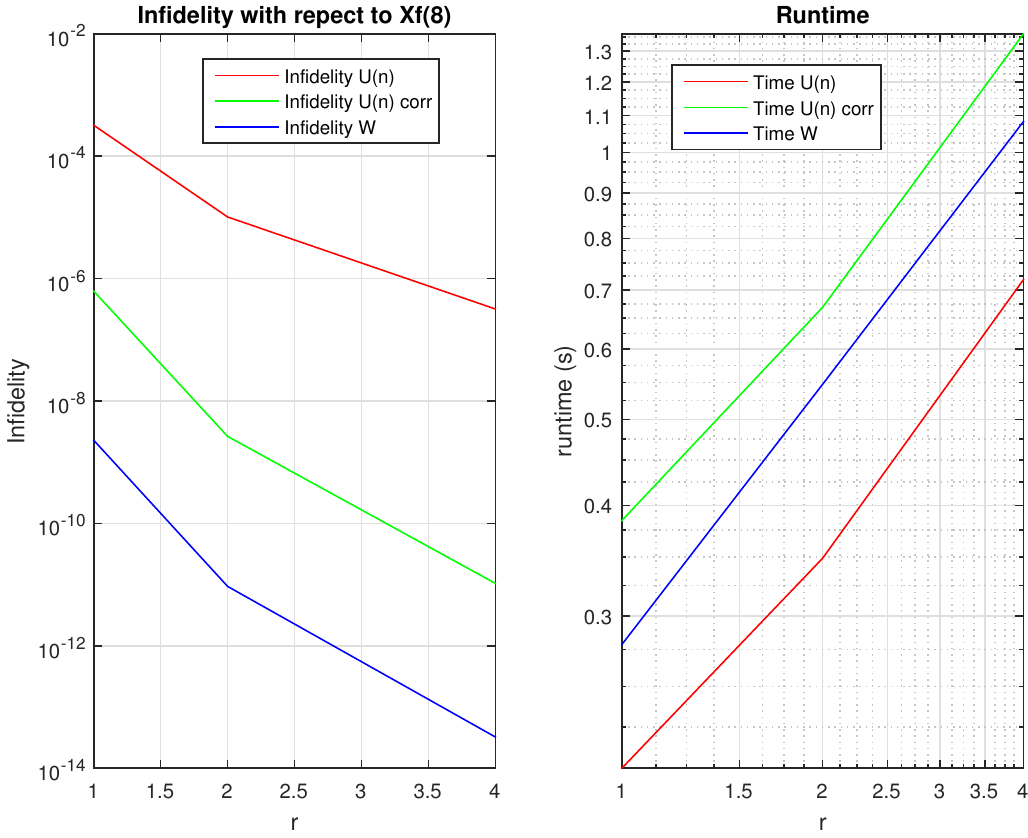}}
    \caption{LEFT PLOT: Infidelity between $X_f(8)$ and the integration in $\Un$ with step is $\frac{\delta}{r}$ with and without correction (legends ``Infidelity in $\Un$'', and ``Infidelity in $\Un$ corr'', respectively. Legend ``Infidelity W'' corresponds to the infidelity between $X_f(8)$ and the integration in $\un$ with step is $\frac{\delta}{r}$.
    RIGHT PLOT : Runtime corresponding of each case of the left plot.}
    \label{fPrecision}
\end{figure*}

Now, what Figure \ref{fPrecision} indicates is that, considering the same runtime, for instance equal to $0.4s$, the precision of the ``uncorrected''  integration
in $\Un$ is worse than 1e-6, the precision of the corrected one is not far from 1e-6, and the precision of RK4 using the Cayley transformation
is not far from 1e-10! Now for a precision of the order of 1e-10,  then   RK4 using the Cayley transformation needs $r\approx 1.6$ (corresponding to a runtime $\approx 0.45s$, the corrected integration in $\Un$ needs
$r\approx 3.2$ (corresponding to a runtime $\approx 1 s$, and the non-corrected integration in $\Un$ will not reach a precision better than 5e-7 for $r=4$ corresponding to a runtime  $\approx 1$ second.
This confirms the study of \cite{Die98}, that claims that the use of the Cayley-transformation is a very efficient way to integrate a dynamics evolving in $\Un$, with an excellent compromise between runtime and precision.

\section{Conclusions}
\label{sConclusions}

The numerical experiments of this paper have shown that RIGA is a powerful algorithm for state preparation  and quantum gate generation for closed quantum systems.
The smooth implementation of RIGA \cite{CODE_OCEAN_SMOOTH}  produces smooth control pulses, which is an advantage with respect to Krotov and GRAPE  that produces piecewise-constant control
pulses (only GOAT, CRAB and RIGA produces smooth pulses). RIGA can tackle  the basic problems without using the penalty functions that are commonly included in order to shorten
\cite{LeuAbdKocSch17}:\\
(A) The bandwidth of the control pulses;\\
(B) The amplitude of control pulses and the one of its time-derivatives.\\
 Note that the presence of penalty functions is a barrier to the convergence of the infidelity to zero for all methods, indeed.

The bandwidth of the control pulses produced by RIGA depends on
 the bandwidth of its seed and on the feedback gain $K$. Greater is the gain, greater is the possibility to appear
 spurious frequencies in the control pulses.

 The heart of RIGA is the choice of the Lyapunov function, that is summarized in the following table

\begin{footnotesize}
\begin{table*}
\begin{tabular}{|c|c|c|c|c|}
  \hline
   $\bar n$ & Lyapunov              & Bounded & set  of critical         & Feedback  \\          & Function $\mathcal V$ & Control & points of $\mathcal V$   & law\\
  \hline
   $ \bar n \leq n$   &  $ \bar n-  \Re \left\{\trace \left[  E^\dag \widetilde X E  \right] \right\}$     & Yes &   Prop. \ref{pCritical} & ${\widetilde u}_k = K \Re \left\{ \trace \left[ E^\dag  {\widetilde S}_k \widetilde X  E \right] \right\}$ \\
                                        &  &     &   part (b)  &                    \\
               \hline
   $\bar n = n$  & $ \trace[ (\widetilde X-I)^2 (\widetilde X+I)^{-2}] =$ & No   & $\{ I \}$ & ${\widetilde u}_k = K \trace [Z (\widetilde X) {\widetilde S}_k] $ \\
               &  $\|(\widetilde X-I) (\widetilde X+I)^ {-1}  \|^2  $                          & & &$Z({\widetilde X}) = {\widetilde X} ({\widetilde X}-I)({\widetilde X}+I)^{-3}$  \\
  \hline
  \end{tabular}
\end{table*}
\end{footnotesize}

 The following remarks are important:
\begin{itemize}
\item Smaller values of the gain $K$ tends to generate smaller values of amplitude of the control pulses, with a compromise
     with the speed of convergence of RIGA.

\item Higher values of $K$ may generate numerical imprecision (and even numerical instability) of the Runge-Kutta method, and the presence of undesirable high frequencies in the control pulses.
To increase the precision for a fixed $K$, one needs to increase $N_{sim}$ , increasing linearly the runtime of each step.

\item The present ``smooth-version''  of  RIGA \cite{CODE_OCEAN_SMOOTH} considers a fixed-step 4th-order Runge-Kutta integration method.

\item We believe that using adaptation of the size of the step is not a great advantage, considering the compromise of precision
and runtime.

\end{itemize}

The implementation of RIGA \cite{CODE_OCEAN_SMOOTH} has the folowing features:
\begin{itemize}
\item The user may specify a maximum amplitude $u_{max}$ of the control pulses of RIGA (that must be respected by the seed inputs ${\overline u}_k^0(t)$).

\item A ``smooth'' saturation function is implemented in \cite{CODE_OCEAN_SMOOTH} that allows to respect the bound $u_{max}$.
    Again, the use of small gains is recommended, because higher gains may produce a set of everywhere saturated control pulses.

\item A Hamming window option generates smooth control pulses that are null at the endpoints of $[0, T_f]$ and vary smoothly
      inside the interval $[0, T_f]$ (see  \refeq{eHamming}). This leads to smaller bandwidths of control pulses.

\item RIGA may also include penalty functions as in \cite{LeuAbdKocSch17}. In particular,  a penalty function for  minimizing a forbidden population is implemented
      in \cite{CODE_OCEAN_SMOOTH}.

\item For the case of cavities, we believe that the best solution found in the numerical experiments is not to include such penalty function,
but to consider a number of levels $n_c$ ``\textbf{large enough}'' in the model. The infidelity will converge to zero and the forbidden populations
will be small because of physical reasons. For instance, as the energy fast grows with the cavity level, it is natural to expect that bounded inputs will produce
small populations in higher levels of the cavity.

 \item The spectra of the feedback ${\widetilde u}_k$, of the control pulses ${\overline u}_k^\ell$, and of the seed  is shown in every step $\ell$ of RIGA . This information helps to choose  $T$ and $M$ considering that the seed may not include
        artificial high frequencies when compared with the feedback.

  \item Looking to the spectrum of the feedback  is useful to see if $N_{sim}$ must be raised (or equivalently, if $\delta$ must be shortened).
       In fact, the feedback must have components of its spectrum
        that are small with respect to the Nyquist frequency related to the period $\delta$ (the maximum frequency that is show in the spectra in the implementation of \cite{CODE_OCEAN_SMOOTH}.).

  \item The fidelity  between the final state $X_f$ obtained by open loop integration with step $\delta= \frac{T_f}{N_{sim}}$ and the one obtained with the half step
  $\delta/2 = \frac{T_f}{2 N_{sim}}$ is informed.

  \item Again, this last information is a measure of the numerical precision, it can be used for estimating whether or not  $N_{sim}$ must be raised for increasing the precision.

\item The last step of RIGA also furnishes the same information, and so the user may estimate the numerical precision of the solution
      that was furnished by RIGA.

\end{itemize}

Now, we compare our results of the first example with the ones of \cite{LeuAbdKocSch17} obtained with GRAPE. As the main parameters, including the system parameters, the desired gate and the final time $T_f$ are the same,
a direct comparison of the results is possible. For the first example,  the appearance of the control pulses obtained in the present paper is not that different
from the results of that paper (for the C-NOT gate only, since the state preparation considered here is not tackled in that paper). In both papers, the rotating wave approximation is not used,
and a bandwidth of the order of 5 GHz is needed. The amplitude of the pulses are also of the same order. For the second example, the appearance of the control pulses of \cite{CavityTransmonGrape}  are
also very similar to the ones of the present paper. The bandwidth of the control pulses estimated to be  $\approx 30Mhz$, which is also of the order of the bandwidth that was found in  \cite{CavityTransmonGrape}. The main difference
for both examples  is that GRAPE is a piecewise-constant method, whereas the implementation \cite{CODE_OCEAN_SMOOTH} of RIGA
consider smooth control pulses. This could be an advantage in practical applications, specially if the control pulse generator is able to consider
 these pulses, which are now possible by the present technology. On the one hand, if the pulse generator implements piecewise-constant control pulses, it will produce
  small discontinuities that will certainly add spurious
 high-frequencies that are not taken into account in the computations of FFT's that was done by any GRAPE implementation. These high frequencies may be filtered by the communication
 channels that link the quantum system to the control pulse generator, producing a degradation in the fidelity.
 On the other hand, if the control pulses are pre-filtered, a degradation of the fidelity will occur, anyway, since a piecewise method will not be able to take into account this filtering process.
 Hence it seems that it is much more appropriate to consider a theory that is able  generate a smooth (or at least a continuous) control pulse.
 The present implementation of RIGA \cite{CODE_OCEAN_SMOOTH} considers control pulses that are continuous and are  linearly interpolated, although one could consider other interpolation methods, like splines, that could
include the continuity of the derivative of the pulses at the instants $t_s = s \delta/N_{sim}, s=1,2, \ldots, N_{sim}$.

 For the third example,  for instance when $N=10$ qubits,  the
the propagator $X(t)$ is a $1024\times 1024$ matrix. This shows that RIGA is able to tackle high-dimensional systems with a runtime that is compatible with the an implementation of GRAPE,
even running in a worst CPU than the one  of  \cite{LeuAbdKocSch17}. A nice feature of RIGA that can be noted in this example, is that the runtime seems to grow exponentially with $N$ for GRAPE, whereas the runtime does not grow exponentially with $N$ for RIGA (this fact was also observed for the piecewise-implementation of RIGA of \cite{CODE_OCEAN_CONSTANT}, as it can be seen in the comparison with an implementation of GRAPE \cite{PerSilRou19}).

Let us state some comments about the robustness of the algorithms. For instance, both GRAPE and RIGA (in its piecewise-constant version)  converges to a stationary point $\overline{\mathfrak U}$ of the objective function $\Omega (\mathfrak U)$ that is, a point for which its gradient is null. This implies robustness with respect to variations of the components of control pulses represented by  ${\mathfrak U}$, including noises and every kind of errors due to the generation and transmission of control pulses. By similar reasons, one may show by a limit process that this robustness is shared by the smooth implementation of RIGA.

The authors are now studying the generalization of RIGA for open systems described by Lindblad equations. These will be the theme of future work.


\appendix

\section{Optimization problem}
 \label{aOpt}

 We shall consider three optimization problems, stated as follows (see Def. \ref{Prob1} for the definition of matrices $E$ and $F$):
 \begin{problem}
 \label{Problem_A}
 Given $X_f \in \Un$, find $X_{goal}^* \in \Un$ with the property that there exists some $\phi \in \RR$ such that $X_{goal}^* E = \exp (\jmath \phi) F$ in a way that $X_{goal}^*$ minimizes $\| X_{goal}^* E - X_f E\|$.
 \end{problem}

 \begin{problem}
 \label{Problem_B}
 Given $X_f \in \Un$, find $X_{goal}^* \in \Un$ with the property that there exists some $\phi \in \RR$ such that $X_{goal}^* E = \exp (\jmath \phi) F$ in a way that $X_{goal}^*$ minimizes $\| X_{goal}^*  - X_f \|$.
 \end{problem}
 
 Recall that  Prop. \ref{pCritical} of Section\ref{s:Refinements} introduces 
 the matrix $X_E \in \Un$  $\Un$ such that its first $\nbar$ columns forms the matrix E. 
 Since the Frobenius norm is invariant by left- and right-multiplication by a unitary 
 matrix, given some $X_{goal} \in \Un$ such that $X_{goal} E =  F$, we shall define $W_f = X_E^\dag X_f X_E$ and $W_f = X_E^\dag X_{goal} X_E$, and state the next optimization problem which relies
 in finding $W_g^o \in \Un$ that minimizes the Frobenius norm  $\| W_f - W_{g^o}\|$ in a way that $X_{goal}^* = X_E W_{g^o} X_E$ solves Problems \ref{Problem_A} and \ref{Problem_B}.  So the next problem will be solved first.
 
\begin{problem}
\label{pOptimal}
Given $W_f, W_g \in \Un$  with $W_f= [W_{f_1}~W_{f_2}]$ and $W_g= [W_{g_1}~W_{g_2}]$ where
$W_{f_1}$ and $W_{f_2}$ are respectively blocks of $\nbar$ rows and $n - \nbar$ columns, and
so are the blocks $W_{g_1}$ and $W_{g_2}$, find $\phi \in \RR$
and a $n \times (n-\nbar)$ complex matrix ${\overline W}_{g_2}$ such that
\begin{equation}
\label{eWgo}
W_g^o = [\exp(\jmath \phi) W_{g_1}~{\overline W}_{g_2}]
\end{equation}
and in a way that $\| W_f - W_g^o\|$ is minimal.
\end{problem}
Note that $W_{g_1} = W_g \Pi^\dag$ and $W_g^o \Pi^\dag$ coincides up to a selection of  a global phase.

\begin{theorem}
\label{tOpt}(Solution of Problem \ref{pOptimal})
Let $W_{11}= W_{g_1}^\dag W_{f_1}$. Let $w_{ii} = a_i + \jmath b_i$ the $i$th-element of the diagonal of $W_{11}$.
Then let $\alpha = \sum_{i=1}^{\nbar} a_i$ and $\beta=\sum_{i=1}^{\nbar} b_i$. Let $(\rho, \theta)$  be such that
$\rho \exp (\jmath \theta) = \beta - \jmath \alpha$. Let $\phi_1=\theta + \pi/2$ and $\phi_2=\theta - \pi/2$.
Let $W_{21} = W_{g_2}^\dag W_{f_2}$. Let $W_{21} = U \Sigma V^\dag$ be a singular value decomposition of
$W_{21}$. Let $H= U V^\dag$. Let ${\overline W}_{g_2}=  W_{g_2} H$. Then one of the two matrices $[\exp(\jmath \phi_i) W_{g_1}~{\overline W}_{g_2}], i=1,2$ solves the optimization problem \ref{pOptimal}.
\end{theorem}

 Let $W_f, W_g \in \Un$ be decomposed in blocks as in the statement of Theorem \ref{tOpt}.
 Then consider the next two optimization problems:
 \begin{problem} \label{pA}
 Find $H \in \mbox{U} (n- \nbar)$ such that $\| W_{g_2} H - W_{f_2} \|$ is minimal.
\end{problem}
\begin{problem} \label{pB}
Find $\phi \in (-\pi, \pi]$ such that $\| \exp(\jmath \phi) W_{g_1}  - W_{f_1} \|$ is minimal.
\end{problem}
Then one will show that:
\begin{proposition} \label{pEnd}
A solution of Problem \ref{pA} is given by $H$ described in the statement of Theorem \ref{tOpt}. A solution
of Problem \ref{pB} is given by some $\phi$ that is described by the statement of Theorem \ref{tOpt}.
\end{proposition}
Firstly, assume that Prop.  \ref{pEnd} holds. Under this assumption, one now proves Theorem \ref{tOpt}. Let ${\mathcal V}_1 \subset \CC^n$ be the subspace that is spanned by the columns of
$W_{g_1}$. It is easy to see that the collums of the matrices $W_{g_2}$ and ${\overline W}_{g_2}$
are both orthonormal bases of ${\mathcal V}_1^{\perp}$. In particular, there exists $H \in \mbox{U}(n - \nbar)$ such that ${\overline W}_{g_2} = W_{g_2} H$. Now note that:
\begin{eqnarray*}
\| W_g^o - W_f \|^2 & = & \| \exp(\jmath \phi) W_{g_1} - W_{f_1} \|^2 \\
                    & + &  \| W_{g_2} H  - W_{f_2} \|^2.
\end{eqnarray*}
As $\phi$ and $H$ may be chosen independently, it is then clear that the statement of Theorem \ref{tOpt} holds.
\begin{proof} (of Prop. \ref{pEnd})
Note first that
\begin{eqnarray*}
\| W_{g_2} H - W_{f_2} \|  & = & \| W_g^\dag[ W_{g_2} H - W_{f_2}] \|\\
 & = &  \left\| \left[
\begin{array}{c}
0\\
I_{n-\nbar}
\end{array}
\right]
 H - W_g^\dag W_{f_2} \right\|
\end{eqnarray*}
Assume that $W_g^\dag W_{f_2} =  \left[
\begin{array}{c}
W_{21}\\
W_{22}
\end{array}
\right]
$
Hence $\| W_{g_2} H - W_{f_2} \|^2 = \| H - W_{22}\|^2 + \|  W_{21}\|^2$. Hence, we have shown the following result:
\begin{proposition}
\label{pH}
Let $W_{22} = W_{g_2}^\dag W_{f_2}$.
Problem  \ref{pA} is then equivalent to find $H \in \mbox{U}(n - \nbar)$ that minimizes $\| H - W_{22}\|$.
\end{proposition}
This is a standard minimization problem solved by \cite[Theo. 3]{Kel75}. From that result it follows easily that if $W_{22} = U \Sigma V^\dag$ is a singular value decomposition, then $H= U V^\dag$ is a solution of the proposed  minimization problem.

Now, in order to show the second part of the theorem, note that, since $W_g$ is unitary, then $\| \exp(\jmath \phi) W_{g_1}  - W_{f_1} \| = \| \exp(\jmath \phi) W_g^\dag W_{g_1} - W_g^\dag W_{f_1}\| =
\left\| \exp(\jmath \phi) \left[ \begin{array}{c}
I_{\nbar}\\
0
\end{array}
\right]
- \left[ \begin{array}{c}
W_{11}\\
W_{21}
\end{array}
\right] \right\|
$.
Note that the square of the last norm is given by $\|W_{21}\|^2 + \| \exp(\jmath \phi) I_{\nbar} - W_{11}\|^2$. Since $W_{21}$ is fixed,
it is clear that to solve Problem \ref{pB} one must minimize $\| \exp(\jmath \phi) I_{\nbar} - W_{11}\|^2$.
  Only the elements of the diagonal of the last matrix depends on $\phi$.  Computing this
norm, the contribution of the diagonal to the Frobenius norm  is given by $\sum_{i=1}^{\nbar} \| \exp(\jmath \phi)  - w_{ii} \|^2$.
Denoting $w_{ii} = (a_i + \jmath b_i)$, one shows easily that the $\phi$-dependent part of this last sum is given by $L(\phi) = \alpha \cos \phi + \beta \sin \phi$. Hence the condition $\frac{\partial L}{\partial \phi} = -\alpha \sin \phi + \beta \cos \phi = 0$ implies that the vector $(\cos \phi, \sin \phi) \in \RR^2$ must be orthogonal to the vector $(\beta, \alpha) \in \RR^2$. If $\beta + \jmath \alpha = \rho \exp(\jmath \theta)$, then $\phi = \theta \pm \pi/2$, completing the proof of the proposition.
\end{proof}

Now we will show that one may obtain the solutions of both Problems \ref{Problem_A} and \ref{Problem_B} from the solution of Problem \ref{pOptimal}.
\begin{theorem} \label{tOpt2} Given $X_f \in \Un$, let $W_f = X_E^\dag X_f X_E$. Choose any $X_{goal}$  such that $X_{goal} E = F$. Let $W_g = X_E^\dag X_{goal} X_E = [ W_{g_1} \, W_{g_2} ]$, where
$W_{g_1}$ is the submatrix that is formed by the first $\nbar$ columns.  Then:\\
(i) Let $W_{g}^o =  = [ W_{g_1}^o \, W_{g_2}^o ]$ be a solution of Problem \ref{pOptimal} as stated in Theorem \ref{tOpt}. Then $X_{goal}^* = X_E W_{g}^o X_E^\dag$ is a solution
 of both Problems  \ref{Problem_A} and \ref{Problem_B}.\\
(ii) Conversely, assume $X_{goal}^*$ is a solution of Problem \ref{Problem_B}. Let 
$W_g^o = X_E^\dag X_{goal}^* X_E = [ W_{g_1}^o \, W_{g_2}^o ]$. Then  $W_{g_1}^o = \exp (\jmath \phi ) W_{g_1}$ is a solution of Problem \ref{pB} and $W_{g_2}^o = W_{g_2} H$, where $H$ is a solution of the problem of Prop. \ref{pH}. 
 \end{theorem}
 
 \begin{proof} To show part (i), we show first that
  $X_{goal}^*$ solves Problem \ref{Problem_A}, note that, by Prop. \ref{pEnd}, $W_g^o =[\exp(\jmath \phi) W_{g_1}~{\overline W}_{g_2}]$ solves Problem \ref{pB}. Now note that
 \[ 
 \|X_{goal}^*E - X_f E\| = \|X_E W_g^0 X_E^\dag E - X_E  W_f X_E^\dag E\|.
 \]
 As $X_E^\dag  E = \Pi^\dag$, then 
 $ 
 \|X_{goal}^*E - X_f E\| = \|X_E (W_g^0 - W_f) \Pi^\dag\| =$  $\|X_E (\exp(\jmath \phi) W_{g_1}- W_{f_1}) \Pi^\dag\| =$   $\|\exp(\jmath \phi) W_{g_1}- W_{f_1} \Pi^\dag\|$
 where the last equality is due to the invariance of the Frobenius norm with respect to left- and right-multiplication by an unitary matrix. In particular, this shows that $X_{goal}^*$ solves Problem \ref{Problem_A}.
 
 Now, to show that  $X_{goal}^*$ also solves Problem \ref{Problem_B}, note that 
 $ X_{goal}^* E = X_E X_g^0 X_E^\dag E = X_E X_g^0 \Pi^\dag = X_E \exp(\jmath \phi) W_{g_1} =  \exp(\jmath \phi) X_E W_{g}^o \Pi^\dag$. Now, as $X_{goal} E = F$, then $X_E X_g^0 X_E^\dag E = F$,
 and so $X_E X_g^0 \Pi^\dag = F$. Hence $ X_{goal}^* E = \exp(\jmath \phi) X_E W_{g}^o \Pi^\dag = \exp(\jmath \phi) F$.  Furthermore, $\|W_f - W_{g}^o\|$ is minimized by $W_{g}^o$, then 
 $\|W_f - W_{g}^o\| = \|X_E^\dag X_f X_E - X_E^\dag X_{goal}^* X_E\| = \|X_f - X_{goal}^* \|$. In particular, $\|X_f - X_{goal}^*\|$ is minimized by $X_{goal}^*$, and so it solves Problem \ref{Problem_B}. This shows part (i). The proof of part (ii) is analogous and is left to the reader. 
 \end{proof}
 
 \begin{theorem}
\label{tEigenOpt}
 Let $Y = X_f E$, $\mathcal Y = \Image Y$ and $\mathcal F = \Image F$. Assume
that $\dim (\mathcal Y + \mathcal F) = \nbar + k \leq 2 \nbar$. 
Let $X_{goal}^*$ be a solution of the previous optimization problem such that
$\widetilde X = (X^*_{goal})^\dag X_f$ admits an eigenvalue equal to one with multiplicity at least equal to $n - \nbar - k$.  
\end{theorem}

\begin{proof} Before proving the lemma, one states the following Lemma whose proof is left to the reader
\begin{lemma}
\label{lSt}
Assume that 
\[
A =\left[
\begin{array}{cc}
A_1 & 0\\
0        & I
\end{array}
\right]\] 
is a square complex matrix. Then a closest unitary matrix $H$  to $A$ is of the form
\[
H=\left[
\begin{array}{cc}
H_1 & 0\\
0        & I
\end{array}
\right]
\] 
where $H_1$ is a closest unitary matrix to $A_1$.
\end{lemma}

The proof of the Theorem \ref{tEigenOpt} relies on the choice of a particular $X_E = [E \, \overline E]$ and on a particular first choice of $X_{goal}$ such that $X_{goal} E = F$. Then the proof will be a simple application of part (ii) of Theorem \ref{tOpt2}  and of Lemma \ref{lSt}. Let us construct first $X_E$ and $X_{goal}$ before applying Theorem \ref{tOpt}.

Assume that $Y = [ y_1 \cdots y_{\nbar} ]$ where $y_j = X_f e_j \in \CC$ are its orthonormal column vectors.
Let $a_j \in \CC^n , j= 1, \ldots, k$ be such that $\{y_1, \ldots y_{\nbar}, a_1, \ldots, a_k\}$ is a orthonormal basis of  $\mathcal Y + \mathcal F$. In particular $\{X_f^\dag y_1, \ldots X_f^\dag y_{\nbar}, X_f^\dag a_1, \ldots, X_f^\dag a_k\} = \{e_1, \ldots, e_\nbar, b_1, \ldots, b_k\}$  is  a
orthonormal set that is mapped by $X_f$ onto an orthonormal basis of  $\mathcal Y + \mathcal F$.
Then we may complete this set to an ortonormal basis $\mathbb B =  \{e_1, \ldots, e_\nbar, b_1, \ldots, b_k, g_1, \ldots, g_{n-\nbar-k}\}$ of $\CC^n$. Let $X_E$ be the unitary matrix that is formed by the column vectors of $\mathbb B$ and $W_f =  X_E^\dag X_f X_E$
\[
W_f = [\overline Y A G]
\]
where $\overline Y, A, G$ are respectively blocks of size $n \times \nbar$, $n \times k$ and $n \times (n-\nbar-k)$. By construction $\mathcal Y + \mathcal F = \Image [\overline Y A]$ (the subspace is an intrinsic object, and the basis transformation $X_E$ only changes its representation). In particular
$\mathcal G = \Image G$ is a subspace that is orthogonal with respect to $\mathcal Y + \mathcal F$.

Let ${\overline f}_i = X_E^\dag f_i$. Note that $\{ {\overline f}_1, \ldots , {\overline f}_\nbar\}$ 
is a orthonormal basis $\mathcal F$ (after a basis transformation). 
Now let $\{c_1, \ldots, c_k\} \subset \CC$ be such that $\{{\overline f}_1, \ldots , {\overline f}_\nbar , c_1, \ldots, c_k\}$
is an orthonormal basis of $\mathcal Y + \mathcal F$. Let
\[
W_g = [\overline F C G]
\]
where $\overline F$ and $C$ are formed respectively by the column vectors $\{ {\overline f}_1, \ldots , {\overline f}_\nbar\}$ and
$\{c_1, \ldots, c_k\}$. Clear $W_g$ is a possible choice of an initial goal matrix $X_{goal}$ (since in the original basis
one has $X_{goal} \overline E = \overline F$). Let $W_{f_1} = \overline Y, W_{f_2} = [A, G]$, $W_{g_1} = \overline F$  and $W_{g_2} = [C G]$. Now, by part (ii) of Theorem \ref{tOpt2}, if $X_{goal}^*$ solves
Problem \ref{Problem_B}, then $W_g^o = X_E^\dag X_{goal}^* X_E = [ W_{g_1}^o \, W_{g_2}^o]$ is such that
$W_{g_2}^o = W_{g_2} H$, and $H$ 
minimizes $\| H - (W_{g_2})^\dag W_{f_2}\|$. Since the matrices $A$, $C$, and $G$ are blocks of unitary matrices and $\Image [\overline Y A] = \Image [\overline F C] = \mathcal Y + F $, it is easy to show that $C^\dag G =0$, $A^\dag G =0$ and $G^\dag G = I_{n-\nbar-k}$. In particular
\[
(W^{g_2})^\dag W_{f_2} = 
\left[
\begin{array}{cc}
C^\dag A & 0\\
0        & I_{n-\nbar-k}
\end{array}
\right]
\]
Then from Lemma \ref{lSt}, $H$ must be in the form
\[
H =
\left[
\begin{array}{cc}
H_1 & 0\\
0        & I_{n-\nbar-k}
\end{array}
\right]
\]
where $H_1$ is a closest unitary matrix to  $C^\dag A$. Then, by part (ii) of Theorem \ref{tOpt2}, the optimal solution will be of the form
$W_g^o = [\exp (\jmath \phi) W_{g_1} \, W_{g_2} H]$, which is given by:
\[
W_g^o =  \left[
\begin{array}{ccc}
\exp (\jmath \phi) \overline F & C H_1 & G
\end{array} \right]
\]
in particular it is easy to see that
\[
(W_g^o)^\dag W_f = \left[
\begin{array}{cc}
 R_1  & 0\\
0 & I_{n-\nbar-k}
\end{array} \right]
\]
for some $R_1 \in \mbox{U}(\nbar + k)$. This shows the desired property of the eigenvalues of error matrix $(W_g^o)^\dag W_f$
which coincides with the eigenvalues of $(X_{goal}^*)^\dag X_f = X_E (W_g^o)^\dag W_f X_E^\dag$.
\end{proof}

\section{Proof of Theorem \ref{tGRAPE_RIGA}}
\label{aGRAPE_RIGA}

Given the reference input $\overline{\mathfrak U} =  \left(  {\overline u}_{k_s} : k=1, \ldots, m, s=1, 1, \ldots, N_{sim} \right)
\in \RR^{m N_{sim}}$, the backward evolution of the system is defined by:
\begin{equation}
\label{eXbars}
\left\{
\begin{array}{l}
{\overline X}_{N_{sim}} = X_{goal}\\
{\overline X}_{s-1}   = \exp \left[ -\delta  (S_0 + \sum_{k=1}^m {\overline u}_{k_s}^\ell S_k)\right] {\overline X}_{s},\\
 s=1, \ldots, N_{sim}
\end{array} \right.
\end{equation}

The following result will show that the algorithm described in the statement of Theorem
\ref{tGRAPE_RIGA} (that is obtained after replacement of ${\overline u}_{k_s}^\ell$ by
${\overline u}_{k_s}^{\ell-1}$ in operation $\sharp 3^\prime$ of that implementation)  is,
up to a first order approximation, an implementation of the first order of GRAPE.
\begin{proposition}
A first order approximation of each component of the gradient of the objective function is given by:
\[
\left. \frac{\partial \Omega ({\mathfrak U})}{\partial {u}_{k_s}} \right|_{\overline{\mathfrak U}}
 = \delta [{\nabla V}_{{\widetilde X}_{s-1}} \cdot ({\widetilde S}_{k_{s-1}} {\widetilde X}_{s-1}) + \frac{\mathcal O(\delta)}{\delta}]
 \]
where ${\widetilde S}_{k_{s-1}} ={X}_{s-1}^\dag S_k X_{s-1}$ and  ${\widetilde X}_{s-1} = {\overline X}_{s-1}^\dag X_{s-1}$,
for $s=1, 2, \ldots, N_{sim}$ where ${\overline X}_{s-1}$ and $X_{s-1}$ are obtained  respectively by \refeq{eXbars} and \refeq{eXs}.
As $\lim_{\delta \rightarrow 0} \frac{\mathcal O(\delta)}{\delta} = 0$, it follows that, with a gain $\mathcal K = \frac{ K}{\delta}$, a first order aproximation
of $\mathcal K \left. \frac{\partial \Omega}{\partial { u}_{k_s}} \right|_{\overline{\mathfrak U}}$ is given by $K [{\nabla_{{\widetilde X}_{s-1}} \mathcal V} \cdot ({\widetilde S}_{k_{s-1}} {\widetilde X}_{s-1})]$,
which coincides with the feedback law \refeq{eUtilk}.
\end{proposition}

The following notations will be used along the proofs of this Appendix
\begin{definition}
Denote ${\Sigma}_{s} = \left(S_0 + \sum_{k=1}^{m} { u}_{k_s} S_k \right)$
and ${\overline {\Sigma}}_{s} = \left(S_0 + \sum_{k=1}^{m} {\overline  u}_{k_s} S_k \right)$.
Let ${\widetilde {\Sigma}}_{s}^k = {\overline{\Sigma}}_{s} + {\widetilde  u}_{k_s} S_k$.
Consider ${\overline {\mathfrak U}} =  \left(  {\overline u}_{k_s} : k=1, \ldots, m, s=1, \ldots, N_{sim} \right)
\in \RR^{m N_{sim}}$. Let $\left\{ {\overline {\mathfrak U}} \right\}_{k_s} = {\overline u}_{k_s}$ be the coordinate function.
Denote ${{\mathfrak U}}^{k_s} \in \RR^{m N_{sim}}$ be such that $ {\mathfrak U}^{k_s}$ is equal to
${\overline {\mathfrak U}}$ with the exception of the component $k_s$, which is given by
$\left\{{ \overline{\mathfrak U}}^{k_s} \right\}_{k_s} = {\overline u}_{k_s} + {\widetilde u}_{k_s}$.
\end{definition}

The proof of the Proposition is based on the following Lemma:
\begin{lemma} The following identity holds
\label{lll}
\[
\begin{array}{l}
\frac{\partial \exp(\delta {\widetilde {\Sigma}}_s^k)}{\partial  {\widetilde u}_{k_s}} =
\lim_{{\widetilde u}_{k_s}\rightarrow 0} \frac{\exp[\delta ({\overline{\Sigma}}_{s} + {\widetilde u}_{k_s} S_k)]}{{\widetilde u}_{k_s}} \\
= \delta \exp [\delta {\overline {\Sigma}}_s]  S_k + \mathcal O(\delta).
\end{array}
\]
\end{lemma}

\begin{proof}
In \cite{Hig08} it is shown that $\frac{ \partial \exp(S+ \mathcal \tau  T)}{\partial \tau}|_{\tau=0} = \int_0^1 \exp[(1-s)S] T \exp(s S) ds$.
Hence, taking $S= \delta {\overline{\Sigma}}_{s}$, $\tau =  {\widetilde u}_{k_s}$  and $T = \delta S_k$, one obtains
$\frac{ \partial \exp(S+ \mathcal \tau  T)}{\partial \tau}|_{\tau=0} = $ $\int_0^1 \exp[(1-s) \delta {\overline{\Sigma}}_{s}] \delta S_k \exp(s \delta {\overline{\Sigma}}_{s}) ds$.
So, defining
$z = \delta s$, one obtains
$\int_0^1 \exp[(1-s) \delta {\overline{\Sigma}}_{s}] \delta S_k \exp(s \delta{\overline{\Sigma}}_{s}) ds =$
$\int_0^\delta  \exp[(1-\frac{z}{\delta}) \delta {\overline{\Sigma}}_{s}]  S_k \exp(\frac{z}{\delta} \delta{\overline{\Sigma}}_{s}) dz =$
$ \exp(\delta {\overline{\Sigma}}_{s}) \int_0^\delta  \exp[-z  {\overline{\Sigma}}_{s}] S_k \exp(z {\overline{\Sigma}}_{s}) dz$.

Now, if $\phi (\delta) =  \int_0^\delta  \exp[-z  {\overline{\Sigma}}_{s}] S_k \exp(z {\overline{\Sigma}}_{s}) dz$, it is clear that
$\phi(\delta)=  \phi(0) + \delta \phi^\prime(0)  + \mathcal O (\delta) =$  $0 + \delta S_k + \mathcal O (\delta)$.
Hence, the desired result follows.
\end{proof}

In order to show the proposition, let us compute first
\begin{equation}
\label{etoile}
\left. \frac{\partial X_{goal}^\dag \mathcal X({\mathfrak U})}{\partial { u}_{k_s}} \right|_{\overline{\mathfrak U}} =
\lim_{{\widetilde u}_{k_s}\rightarrow 0} \frac{ X_{goal}^\dag [\mathcal X({\overline{\mathfrak U}}^{k_s}) - \mathcal X({\overline{\mathfrak U}})]}{{\widetilde u}_{k_s}}
 \end{equation}
From \refeq{eXbars}, it follows that
\begin{small}
\[
X_{goal}^\dag = {\overline X}_{s-1}^\dag \exp(-\delta {\overline{\Sigma}}_{s})  \exp(\delta {\overline{\Sigma}}_{s+1})^\dag \cdots \exp(\delta {\overline{\Sigma}}_{N_{sim}-1})^\dag.
\]
\end{small}
Now, from \refeq{eXs}, it follows that
\begin{small}
\[
\mathcal X ({\overline {\mathfrak U}}) = \exp(\delta {\overline {\Sigma}}_{N_{sim}-1}) \cdots \exp(\delta {\overline {\Sigma}}_{s+1}) \exp[\delta ({\overline {\Sigma}}_{s})] {\overline X}_{s-1}.
\]
\end{small}
and
\begin{small}
\[
\mathcal X ({\overline {\mathfrak U}}^{k_s})
= \exp(\delta \overline {\Sigma}_{N_{sim}-1}) \cdots \exp(\delta {\overline {\Sigma}}_{s+1}) \exp[\delta ({\widetilde {\Sigma}}_{s})] {\overline X}_{s-1}.
\]
\end{small}
So $X_{goal}^\dag \mathcal X(\overline{\mathfrak U}) =$ ${\overline X}_{s-1}^\dag X_{s-1} =$ ${\widetilde X}_{s-1}$ and
 $X_{goal}^\dag \mathcal X({\overline{\mathfrak U}}^{ks}) = {\overline X}_{s-1}^\dag \exp(-\delta  {\overline{\Sigma}}_s) \exp [\delta {\widetilde \Sigma}_{s}] X_{s-1}$

From \refeq{etoile}, it follows that
\[
\left. \frac{\partial X_{goal}^\dag \mathcal X({\mathfrak U})}{\partial { u}_{k_s}} \right|_{\overline{\mathfrak U}} = {\overline X}_{s-1}^\dag
\exp(-\delta  {\overline{\Sigma}}_s ) \frac{\partial \exp (\delta {\widetilde {\Sigma}}_s^k)}{\partial  {\widetilde u}_{k_s}}
X_{s-1}
 \]
By Lemma \ref{lll},  from the fact that ${\overline X}_{s-1}^\dag S_k X_{s-1} = {\widetilde S}_{k_s} {\widetilde X}_{s-1}$ and ${\widetilde X}_{s-1} = {\overline X}_{s-1}^\dag X_{s-1}$, it follows that
\[
\left. \frac{\partial X_{goal}^\dag \mathcal X({\mathfrak U})}{\partial {u}_{k_s}} \right|_{\overline{\mathfrak U}} = \delta {\widetilde S}_{k_s} {\widetilde X}_{s-1} + \mathcal O(\delta)
\]
By the chain rule, one may write
\[
\left. \frac{\partial \Omega ({\mathfrak U})}{\partial {u}_{k_s}} \right|_{\overline{\mathfrak U}} =
\nabla_{{X_{goal}^\dag \mathcal X(\overline{\mathfrak U})}} {\mathcal V} \cdot \left. \frac{\partial X_{goal}^\dag \mathcal X({\mathfrak U})}{\partial { u}_{k_s}} \right|_{\overline{\mathfrak U}}
\]
Then
\[
\left. \frac{\partial \Omega ({\mathfrak U})}{\partial { u}_{k_s}} \right|_{\overline{\mathfrak U}} = \delta \nabla_{{\widetilde X}_{s-1}} {\mathcal V} \cdot ({\widetilde S}_{k_s} {\widetilde X}_{s-1}) + \mathcal O(\delta)
\]
showing the desired result.

\section{Some properties of the Frobenius norm}

\begin{proposition}
 \label{ineqFrob}
Assume that $E$ is a complex matrix whose columns form an orthonormal set. Let $U \in \Un$. \\
(a) If $A$ is a complex matrix such that $A U$ is well defined,
then $\| A U \| = \|A\|$.\\
(b) If $A$ is a complex matrix such that $U A$ is well defined,
then $\| U A \| = \|A\|$.\\
(c) If $A$ is a complex matrix such that $A E$ is well defined,
then $\| A E \| \leq \|A\|$.\\
(d) If $A$ is a complex matrix such that $E^\dag A$ is well defined,
then $\| E^\dag A \| \leq \|A\|$. In both cases, if $E \in \Un$
then the equality is attained.
\end{proposition}

 \begin{proof}
 (a) Recall that the Frobenius norm $\|A\|^2 = \trace [A^\dag A] = \trace[ A A^\dag]$.
  Then if $U \in \Un$ , $\| A U \|^2 = \trace[ A U U^\dag A^\dag\| = \|A\|^2$.  \\
 (c) Note that, by the Gram-Schmidt algorithm, one may complete the matrix
 $E$ to some $U = [E \; \widehat E] \in \Un$. Since the square of the Frobenius  norm of a matrix
 is the sum of the square of complex norm of all elements of this matrix, and  as $A E$ is a submatrix
 of $A U = [ A E \; A \widehat E]$, then it is clear that $\| A E \|^2 \leq \| A U\|^2 = \|A\|^2$.
 The proof of (b) and (d) are analogous.
\end{proof}



\section{Is the Lyapunov functon  nonincreasing along the steps of RIGA?}
\label{sNonIncreasing}

If no policy of avoiding singular and or critical points are applied, the the Lyapunov function would be monotonic along the steps of RIGA, as shown by the next proposition. However, these policies may produce a non-monotonic situation, as it was  observed  in some  numerical examples.   
\begin{proposition}
 \label{pMonotonous} 
 Recall that $\pdist(X_1, X_2) = \| X_1 R - X_2 E\|$. Then,
 the following properties holds for the steps $\ell=1, 2, \ldots$ of RIGA:\\
(a)   ${\overline X}^{\ell-1} (T_f) = \Xgoalbar$.\\
(b)  ${\widetilde X}^\ell (0) = \left(R^\ell\right)^\dag = \left( \Xgoalbar \right)^\dag  X_f^{\ell-1}$.\\
(c)  ${\widetilde X}^{\ell-1}(T_f) = \left( \Xgoalbarellminusone \right)^\dag X_f^{\ell-1}$.\\
(d)  $\pdist(X^\ell(0), \overline X^{\ell-1}(0) )  =  \pdist (X_f^{\ell-1}, \Xgoalbar)$.\\
(e) If $\Xgoalbarellminusone E = \Xgoalbar E$, then  $\mathcal V ({\widetilde X}^\ell (0)) = $
 $\mathcal V ({\widetilde X}^{\ell-1} (T_f))$. In particular, $\mathcal V ({\widetilde X}^{\ell} (T_f))
 \leq \mathcal V ({\widetilde X}^{\ell-1} (T_f))$.  As the Lyapunov function
 is non-increasing inside a step of RIGA, this implies that the Lyapunov
function is non-increasing along the steps of RIGA, at least while $\Xgoalbarellminusone E = \Xgoalbar E$.\\
(f) If $\Xgoalbar$ is constructed by the optimization process of Appendix \ref{aOpt}, then $\mathcal V ({\widetilde X}^{\ell} (0))
 \leq \mathcal V ({\widetilde X}^{\ell-1} (T_f))$.
\end{proposition}

\begin{proof}

{\noindent}(a) Note that $X^{\ell-1}(t)$ is obtained in operation $\sharp 2$ of RIGA from the integration of system \refeq{cqs} with input ${\overline u}^{\ell-1}$ with $X(0) = I$.
Then, in operation $\sharp 5$, one gets ${\overline X}^{\ell-1} (t) =  X^{\ell-1}(t) R_{\ell}$. By right-invariance, ${\overline X}^{\ell-1} (t)$ is the solution of \refeq{cqs} with
$X(0) = R_{\ell}$. Then ${\overline X}^{\ell-1} (T_f) =   X^{\ell-1}(T_f) X_f^{\ell-1} \Xgoalbar =$ $\Xgoalbar$.\\
(b) Recall that  the reference trajectory of step $\ell$ is ${\overline X}^{\ell-1} (t)$. As $X(0) = I$, then ${\widetilde X}^\ell(0) =  {\overline X}^{\ell-1}(0)^\dag X^\ell(0) =$ $({\overline X}^{\ell-1}(0) R^{\ell})^\dag =$
$R_\ell^\dag = (\Xgoalbar)^\dag X_f^{\ell-1}$.\\
(c) As ${\widetilde X}^{\ell-1}(T_f) =  {\overline X}^{\ell-2}(T_f)^\dag X^{\ell-1}(T_f)$, from part (a), ${\widetilde X}^{\ell-1}(T_f) = (\Xgoalbarellminusone)^\dag X_f^{\ell-1}$.\\
(d) From Definition \ref{dPartial}, from (b) and part (b) of Prop. \ref{pCritical}, it follows that $ \pdist(X^{\ell}(0), {\overline X}^{\ell-1}(0))^2 =$
$ \mathcal V({\widetilde X}^\ell(0)) =   \mathcal V( {\Xgoalbar}^\dag X_f^{\ell-1})= $
$\pdist(\Xgoalbar,X_f^{\ell-1})^2$.\\
(e) By Prop. \ref{pCritical}, since $\Xgoalbar E = \Xgoalbarellminusone E$, then  $\pdist(X_f^{\ell-1}, \Xgoalbarellminusone ) = \| X_f^{\ell-1} E -  \Xgoalbarellminusone E\| =$
$\| X_f^{\ell-1} E -  \Xgoalbar E\| =$  $\pdist(X_f^{\ell-1}, \Xgoalbar)$.
From (b), and (a), then $\mathcal V({\widetilde X}^\ell(0)) = \pdist(X^{\ell}(0), {\overline X}^{\ell-1}(0))^2 =$
$\pdist(X_f^{\ell-1}, \Xgoalbar)^2 =$ $\pdist(X_f^{\ell-1}, \Xgoalbarellminusone )^2 =$ $\pdist(X^{\ell-1}(T_f), {\overline X}^{\ell-1}(T_f))^2 =$
 $\mathcal V({\widetilde X}^{\ell-1}(0))$.\\
 (f) Note first that, from the previous statements, $\mathcal V({\widetilde X}^\ell(0)) = \| \Xgoalbar E - X_f^{\ell-1}\|^2$ and 
 $\mathcal V({\widetilde X}^{\ell-1}(T_f)) = \| \Xgoalbarellminusone E - X_f^{\ell-1}\|^2$. Now, as
 stated in Appendix \ref{aOpt}, one has that $\Xgoalbar$ minimizes $\| \Xgoalbar E - X_f^{\ell-1}\|$ with the restriction that $\Xgoalbar E = \exp{\jmath \phi} F$ for some $\phi \in \RR$. Then the result follows.
\end{proof}

\section{Proof of Proposition \ref{pCritical}}
\label{aCritical}

\subsection{Proof of Prop. \ref{pCritical}}

\begin{proof}\\
(a) $\|(\widetilde X - I) E\|^2 =$ $\trace[  E^\dag (\widetilde X - I)^\dag(\widetilde X - I) E  ]$
$\trace[  E^\dag ( 2 I  - {\widetilde X}^\dag  - \widetilde X)^\dag E  ] =$
$ 2 \trace[  E^\dag E] - \trace [ \left\{E^\dag {\widetilde X} E \right\} + \left\{E^\dag {\widetilde X} E \right\}^\dag   ] =$
$ 2 \trace[ I_\nbar] - 2 \Re \trace [ \left\{E^\dag {\widetilde X} E \right\} ]$.\\
(b) By Prop. \ref{ineqFrob}, $\|( {\widetilde X} - I)E\| =$ $\| {\overline X}^\dag ( X - \overline X) E\| =$
$\| ( X - \overline X) E\|$.\\
(c) Consequence of part (c) of Lemma \ref{lCritical}.\\
(d) By Prop. \ref{ineqFrob}, $\mathcal V(\widetilde X) = \| ({\widetilde X} - I)E\|^2 =$ $\| X_E^\dag({\widetilde X} - I)E\|^2=$
$\| \left\{ [X_E^\dag {\widetilde X} X_E] X_E^\dag - X_E^\dag \right\}E\|^2=$
$\| (\widetilde W - I) X_E^\dag E\|^2=$
$\| (\widetilde W - I) \Pi^\dag\|^2$.
Now, from (a) with $E=\Pi^\dag$ and $\widetilde X = \widetilde W$, one gets
$\| (\widetilde W - I) \Pi^\dag\|^2= 2 \nbar - 2 \Re[ \trace( \Pi \widetilde W \Pi^\dag)]$.

\end{proof}

\subsection{Critical points of the partial trace}

Let $\widetilde X \in \Un$. Let $\widetilde W = X_E^\dag \widetilde X X_E$ as  in Prop. \ref{pCritical}.
It is clear, that in the new basis induced by $X_E$, a tangent vector of $\Un$ will be of the form
$\sigma \widetilde W$ where $\sigma \in \un$.  Then a critical point $\widetilde X$ of the partial trace will be any point $\widetilde W$
such that $\Re \left[ \trace \left( \Pi \sigma \widetilde  W \Pi^\dag\right) \right] = 0, \forall \sigma \in \un$, where $\Pi$ is defined in Prop. \ref{pCritical}.
The following result is a little bit more general than the case for which $\xi = \widetilde W \Pi^\dag$, with $\widetilde W \in \Un$.

\begin{lemma} \label{lCritical}
Let $\xi = \left[  \begin{array}{c} \xi_1 \\ \xi_2 \end{array} \right]$ be a complex
 $n \times \nbar$ matrix where
$\xi_1$ is a $\nbar \times \nbar$ block. One says that
$\xi_1$ is diagonalizabe if $\xi_1 =  V^\dag D V$, where $V \in \Unbar$ and $D$ is a diagonal matrix.
One says that $\xi$ is a critical point of the partial trace if $\Re [\trace ( \Pi \sigma \xi)] = 0$ for all
$\sigma \in \un$.
\begin{equation}
\label{eCritical}
\end{equation}

The following affirmations holds:
\begin{itemize}
\item[(a)] If $\xi$ is a critical point of the partial trace, then $\xi_2 =0$.
\item[(b)] If $\xi$ is a critical point of the partial trace and $\xi_1$ is diagonalizable, then $D$ is a real matrix.
\item[(c)] If $\xi_2 =0$ and $\xi$ is diagonalizable with a real matrix $D$, then $\xi$ is a critical point of the partial trace.
\item[(d)] Let $\widetilde W \in \Un$. Let $\xi = \left[  \begin{array}{c} \xi_1 \\ \xi_2 \end{array} \right] = \widetilde W \Pi^\dag$,
whre $\xi_1 $ is a square $\nbar \times \nbar$ block.
Then $\widetilde W$ is a critical point of the partial trace if and only if $\xi_2 =0$,
            $\xi_1$ is diagonalizable, and the entries of the diagonal of $D$ are of the form $d_{ii} = \pm 1$, $i=1, \ldots, \nbar$.
\end{itemize}
\end{lemma}

\begin{proof}
     (a) Consider the basis of $\un$ given by
    \[
    \mathbb{B}_n = \{ J_{kp}, R_{kp}, D_k : k, p \in \{1, 2, \ldots, n\}, k < p < n \}
    \]
    where these matrices are defined by
    \begin{eqnarray*}
    \left\{ J_{kp} \right\}_{i\ell} & = &
    \left\{ \begin{array}{l}
     \jmath, \mbox{if}~(i,\ell) = (k,p),\\
     \jmath, \mbox{if}~(i,\ell) = (p,k),\\
     0, \; \mbox{otherwise}.
     \end{array}\right.\\
     \left\{ R_{kp} \right\}_{i\ell} & = &
    \left\{ \begin{array}{l}
     +1, \mbox{if}~(i,\ell) = (k,p),\\
     -1, \mbox{if}~(i,\ell) = (p,k),\\
     0, \; \mbox{otherwise}.
     \end{array}\right.\\
      \left\{ D_{k} \right\}_{i\ell} & = &
    \left\{ \begin{array}{l}
     \jmath, \mbox{if}~i=\ell=k,\\
     0, \; \mbox{otherwise}.
     \end{array}\right.
    \end{eqnarray*}
    Analogously, one may define the basis $\mathbb{B}_{\nbar}$ of $\unbar$ given by
    \[
    \mathbb{B}_{\nbar} = \{ {\overline J}_{kp}, {\overline R}_{kp}, {\overline D}_k : k, p \in \{1, 2, \ldots, \nbar\}, k < p < n \}
    \]
    To show (a), note that for $i \in \{1, 2, \ldots, n\}$ and $\ell \in \{1, 2, \ldots, \nbar\}$, since $\left\{ J_{kp} \xi \right\}_{i \ell} = \sum_{s=1}^{n} \left\{ J_{kp}  \right\}_{i s} \xi_{s \ell}$, one has:
    \begin{eqnarray*}
    \left\{ J_{kp} \xi \right\}_{i\ell} & = &
    \left\{ \begin{array}{l}
     0, \mbox{if}~ i \not\in \left\{ k, p \right\}),\\
     \jmath \xi_{p\ell}, \mbox{if}~ i=k,\\
     \jmath \xi_{k\ell}, \mbox{if}~ i=p
     \end{array}\right.
     \end{eqnarray*}
     In particular, for $i \in  \{1, 2, \ldots, \nbar\}$, it follows that
    \begin{eqnarray}
    \label{eBB}
    \left\{ J_{kp} \xi \right\}_{i i} & = &
    \left\{ \begin{array}{l}
     0, \mbox{if}~ i \not\in \left\{ k, p \right\}),\\
     \jmath \xi_{p k}, \mbox{if}~ i=k,\\
     \jmath \xi_{k p}, \mbox{if}~ i=p
     \end{array}\right.
     \end{eqnarray}
    From this, one concludes that, for $k \in \{1, 2, \ldots, \nbar\}$, $p \in \{\nbar+1, \nbar+2, \ldots, n\}$ and
    $i \in  \{1, 2, \ldots, \nbar\}$ one has
    \begin{eqnarray}
    \left\{ J_{kp} \xi \right\}_{i i} & = &
    \left\{ \begin{array}{l}
     0, \mbox{if}~ i \neq k,\\
     \jmath \xi_{p k}, \mbox{if}~ i=k
     \end{array}\right.
     \end{eqnarray}
      Hence, $\Re \left[ \trace \left( \Pi J_{kp} \xi \right) \right] = 0$ implies that $- \Image \left( \xi_{p k} \right) =0$, for $k \in \{1, 2, \ldots, \nbar\}$ and $p \in \{\nbar+1, \nbar+2, \ldots, n\}$. Proceeding in a similar way with the matrices $R_{kp}$, one shows that   $ \Re \left( \xi_{p k} \right) =0$, for $k \in \{1, 2, \ldots, \nbar\}$ and $p \in \{\nbar+1, \nbar+2, \ldots, n\}$. This proves that $\xi_2$ is the null matrix, showing (a).\\
      (b) Assume that $\xi_1$ is diagonalizable
      with $ \xi_1 = V^\dag D V$, where $D$ is a complex diagonal matrix and $V \in \Unbar$. Assume that$\sigma \in \un$ is of the form:
     \[
     \sigma = \left[
     \begin{array}{cc}
     V^\dag {\overline D}_i V & 0\\
     0 & 0
     \end{array}
     \right]
     \]
     then $\Pi \sigma \xi = V^\dag {\overline D}_i D V$. Since the trace is an invariant with respect to a basis change, then it is clear that $\Re \left[ \trace \left( \Pi \sigma \xi \right) \right] = \Re \left[ \trace \left( {\overline D}_i D\right) \right] =0$. It follows easily that $D$ is a real matrix, showing (b).\\
      (c) Assume that $\xi_1 = V^\dag D V$ is diagonalizable as above. Let $U \in \Un$ given by
     \[
     U = \left[
     \begin{array}{cc}
     V & 0\\
     0 & I_{n-\nbar}
     \end{array}
     \right].
     \]
     Let $\sigma = U^\dag \widehat \sigma U$ where
     \[
     \widehat \sigma = \left[
     \begin{array}{cc}
     \sigma_{11}  & \sigma_{12} \\
     \sigma_{21}  & \sigma_{22}
     \end{array}
     \right].
     \]
     As $\xi_2=0$ simple computations shows that $\Pi U^\dag \widehat \sigma U \xi= V^\dag \sigma_{11} D V$. In particular $\Re \left[ \trace \left( \Pi \sigma \xi \right) \right] = \Re \left[ \trace \left( \sigma_{11} D \right) \right]$. Hence it suffices to show that $\Re \left[ \trace \left( \sigma_{11} D \right) \right] = 0$ for all $\sigma_{11} \in \unbar$. A similar reasoning that was used to obtain
     \refeq{eBB} may show that, if $\sigma_{11}= {\overline  J}_{kp}$ with $k \in \{1, 2, \ldots , \nbar\}$ and
     $p \in \{k+1, k+2, \ldots , \nbar\}$
     \begin{eqnarray*}
    \left\{ {\overline J}_{kp} D\right\}_{i i} & = &
    \left\{ \begin{array}{l}
     0, \mbox{if}~ i \not\in \left\{ k, p \right\}),\\
     \jmath D_{p k}, \mbox{if}~ i=k,\\
     \jmath D_{k p}, \mbox{if}~ i=p
     \end{array}\right.
     \end{eqnarray*}
     As $p > k$ and $D$ is a diagonal matrix, one concludes that $\Re \left[ \trace \left( {\overline J}_{kp} D \right) \right] = 0$ for $k \in \{1, 2, \ldots , \nbar\}$ and
     $p \in \{k+1, k+2, \ldots , \nbar\}$. An analogous  reasoning shows that this is true when replacing ${\overline J}_{kp}$ by ${\overline R}_{kp}$. In a similar way one shows that
     \begin{eqnarray*}
    \left\{ {\overline D}_{k} D \right\}_{i i} & = &
    \left\{ \begin{array}{l}
     0, \mbox{if}~ i \neq k,\\
     \jmath D_{i i}, \mbox{if}~ i=k
     \end{array}\right.
     \end{eqnarray*}
     As $D$ is a real matrix, it follows easily that $\Re \left[ \trace \left(D_k D\right) \right] =0$
     for $k = \{1, 2, \dots ,\nbar\}$. This completes shows (c).\\
     (d) From (a), note that $\xi_2 =0$ implies that $\xi_1 \in \Unbar$. So $\xi_1$ is diagonalizable, and the result
     follows easily from (c) and (d).
   \end{proof}

\nocite{Kha02}

\section{Avoiding critical and/or singular points (second strategy)}
\label{sSecondStrategy}

Recall that the strategy of section \ref{sFirstStrategy} works well in numerical experiments of RIGA, but convergence proofs for RIGA equipped with this strategy are very difficult to obtain. Note that, convergence proofs
can be obtained with an strategy that is similar\footnote{With the difference that no optimization is done in \cite{PerSilRou19}.} to the strategy of \cite[Algorithms B and C]{PerSilRou19}, which is the subject of this Section.

 The idea is to choose a goal matrix $\Xgoalbar$ in each step $\ell$ among the elements of a finite pre-computed set in a way that $\pdist (\Xgoalbar, X_f^{\ell-1}) < \alpha$.
So, when $n < \nbar$, Prop. \ref{pCritical} assures that this strategy avoids critical points, essentially because nontrivial critical points occur only for partial distances $\sqrt{4 k}$, $k=1, 2, 3, \ldots$.
Hence it suffices to take $\alpha < 2$, when $\nbar < n$. When $n= \nbar$, $\alpha$ may be any positive real number. The second strategy is a little bit different when
$n=\nbar$. In this case one replaces $\pdist (\cdot, \cdot)$ by $\dist(\cdot, \cdot)$, where $\dist (X_1, X_2) = \mathcal V(X_1^\dag X_2)$ (see\cite{PerSilRou19}).
The second strategy consists in a pre-computation of a set of goal matrices, done in the first step only, and a choice of the goal matrix among the elements of this set, done at every step $\ell$ of RIGA. It is summarized as follows:\\
(a) \textbf{Pre-Computation, done in step $\ell=1$ only).}
Let $X_{goal} \in \Un$ such that $X_{goal} E = F$. Choose  $\alpha > 0$ (note that $\alpha$ must be such that
$\alpha <2$ when $\nbar < n$). Using the algorithm below, construct a set   $\mathcal G =\{X_g^q \in \Un: q=0, 1, \ldots, p\}$ of goal matrices
with the following properties.
\begin{itemize}
  \item  $X_g^0 = X_f^{0}$, where $X_f^0$ is obtained in operation $\sharp 2$ in the first step  of RIGA.

 \item The partial distance $\pdist(X_g^{q-1}, X_g^{q}) \leq \alpha$
 for $q = \{1, 2, \ldots, p\}$.

  \item $X_g^p = X_{goal}^*$, where $X_{goal}^* = \optgoal ( X_{goal}, X_f^0)$.
  The final goal matrix $W_g^p$ is given by $X_{goal}^*$.
 \end{itemize}

 \textbf{(b) Choosing $\Xgoalbar$ in a step $\ell$ --- Switching Policy.} Let $\beta$ such that $0  < \alpha < \beta < 2$. Inicialize $q_0 = 1$. Choose the maximum
 $q^* \in \{1, 2, \ldots, p\}$ such that $q^* \geq q^{\ell-1}$ and $\pdist (X_g^{q^*}, X_f^{\ell-1}) \leq \beta$.
 Then take $q_\ell = q^*$ and $\Xgoalbar = X_g^{q^*}$. \textbf{When $q_\ell > q_{\ell-1}$, one says that a switching has occurred
 in step $\ell$}.

\begin{remark}
Convergence results of RIGA with this Switching policy for the case where $n$ coincides with $\nbar$ are provided in \cite{PerSilRou19}.
In this  paper we generalize this result for the case where $\nbar < n$.
Similarly to the results of that paper,  several switchings may occur along the steps of RIGA, until
 $q_\ell$ reaches $p$, for $\ell \geq \ell^*$ that is, $\Xgoalbar$ will be equal to $X_{goal}^*$ for all $\ell \geq \ell^*$, with $X_{goal}^* E = F$.
 An algorithm for computing the set
 of $\mathcal G$ of goal matrices is presented in the sequel.
\end{remark}

Now we present the Algorithm for the construction of the set of goal matrices $X_g^q, q=1, \ldots, p$ for a fixed $p$.
Let $R = U^\dag \diag [ \exp(\jmath \theta_1 ), \exp(\jmath \theta_2 ), \ldots , \exp(\jmath \theta_n ) ] U$ be the eigenstructure\footnote{The Schur decomposition is indicated for this computation due to its
numerical stability.} of a unitary matrix $R$, where $U$ is unitary and $\theta_i \in (-\pi, \pi]$.
Define $\Sigma=\sqrt[p]{R} = U^\dag \diag \left[ \exp \left(\jmath \frac{\theta_1}{p} \right), \exp \left(\jmath\frac{\theta_2}{p} \right), \ldots , \exp \left(\jmath \frac{\theta_2}{p} \right) \right] U$. By definition, $\Sigma^p=R$ and $|\theta_i / p |\leq \pi/p$, $i=1, \ldots, n$. From Prop. \ref{ineqFrob},
it follows that $\| \Sigma - I\| = \|  U^\dag \diag [ \exp(\jmath \frac{\theta_1}{p} ), \exp(\jmath\frac{\theta_2}{p}), \ldots,  \exp(\jmath\frac{\theta_n}{p})] U - U^\dag U] =
\| \diag [ \exp(\jmath \frac{\theta_1}{p} ), \exp(\jmath\frac{\theta_2}{p} ), \ldots , \exp(\jmath \frac{\theta_n}{p} ) ] -I\|
$.
Then it is clear that $\lim_{p \rightarrow \infty} \| \sqrt[p]{R} - I\| = 0$.

Now let $R=(X_f^\ell)^\dag X_{goal}^*$. Fix $p \in \NN$ and define $\Sigma = \sqrt[p]{R}$. Fix $\alpha  > 0$. So there exists
$p=p^*$ large enough such that $\| (\Sigma - I) E\| \leq  \alpha$, where $\Sigma =  \sqrt[p]{R}$. Let $X_g^q = (X_f^\ell)^\dag \Sigma^q$. By Prop. \ref{ineqFrob}
$\pdist(X_g^{q}, X_g^{q-1}) = \| [(X_f^\ell)^\dag \Sigma^{q} - (X_f^\ell)^\dag \Sigma^{q-1}] E \| =$  $\| (X_f^\ell)^\dag \Sigma^{q-1}[ \Sigma - I] E \| =$
$\|[ \Sigma - I] E \| < \alpha$. By construction, it is clear that $X_g^p = (X_f^\ell)^\dag  (\sqrt[p]{R})^p = (X_f^\ell)^\dag R = X_{goal}^*$.

\section{Mathematical results}
\label{aMathematical}

The following  results are natural generalizations of the corresponding results of \cite{PerSilRou19}
regarding the case where $\nbar$ is equal to $n$.  In this paper one is interested in the situation where the quantum gate is encoded, that is, $\nbar$ is less than $n$.
In this case, recall that the partial trace Lyapunov function induces a convenient notion of distance $\pdist(\cdot, \cdot)$ that regards
the encoded space (see Definition \ref{dPartial} and Prop. \ref{pCritical}).
The partial distance is only a semi-norm in $\Un$, but it will be useful
to define the closed ``ball'' ${\overline{ \mathcal B}}_c (X)$ with center $X \in Un$ and radius $c>0$ that is induced by this notion:
\[
 {\overline{ \mathcal B}}_c (X)  = \{ Y \in \Un ~|~ \pdist(Y, X) \leq c \}
\]

As in \cite{PerSilRou19},
 the notion of attractive reference trajectories is instrumental for proving
 the convergence of RIGA when it is equipped with the partial trace Lyapunov function.
One may state the following
notion of attraction of a reference trajectory:
\begin{definition}
\label{dAttractive}
Let ${\overline u}_k : [0, T_f] \rightarrow \RR, k=1, \ldots, m$ be a set continuous reference inputs
and let ${\overline X}_0 \in \Un$. Let $\overline X: [0, T_f] \rightarrow \Un$ be the reference trajectory corresponding to the solution
of \refeq{reference} with initial condition ${\overline X}_0$. One says that this reference trajectory is $\lambda$-attractive  in
${\overline{ \mathcal B}}_c (I)$, where $\lambda \in (0, 1)$, if for every ${X}_0 \in  \Un$ such that
${\widetilde X}_0 \in  {\overline{ \mathcal B}}_c (I)$, where\footnote{In this paper, the initial condition $X(0)$ will always be the identity matrix,
but the notion of attraction may be defined in a more general situation, with arbitrary $X(0) \in \Un$.}
 ${\widetilde X}_0 = {\overline X}^\dag(0) {X}(0) $,
then the solution
$\widetilde X(t)$ of the closed  loop system \refeq{eClosedLoopComplete} is such that
\begin{equation}
 \label{eAttractive}
  \mathcal V (\widetilde X(T_f)) \leq \lambda \mathcal V (\widetilde X(0))
\end{equation}
\end{definition}
From Proposition \ref{pCritical} and Definition \ref{dPartial}, the condition \refeq{eAttractive} is equivalent to have
\begin{equation}
 \label{eAttractive2}
  \pdist (X(T_f), {\overline X}(T_f) \leq \sqrt{\lambda} \; \pdist (X_0, {\overline X}_0)
\end{equation}
for all $X_0, {\overline X}_0$ such that $\pdist (X_0, {\overline X}_0) \leq c$. Note that in \refeq{eAttractive2}, one is considering
 the closed loop system \refeq{eClosedLoopComplete2} instead of \refeq{eClosedLoopComplete}.

The next result shows that the $\lambda$-attraction of a reference trajectory depends only
on the reference input. Furthermore, this property holds with probability one with respect to the choice of
the jet of such reference input. It is important to note that one is not claiming here that,
if a reference trajectory  ${\overline X}_1$ is $\gamma$-attractive, it is also $\gamma$-attractive after  a right-translation\footnote{
That is, given an arbitrary $R \in \Un$, then the
reference trajectory ${\overline X}_2(t) = {\overline X}_1(t) R$  is also  $\gamma$-attractive. This result  holds for $n = \bar n$
in the context of \cite{PerSilRou19}.}
What is shown here is that there exists $\lambda \in (0,1)$, possibly greater than $\gamma$, such that
${\overline X}_2(t) = {\overline X}_1(t) R$ is $\lambda$-attractive for all $R \in \Un$.

\begin{theorem} \label{t1} Assume that system \refeq{cqs} is controllable. Choose $c >0$ with $c < 2$.
Then there exists $M_0$ large enough and a nontrivial polinomial\footnote{Note that $Q$ is a differential polynomial that that depends only on the matrices $H_k, k=0, 1, \ldots, m$ that defines system \refeq{cqs}.}
differential equation
\begin{equation}
\label{eQ}
Q\left( {\overline u}_k^{(j)}(t) : k \in \{1, \ldots , m\}, j \in \{0, 1, \ldots, M_0\} \right) =0
\end{equation}
such that the following results regarding \refeq{eQ} holds for the closed loop system \eqref{eClosedLoopComplete}:\\
(a) Assume that the reference ${\overline u}_k : [0, T_f] \rightarrow \RR, k \in \{1, 2, \ldots , m\}$ do not form a solution of \refeq{eQ}.
Then there exists\footnote{This value of $\lambda$ is common for all ${\overline X}_0 \in \Un$.}  $\lambda \in (0,1)$ such that, the reference trajectory $\overline X(t)$
that is a solution of \refeq{reference}, with these reference inputs, is $\lambda$-attractive in ${\overline B}_c(I)$ for all ${\overline X}_0 \in \Un$.
In particular, this property holds with probability one with respect to the choice of the $M_0$-jet of the reference inputs.\\
(b) Assume that reference inputs are of the form\footnote{The proof is written for the case when the window function
is absent, but is not difficult to generalize this result in the presence of the window function.} \refeq{refcon} with $M \geq M_0$.
 Then with probability one
with respect to the choice of  the coefficients $(\mathbf a,  \mathbf b)$, such reference inputs does not obey \refeq{eQ} for $t \in [0, T_f]$.
\end{theorem}

\begin{proposition}
\label{pEquivalence} The following affirmations are equivalent for a given reference trajectory $\overline X: [0, T_f] \rightarrow \Un$ corresponding to the solution
of \refeq{reference} with chosen ${\overline u}_k^{(j)}(t)$, and with an initial condition ${\overline X}_0 \in \Un$:\\
(a) There exists $\lambda \in (0,1)$ such that the reference trajectory $\overline X(t)$ is $\lambda$-attractive in ${\overline B}_c(I)$.\\
(b) There exists $\lambda  \in (0,1)$ such that, for every initial condition ${\widetilde X}_0 \in {\overline B}_c(I)$ of the closed loop system \refeq{eClosedLoopComplete} one has
 \begin{equation}\label{ineq1}
\begin{array}{c}
 \sum_{k=1}^m \int_{0}^{T_f}  4 K \left\{ \Re \left[ \trace \left(  E^\dag {\widetilde S}_k(t) {\widetilde X}(t) E \right)\right] \right\}^2 dt   \\
 \leq {\mathcal V ({\widetilde  X}_0)}(1 - \lambda)
\end{array}
\end{equation}

\end{proposition}

\begin{corollary}
\label{c1}
Assume that:\\
(a) The reference inputs  ${\overline u}^{\ell-1}_k(t), k=1, \ldots, m$ are not solutions
of the differential polynomial \refeq{eQ} (which is true with probability one).\\
(b) $\pdist (X_f^{\ell-1}, \Xgoalbarellminusone) < 2$ (or equivalently $\mathcal V({\widetilde X}^{\ell-1}(T_f)) < 4)$.\\
(c) $(\Xgoalbar -\Xgoalbarellminusone) E = 0$.\\
Then there exists $\lambda \in (0,1)$ such that $\pdist (X_f^\ell, \Xgoalbar) \leq  \sqrt{\lambda} \; \pdist (X_f^{\ell-1}, \Xgoalbarellminusone)$
(or equivalently $\mathcal V({\widetilde X}^{\ell}(T_f)) \leq \lambda \mathcal V({\widetilde X}^{\ell-1}(T_f))$.
\end{corollary}

 When one chooses the same $\Xgoalbar$ for all $\ell$, this means that no strategy for avoiding critical points of the Lyapunov funtion
 is implemented (see remark \ref{NoStrategy}). Then only a local convergence result of RIGA can be obtained:
 \begin{theorem}
     \label{tLocal} Assume that system \refeq{cqs} is controllable. Fix $X_{goal} \in \Un$ such that $X_{goal} E = F$ and take $\Xgoalbar =X_{goal}$ for all $\ell=1, 2, \ldots$.
     Choose $T>0$, $M > 0$ and coefficients $(\mathbf a, \mathbf b)$ for the reference input \refeq{refconab}.  Choose some $c$ with $2 >c > 0$. Suppose
     that the seed input is uniformly bounded, that is, there exists $L_0 >0$ such that  $|{\overline u}_k^0(t) | \leq L_0, \forall t \in [0, T_f]$
     and $k=1, 2, \ldots, m$.
     Suppose that  reference input $\overline X^0(t)$ is $\lambda$ attractive in ${\overline B}_c (I)$.
     If  $\mathcal V ({\widetilde X}^{\ell}(0))  \leq c^2$, then there exists $T^*$ large enough such that, for all $T_f > T^*$ one has;\\
     (a) There exists some $\theta \in (0, 1)$ such that the reference trajectory ${\overline X}^\ell$ generated by RIGA is $\theta$-attractive in ${\overline B}_c (I)$ for
      for all $\ell \in \NN$.\\
     (b) Let $d_\ell =  \pdist ({\overline X}^\ell(T_f), X_{goal})$. Them  $d_{\ell+1} \leq (\sqrt{\theta}) \, d_{\ell}$ for all
    $\ell=1, 2, \ldots$. In particular, $d_\ell$ converges exponentially to zero, monotonically.\\
     (c) The reference inputs generated
     by RIGA are uniformly bounded, that is, there exists some $\overline L$ (that depends on the system, on $c$ and on $\theta$) such that
     $|{\overline u}_k^\ell(t) | \leq L_0+ {\overline L}, \forall t \in [0, T_f], k=1,2, \ldots, m$, and for $\ell \in \NN$.
\end{theorem}

 The next result states  a convergence result for RIGA:
 \begin{theorem}
     \label{tGlobal}
     Assume that system \refeq{cqs} is controllable.
     Construct $X_{goal} \in \Un$ such that $X_{goal} E = F$. Choose $\alpha, \beta$, with $0 < \alpha < \beta < 2$.
     Choose $T>0$, $M > 0$ and coefficients $(\mathbf a, \mathbf b)$ for the reference input \refeq{refconab}. Assume that RIGA
     is executed with the second strategy of avoiding critical points that is presented in Appendix \ref{sSecondStrategy} ,
     that the seed input is uniformly bounded, that is, there exists $\mu_0 >0$ such that  $|{\overline u}_k^0(t) | \leq \mu_0, \forall t \in [0, T_f]$
     and $k=1, 2, \ldots, m$, and that the reference input ${\overline X}^0(t)$ is $\lambda$-attractive in ${\overline B}_\beta (I)$.
      Then there exists $T^*$ large enough such that, for all $T_f \geq T^*$ one has:\\
    (a)The reference trajectory ${\overline X}^\ell$ that is generated in step  $\ell$  of RIGA is $\theta$-attractive in ${\overline B}_\beta (I)$ for
     some fixed $\theta \in (0, 1)$ for all $\ell \in \NN$;\\
    (b)RIGA switches until the final goal matrix $X_{goal}^*$ is chosen, that is the choice of $\Xgoalbar$ among
    the elements of $X_g^q$ will attain $q=p$.
    When RIGA switches in step $\ell$, then  $d_\ell = \pdist(\Xgoalbar, X^{\ell-1} (T_f))$
     may be greater than $d_{\ell-1}$, but with $d_\ell \leq \beta$ for all $\ell \in \NN^*$.
    Between switchings, the inequality $d_\ell \leq (\sqrt{\theta}) d_{\ell-1}$ always holds.
    In particular, after the last switching,  $d_\ell$ converges exponentially to zero, monotonically.\\
    (c) Let $p$ be the number of goal matrices of the construction of the second strategy for avoiding critical points. The reference inputs generated
     by RIGA are uniformly bounded, that is, there exists some $\overline \mu$ such that
     $|{\overline u}_k^\ell(t) | \leq \mu_0+ p \overline \mu \beta, \forall t \in [0, T_f], k=1,2, \ldots, m$, and for $\ell \in \NN$.
   \end{theorem}

\section{Characterization of $\lambda$-attractive trajectories (Theorem \ref{t1})}

\subsection{Auxiliary results}

As explained in the beginning of section \ref{aCritical}, one will consider in this entire section the change of basis $X(t) = X_E^\dag(t) X(t) X_E$,
$\overline X(t) = X_E^\dag(t) \overline X(t) X_E$, and $\widetilde W (t) = {\overline W}^\dag(t) W(t)$. After this change of basis, the system \refeq{cqs} reads
\begin{subequations}
\begin{equation}
\label{eWcqs}
 \dot W (t) = (T_0 + \sum_{k=0}^{m} u_k(t) T_k ) W(t)
\end{equation}
where $T_k = X_E^\dag S_k X_E$, $k=0, 1, \ldots, m$. The reference system reads
\begin{equation}
\label{eWReference}
 \dot {\overline W} (t) = (T_0 + \sum_{k=0}^{m} {\overline u}_k(t) T_k ) \overline W(t)
\end{equation}
and the error system reads
\begin{equation}
\label{eWerror}
 \dot {\widetilde W} (t) = (\sum_{k=0}^{m} {\widetilde u}_k(t) {\widetilde T}_k(t) ) \widetilde W(t)
\end{equation}

where  ${\widetilde T}_k(t) = {\overline W}^\dag (t) T_k \overline W(t)$.
The closed loop control after this change of basis reads
\begin{equation}
 \label{eWfeedback}
{\widetilde u}_k (t) = K \Re [\trace( \Pi {\widetilde T}_k(t) \widetilde W(t) \Pi^\dag)]
\end{equation}
\end{subequations}
Furthermore, by proposition \ref{pCritical}, $\mathcal V (\widetilde W) = 2 \nbar - \Re[\trace(\Pi \widetilde W \Pi^\dag)]$.
The following lemma regards the system after this change of basis.

\begin{lemma}
\label{lWhatSigma}
Assume that, for all $\lambda \in (0 ,1)$ there exists an initial condition ${\widetilde W}(0) = {\widetilde W}_0 \in \Un$ for which $0 < \mathcal V({\widetilde W}_0)\leq h < 4$ and $\mathcal V ({\widetilde W}(T_f)) > \lambda \mathcal V({\widetilde W}_0)$. Then there exists a nonzero $n \times \nbar$ complex matrix $\widehat \xi$ such that:
\begin{subequations}
\begin{equation}
\label{Re0}
\begin{array}{c}
\Re \left[ \trace \left( \Pi {\overline W}^\dag(t) T_k(t) {\overline W}(t) \widehat \xi  \right) \right] = 0, \\
\forall t \in [0, T_f],  \forall k \in \{1, \ldots, m\}
\end{array}
\end{equation}
 Furthermore,  one of two conditions
\refeq{eCa} or \refeq{eCb} holds
\begin{eqnarray}
\label{eCa}
& \;\;\;\; & {\widehat \xi}^\dag  {\widehat \xi}  =  I_{\nbar},\; \mbox{with}\; \| \widehat \xi - I_n \Pi^\dag \|^2 \leq h < 4\;   \\
\label{eCb}
& \;\;\;\; & {\Pi \widehat \xi} + {\widehat \xi}^\dag \Pi^\dag =  0,\; \mbox{with}\; \| \widehat \xi \| = 1.
\end{eqnarray}
\end{subequations}

\end{lemma}

\begin{proof}
As in the proof of Prop. \ref{pEquivalence},
The assumption of the Lemma implies that
$\int_0^{T_f} {\dot {\mathcal V}}(t)  = \mathcal V (\widetilde  W(T_f)) - \mathcal V (\widetilde  W(0)) > (\lambda - 1) \mathcal V (\widetilde  W(0))$.
Recall that, for the closed loop system, the derivative of the Lyapunov function is given by
\refeq{eVdot}. Hence, computing $\mathcal V (\widetilde  W(T_f)) - \mathcal V (\widetilde  W(0)) = \int_{0}^{T_f} \dot{\mathcal V}(t) dt$, after dividing by $ \mathcal V ({\widetilde  W}_0)$,  one concludes that, for every $\lambda \in (0, 1)$ there exists ${\widetilde  W}_0$
with $ \mathcal V ({\widetilde  W}_0) \in (0, h]$ such that, the corresponding solution $\widetilde W(t)$ is such that
\begin{equation}\label{ineq}
\begin{array}{c}
 \sum_{k=1}^m \int_{0}^{T_f} \frac{ 4 \left\{ \Re \left[ \trace \left(  \Pi {\widetilde T}_k(t) {\widetilde W}(t) \Pi^\dag\right)\right] \right\}^2 }{\mathcal V ({\widetilde  W}_0)} dt  = \\
 \sum_{k=1}^m \int_{0}^{T_f} \frac{{\widetilde u}^2_k(t)}{K \mathcal V ({\widetilde  W}_0)} dt  <
 (1 - \lambda)
\end{array}
\end{equation}
Let $\lambda_\ell, \ell \in \NN$ be a sequence with $\lambda_\ell \in (0,1)$ such that
$(1 - \lambda_\ell) = 1/\ell$. For each $\lambda_\ell$ constructed in this way one may choose
${\widetilde W}_{0_\ell}$ in a way that \refeq{ineq} holds for ${\widetilde W}(0) = {\widetilde W}_{0_\ell}$,
 ${\widetilde W}(t) = {\widetilde W}_\ell(t)$, ${\widetilde u}_k (t) = {\widetilde u}_{k_\ell}(t)$,
and $\lambda = \lambda_\ell$.  Define $\alpha_\ell(t) = ({\widetilde W}_\ell(t)-I_n) \Pi^\dag$ and
$\alpha_{0_\ell} = \alpha_\ell(0)$. By part (d) of Prop. \ref{pCritical} it follows that
 $\mathcal V ({\widetilde W}(t)) = \| \alpha_{\ell}(t)\|^2$. Now
note that
\begin{equation}
\label{eZeroSigma}
\Re \left[ \trace \left(  \Pi \sigma \Pi^\dag\right)\right] =0, \forall \sigma \in \un.
\end{equation}
 In particular, since ${\widetilde T}_k(t) \in \un, \forall t \in [0, T_f]$,
 one can replace ${\widetilde W}_\ell (t) \Pi^\dag$ by $\alpha_\ell(t)$
 in \refeq{ineq}. Define $\xi_\ell(t) = \alpha_\ell(t)/ \| \alpha_{0_\ell}\|$. This
  is well defined because $\|\alpha_{0_\ell}\| =  \mathcal V ({\widetilde  W}_0) >0$. After the last substitutions,   by linearity of the trace one gets:
\begin{equation}
\label{ineq2}
\begin{array}{c} \sum_{k=1}^m \int_{0}^{T_f} 4 \left\{ \Re \left[ \trace \left(  \Pi {\widetilde T}_k(t) \xi_\ell(t) \right)\right] \right\}^2 dt = \\
\sum_{k=1}^m \int_{0}^{T_f} \frac{{\widetilde u}^2_{k_\ell}(t)}{K \| \alpha_{0_\ell}\|^2} dt < \frac{1}{\ell}
\end{array}
\end{equation}
Now, the same ideas of \cite[eqs. (33)-(34)]{SilPerRou14} will be considered (with $Z_n$ replaced by $\alpha_\ell$). It will be shown that a subsequence of $\xi_\ell$ converges uniformly
in the interval $[0, T_f]$ to some fixed $\xi^* \in G$, where $G$ is in the compact subset of complex matrices with unitary Frobenius norm. In fact,
\begin{equation}\label{contradiction}
\begin{array}{c}
 \int_{0}^{T_f}
\sum_{k=1}^{m} K\mbox{Tr}^2(\alpha_\ell(t) \widetilde{W}_k(t)) \, dt = \\
 \int_{0}^{T_f}
\sum_{k=1}^{m} \left(\frac{\widetilde{u}_{k_\ell} (t)}{\sqrt{K}}\right)^2  \, dt < \\
\dfrac{1}{\ell} \| \alpha_{0_\ell} \|^2,
\end{array}
\end{equation}
Fix $\ell \in \NN$ with $\ell > 0$. The Cauchy-Schwartz inequality provides that\footnote{Remember that, for measurable functions $f, g$ one has $\int_S | f(x) g(x) | dx \leq \sqrt{\int_S  f^2(x) dx \int_S  g^2(x) dx}$.}, for $t \in [0,T_f]$ one has
\begin{equation}\label{utilde1}
\begin{array}{c}
    \int_{0}^{t}  \left|\frac{\widetilde{u}_{k_\ell} (s)}{\sqrt{K}}\right|  \, ds   \leq
    \int_{0}^{T_f}  \left|\frac{\widetilde{u}_{k_\ell} (s)}{\sqrt{K}}\right|  \, ds
    \leq \\
     \sqrt{T_f} \sqrt{\int_{0}^{T_f} \left(\frac{\widetilde{u}_{k_\ell} (t)}{\sqrt{K}}\right)^2 \, dt}  <
     \sqrt{T_f/\ell} \| \alpha_{0_\ell} \|.
\end{array}
\end{equation}
Now, note that
\[
    \dot{\alpha_\ell}_n(t) = \sum_{k=1} \widetilde{u}_{k_\ell}(t) R_k(t)  ,
\]
where the norm of  $R_k(t)= {\widetilde T}_k(t) {\widetilde W}_k \Pi$ is continuous and uniformly bounded on $[0,T_f]$ by some $D_k > 0$. Thus, for every $t \in [0, T_f]$,
\begin{equation}
\label{eZn}
\begin{array}{c}
    \left\| \dfrac{\alpha_\ell(t) - \alpha_{0_\ell}}{\|\alpha_{0_\ell} \|} \right\| =
    \left\| \int_{0}^t \dfrac{\dot{\alpha}_\ell(s)}{\|\alpha_{\ell_0} \|} \, ds \right\| \leq \\
    \sum_{k=1}^m D_k \int_{0}^{t} \frac{| \widetilde{u}_{k_\ell}(s) |}{\| \alpha_{0_\ell} \|} \, ds.
\end{array}
\end{equation}
From \refeq{utilde1} and \refeq{eZn}, it follows that
 \[
 \left\| \dfrac{\alpha_\ell(t) - \alpha_{0_\ell}}{\| \alpha_{0_\ell} \|} \right\|
\leq \sqrt{K} D \sqrt{T/\ell},
\]
where $D=D_1 + \dots + D_m > 0$.
As the sequence $\alpha_{0_\ell}/\|\alpha_{0_\ell}\|$ belongs to the compact set $\mathbf{G}$, there exists a convergent subsequence. For simplicity, denote such subsequence by
$\alpha_{0_\ell}/\|\alpha_{0_\ell}\|$ and let $\xi^* \in \mathbf{G}$ be its limit. It follows that $\alpha_{\ell}/\|\alpha_{\ell}\|$ uniformly converges to $\xi^*$ on the interval $[0,T]$ as $\ell \to \infty$, as claimed.
So, after taking the limit $\ell \rightarrow \infty$, one gets:
\[
\sum_{k=1}^m \int_{0}^{T_f} 4 \left\{ \Re \left[ \trace \left(  \Pi {\widetilde T}_k(t) \xi^* \right)\right] \right\}^2  =  0
\]
and by the continuity of ${\widetilde T}_k(t)$ it follows easily that
\begin{equation}
\label{eXiStar}
\Re \left[ \trace \left(  \Pi {\widetilde T}_k(t) \xi^* \right)\right]   =  0,
\end{equation}
for all $k \in \{1, \ldots, m\}$.
As $\alpha_{0_\ell} = ({\widetilde W}_{0_\ell} - I_n) \Pi^\dag$ with ${\widetilde W}_{0_\ell} \in \Un$, then, one may assume without loss of generality, possibly after taking a convenient subsequence, that $\lim_{\ell\rightarrow\infty} \alpha_{0_\ell} (t) = \alpha^*$ and $\lim_{\ell\rightarrow\infty} \| \alpha_{0_\ell} (t) \| = v* \leq  h$. Let $\beta_\ell = \alpha_{0_\ell} + I_n \Pi^\dag = W_{0_\ell} \Pi^\dag$. Then simple calculations  shows that $\beta_\ell^\dag \beta_\ell = \Pi I_n \Pi^\dag = I_{\nbar}$.
Let $\widehat \xi = \lim_{\ell \rightarrow \infty}  \beta_\ell$. Note that $\widehat \xi =
\lim_{\ell \rightarrow \infty} \left( \alpha_{0_\ell} + I_n \Pi^\dag \right) = \alpha^* + I_n \Pi^\dag \neq I_n \Pi^\dag$. Furthermore, as $\beta_\ell^\dag \beta_\ell = I_{\nbar}$, it follows that ${\widehat \xi}^\dag {\widehat \xi} = I_{\nbar}$.
 By the linearity of the trace, by \refeq{eXiStar} and by \refeq{eZeroSigma}, it follows that
\[
\Re \left[ \trace \left(  \Pi {\widetilde T}_k(t) {\widehat  \xi} \right)\right]   =  0.
\]
Since $\widehat \xi - I_n \Pi^\dag = v^* \xi^*$ with $0< v^*\leq h$ and $\|\xi^* \| =1$, then \refeq{eCa} holds.
Now, to show the other possible situation when \refeq{eCb} holds, assume that $v^*= \lim_{\ell \rightarrow \infty} \|\alpha_{0_\ell}\| =0$.
Since $\beta_\ell = (\alpha_{0_\ell} + I_{\nbar} \Pi^\dag)$, then $I_{\nbar} = \beta_\ell^\dag \beta_\ell =$
$ \alpha_{0_\ell}^\dag \alpha_{0_\ell}+ \Pi \alpha_{0_\ell}^\dag  + \alpha_{0_\ell} \Pi^\dag + \Pi I_n \Pi^\dag$. Then, dividing
both sides of the last equation by $\| \alpha_{0_\ell}\|$ and recalling that $\xi_\ell = \alpha_{0_\ell}/\|\alpha_{0_\ell}\|$, one gets
$\frac{\alpha_{0_\ell}^\dag \alpha_{0_\ell}}{\| \alpha_{0_\ell} \|} + \Pi \xi_\ell  + \xi_\ell^\dag \Pi^\dag =0$. Since $\left\| \frac{\alpha_{0_\ell}^\dag \alpha_{0_\ell}}{\| \alpha_{0_\ell} \|} \right\| \leq \| \alpha_{0_\ell} \|$,  taking the limit $\ell \rightarrow\infty$ in both sides of the last equation and recalling that  $\|\xi^*\| = 1$,  one shows \refeq{eCb} for $\widehat \xi = \xi^*$.
\end{proof}

\begin{lemma}
\label{lCkj} Consider system \refeq{eWcqs}-\refeq{eWReference} and define
 $C_k^j(t) \in \un$ as in \cite{SilPerRou14,SilPerRou16}, by:
\begin{equation}
 \begin{array}{rcl}\label{defCkj}
  C_k^0(t) & = & T_k, \\
  C_k^{j+1}(t) & = & \dot{C}_k^{j}(t) + \left[ {C}_k^{j}(t),  A(t) \right],
 \end{array}
\end{equation}
for $k =1, \ldots m$, $j \in \NN$, $t \in \RR$, where $A(t) = T_0 + \sum_{k=1}^{m} u_{k}(t) T_k \in \mathfrak{u}(n)$.
Assume that $\xi$ is a $n \times \nbar$ complex matrix such that $\Re \left[ \trace \left( \Pi {\overline W}^\dag(t) { T}_k {\overline W}(t) \widehat W \Pi^\dag \right) \right] = 0$ for all $k \in \{1, \ldots, m\}$ and for all $t \in [0, T_f]$. Then
\begin{equation}
\label{Re0Ckj}
\begin{array}{c}
\Re \left[ \trace \left( \Pi {\overline W}^\dag(t) C_k^j(t) {\overline W}(t) \xi\right) \right] = 0 \\ \forall k \in \{1, \ldots, m\},
\forall t \in [0, T_f], \forall j \in \NN.
\end{array}
\end{equation}
\end{lemma}

\begin{proof}
The proof is easy by derivation and it is analogous to the proof of equation (23)
of \cite{SilPerRou14}.
\end{proof}

\begin{lemma} \label{Vmt} Assume that the quantum system is controllable. Fix some $t \in [0, T_f]$ and define the subspace ${\mathcal V}^M_t$ of $\un$ by
\begin{equation}
\label{eSpan}
{\mathcal V}^M_t = \mbox{\textrm{span}} \{ C_k^j(t) :  k \in \{1, \ldots, m\}, j  \in \{1, \ldots, M\}
\end{equation}
Given $M \in \NN$, define the set  $\mathbb{U}_{M}$ of differential variables by:
\[\mathbb{U}_{M} = \left( {u}_k^{(j)} : k \in  \{1, \ldots, m\}, j \in \{1, \ldots , M\}\right). \]
 Then $\mathfrak{P}(\mathbb{U}_{M})$ will denote a differential polynomial in those variables. Then there exists $M_0 \in \NN$ and a nontrivial differential polynomial\footnote{Note that $\mathfrak{P}$ is the same for all $t\in [0, t_f]$.}  $\mathfrak{P}\left( \mathbb{U}_{M_0} \right)$ such that, if
 \begin{equation}
 \label{eSpan2}
\left. \mathfrak{P}\left( \mathbb{U}_{M_0} \right) \right|_{ \mathbb{U}_{M_0} = \left( {\overline u}_k^{(j)} (t) : k \in  k \in \{1, \ldots, m\}, j \in \{1, \ldots, M_0\} \right) } \neq 0
\end{equation}
then ${\mathcal V}^M_t = \un$. Now assume that the reference controls ${\overline u}_k^{(j)} (t), k \in \{1, \ldots, m\}$ are of the form \refeq{refcon} for some $M \geq M_0$. Fix $t \in [0, T_f]$. Then \refeq{eSpan2} holds
  with probability one with respect to the choice of a random pair $(\mathbf a, \mathbf b)$ that defines \refeq{refcon}.
\end{lemma}

\begin{proof}
The Proof follows easily from the same arguments of \cite[Section 1]{SilPerRou16}, that holds for $t \in [0, T_f]$ (in that paper one considers only  $t=0$).
\end{proof}

\subsection{Proof of Theorem \ref{t1}}

\begin{proof}
 Let $M_0 >0$ and  $\mathfrak P$ be the differential polynomial whose existence is assured by Lemma \ref{Vmt}.
 Assume that the reference controls does not obey the condition \refeq{eQ} for $Q = \mathfrak P$, at least for some $t \in [0, T_f]$.
  Now assume that does not exist $\lambda \in (0, 1)$ such that the statement of the Theorem \ref{t1} holds
 for $Q = \mathfrak P$. Hence, for all $\lambda \in (0 ,1)$ there exists an initial condition ${\widetilde W}(0) = {\widetilde W}_0 \in \Un$ for which $\mathcal V({\widetilde W}_0)\leq h < 4$ and $\mathcal V ({\widetilde W}(T_f)) > \lambda \mathcal V({\widetilde W}_0)$. One shows now that the last condition
 implies that $\mathcal V({\widetilde W}_0)=0$.
 Note that, for the closed loop system, as the Lyapunov function is positive and nonincreasing, then  $\mathcal V({\widetilde W}_0)=0 $ implies that $\mathcal V({\widetilde W}(t)) =0, \forall t \in [0, T_f]$, and so, in this case, one cannot have $\mathcal V ({\widetilde W}(T_f)) > \lambda \mathcal V({\widetilde W}_0)$. One concludes that the asumptions of Lemma  \ref{lWhatSigma}  holds.
 Hence, from that Lemma, there exists a $n \times \nbar$ complex matrix satisfying \refeq{Re0}, for which one of the two conditions \refeq{eCa} of \refeq{eCb} holds.
 By Lemma \ref{lCkj}, then it follows that \refeq{Re0Ckj} holds. In particular, Lemma \ref{Vmt} will also
 show that the assumption of Lemma \ref{lCritical} holds. By part (a) of that Lemma, one concludes that
 \begin{equation}
 \label{eXi}
 \widehat \xi =
  \left[
  \begin{array}{c}
 \xi_1\\
 \xi_2
 \end{array}
 \right]
 \end{equation}
 where $\xi_1$ is a square block and $\xi_2 =0$.
  Assume that condition  \refeq{eCa} of Lemma \ref{lWhatSigma}  holds. In particular, $\xi_1 \in \Unbar$. hence $\xi_1 = V^\dag D V$ where $V \in \Unbar$ and $D$ is a diagonal matrix such that $\{D\}_{ii} = \exp (\jmath \theta_i), i \in \{1, \ldots, \nbar\}$. By part (b) of Lemma \ref{lCritical}, then $D$ is real and so $\{D\}_{ii} = \pm 1, i \in \{1, \ldots, \nbar\}$.  Now note that
  $\| \widehat \xi - I_n \Pi^\dag\|^2  = \trace \left\{ ({\widehat \xi} - I_n \Pi^\dag)^\dag ({\widehat \xi} - I_n \Pi^\dag)\right\} = \trace \left\{ {\widehat \xi}^\dag {\widehat \xi} + I_{\nbar} + \Pi {\widehat \xi} + {\widehat \xi}^\dag \right\} \Pi^\dag = 2 \nbar - 2 \Re [\trace( \xi_1)] =
   2 \nbar - 2\Re \left[ \trace \left( V^\dag D V \right) \right]= 2 \nbar - 2\Re \left[ \trace \left(  D  \right) \right] \leq h < 4$.
 It follows easily that the only possibility is to have $D = I_{\nbar}$, and this
 implies that $\| \widehat \xi - I_n \Pi^\dag| =0$. This
 contradicts  \refeq{eCa}.

 Now assume that condition \refeq{eCb}  of Lemma \ref{lWhatSigma} holds. Recall that one has already shown
 that \refeq{Re0Ckj} holds with $\widehat \xi$ given by \refeq{eXi}, where $\xi_1$ is a square block and $\xi_2 =0$. Now, condition \refeq{eCb} implies that $\xi_1 \in \unbar$ with $\xi_1 \neq 0$. Hence $\xi_1  = V^\dag D V$ where $V \in \Unbar$ and $D$ is an imaginary diagonal matrix. Similarly to the last case, one shows using part (b) of Lemma \ref{lCritical}  that $D$ is real. Hence $D=0$ and so $\widehat \xi =0$. This contradicts \refeq{eCb}.
\end{proof}

\section{Proofs of Convergence of RIGA}
\label{aLocal}
The proofs of convergence of the RIGA in its local and global versions are not that simple, but are essentially
adaptations of the corresponding proofs of the results of \cite{PerSilRou19} for the case $n=\bar n$, with a different
Lyapunov function.
As in the last paper, a notion of attraction in a subinterval $J \subset [0, T_f]$ is stated, which is instrumental
for these proofs. After that, one states the main ideas of these proofs. The details are then
proved along a series of auxiliary Lemmas.

\subsection{The notion of $\lambda$-attraction in an interval $J \subset [0, T_f]$}

Instrumental for the proof of Theorem \ref{tLocal} (and \ref{tGlobal} and is the notion of
$\lambda$-attractive trajectories in an interval $J \subset [0, T_f]$:
\begin{definition} \label{dAttractiveJ} Let $J =[\tau_0, \tau_1] \subset [0, T_f]$. Let $c>0$ and let $\lambda \in (0,1)$. A solution $\overline X(t)$ of \refeq{reference} for some set
of continuous reference controls
is said to be \emph{$\lambda$-attractive} on $(J, \BallV{c})$ if the solution $\widetilde X(t)$ of the closed-loop system \refeq{eClosedLoopComplete} on the interval $J$ with ${\widetilde X}(\tau_0) = {\widetilde X}_0$ is such that $\mathcal V ( \widetilde X( \tau_1)) \leq \mathcal \lambda \mathcal V ( \widetilde X( \tau_0))$, for every ${\widetilde X}_0 \in \BallV{c}$.
\end{definition}

\begin{remark}\label{rJBallv}
 If a reference trajectory is $\lambda$-attractive on $(J, \BallV{c})$, then it is $\lambda$-attractive on $\BallV{c}$. It is easy to show
\refeq{eAttractive} from the fact that the Lyapunov function $\mathcal V(\widetilde X(t))$ is nonincreasing.
Note that all reference trajectories are $\gamma$-attractive on  $(J, \BallV{c})$  with
$\gamma \geq 1$. However, in the definition of $\lambda$-attraction (see Definition \ref{dAttractive}), it is assumed
that $\lambda <1$.
\end{remark}

The proof of Theorems \ref{tLocal} and \ref{tGlobal} are  mainly based on the following lemma,
whose proof is given in Subsection~\ref{alMain} of this section.  Assume that RIGA is being executed in a way that $\Xgoalbar =\Xgoalbarellminusone$
at least between steps $\ell_0+1$ and $\ell_1$.
 All the conclusions of the next Lemma will hold for the steps $\ell$ of RIGA such that $\ell_0 < \ell < \ell_1$.
\begin{lemma}
\label{lMain}

  Fix $T_f>0$. Let ${\overline u}_k^0 : [0, T_f] \rightarrow \RR$ be a set of uniformly bounded reference
 inputs, that is. there exists a pair of positive real numbers $(L_0, {\mathcal L}_0)$, with $L_0\geq {\mathcal L}_0 > 0$, such that $\max_{t \in [0, T_f]} |{\overline u}_k^0 (t)| \leq {\mathcal L}_0, k=1, \ldots, m$.
  Let $c>0$ such that $c<2$.

 Let $J = [\tau_0, \tau_1] \subset [0, T_f]$, with $\tau_1 - \tau_0 = T_1$. Let $\lambda \in (0,1)$. Assume that ${\overline X}^0(t)$ is $\lambda$-attractive on  $(J, \BallV{c})$. Let $t^* \in [0, \tau_0]$.
Given $\theta \in (0, 1)$, define  the  sequence
\begin{equation}
\label{phi_ell}
\phi_\ell(\theta) = \sum_{j=0}^{\ell} (\sqrt{\theta})^j = \frac{ 1- \sqrt{\theta}^{j+1}}{1-\sqrt{\theta}}, \ell \in \NN.
\end{equation}
Let $\phi_\infty (\theta)$ stands for $\lim_{\ell \rightarrow \infty} \phi_\ell(\theta) =  \frac{1}{1-\sqrt{\theta}}$. Let $\delta_\lambda > 0$ with $\delta_\lambda < 1 - \lambda$, and define $\theta = \lambda + \delta_\lambda$.
Assume that $\theta \in (0, 1)$.
There exists $\Lambda >0$, where $\Lambda$ depends only on the parameters of the set
$\mathcal P =\{L_0, \theta, K, T_1, c, \|S_k\|: k=1,\ldots , m\}$ (and so it does not depend on $T_f$), such
that, if $\sqrt{\mathcal V({\widetilde X}^0(t^*))} \leq \min \left\{ \frac{\delta_\lambda}{\Lambda \phi_{\infty}(\theta) }, \sqrt{c} \right\} $,then
\begin{itemize}
 \item ${\overline X}^\ell$ is $\theta$-attractive on  $(J, \BallV{c})$. In particular, ${\overline X}^\ell$ is $\theta$-attractive on  $\BallV{c}$.
 \item  $\mathcal V ({\widetilde X}^\ell(0)) \leq \theta^\ell \mathcal V ({\widetilde X}^0(0))$.
 \item There exists $\overline M > 0$, where $\overline M$ depends on the  set $\mathcal P$
     such that $\| {\overline u}_k (t)\| < {\mathcal L}_0 + \overline M \phi_\infty(\theta) \sqrt{c}$, for all $t \in [0, T_f]$, for all $k \in \{1, \ldots , m\}$, and for all $\ell \in \NN$.
 \end{itemize}
\end{lemma}

\begin{remark}
In the context of the proof of Theorem \ref{tLocal}, the values of ${\mathcal L}_0$ and $L_0$ coincide.
Furthermore, in the context
of Theorem \ref{tGlobal}, the Lemma \ref{lMain} will be applied  between the steps  where switchings do occur. If a switching occurs at step $\overline \ell$,
one may consider the application of the last lemma until the next switching, considering that ${\overline u}^{\overline \ell}_k(t)$ as a new seed, and so on.
In the context of
the proof of Theorem \ref{tGlobal} one will take  $L_0 > {\mathcal L}_0$ in the applications of Lemma \ref{lMain}. This will produce a pessimistic (greater) value of
$\Lambda$, but this value of $\Lambda$  can be used in all the steps of RIGA. The details will be explained later.
\end{remark}

\section{Proof of Theorem \ref{tLocal}}
\label{a34}

 The  idea is to divide the interval $[0 , T_f]$ into two parts, that is
$[0, T_f] = \overline J \cup J_{r+1}$, where $\overline J = \bigcup_{i=1}^{r} J_i$. The integer
$r$ will be determined in a way that the first
interval $\overline J$ will be responsible for delivering a ``sufficiently small''  error matrix
$\widetilde X(t^*)$ for $t^* = \min J_{r+1}$.  Lemma \ref{lLambda0} will imply that each interval $J_i$
is $\gamma_i$-attractive, with $\gamma_i \geq \lambda_0$. This will allow the application of Lemma \ref{lMain}
to the interval $J_{r+1}$.
The proof is divided in several parts:\\
(A) Firstly, construct a seed ${\overline X}^0(t)$ using reference inputs of the form \refeq{refcon}-\refeq{refconab}
for chosen $T$ and $M$. Choose $T_1 = T/h> 0$, where $h \in \NN$. For the moment, choose $T_f \geq (r+1) T_1$ where $r$ is to be determined.
On will show in Lemma  \ref{lLambda0} that with probability one, the reference trajectory ${\overline X}^0(t)$  is $\gamma_i$-attractive in
 $(J_i, \BallV{c})$, where
$J_i = [\tau_{0_i}, \tau_{1_i}]= [(i-1) T_1, i T_1]$ for $i=1, \ldots, r$ and $\gamma_i \in (0, 1)$. Let $\lambda_0 = \min \{\gamma_1, \ldots, \gamma_{r+1}\}$. As \refeq{refcon} is $T$-periodic, then Lemma \ref{lLambda0} also shows that
 $\gamma_{i+h} = \gamma_i$ for $i> h$.  This means that $\lambda_0 = \min \{\gamma_1, \ldots, \gamma_{h}\}$.\\
(B) Let $\overline J = \bigcup_{i=1}^r J_i$. It follows easily that the  ${\overline X}^0(t)$ is $\lambda_0^r$-attractive
in $(\overline J, \BallV{c})$. In particular, in the first step of
RIGA with constant $\Xgoalbar$,  $\mathcal V(\widetilde X(\tau_{1r}) ) \leq \lambda_0^r {\mathcal V}(\widetilde X(0))$.\\
(C) Fix $\theta \in (\gamma_0, 1)$. Let $\delta_\lambda = \lambda_0 - \theta$. Let $t^* = \min J_{r+1}$.
It is clear that, for $r$ large enough, in the first step one has $\mathcal V(\widetilde X(t^*) \leq \lambda_0^r \mathcal V ({\widetilde X}^0(0))\leq
 \min \left\{ \left( \frac{\delta_\lambda}{\Lambda \phi_{\infty}(\theta) } \right)^2, c \right\}$,
where the value of $\Lambda$ regards a convenient application of Lemma \ref{lMain} to the interval $J=J_{r+1}$.
So, the appivation of Lemma \ref{lMain}  concludes the Proof of the Theorem.

\begin{lemma} \label{lLambda0}
 Fix $T>0$, $M>0$ and $(\mathbf a, \mathbf b)$ defining a reference input \refeq{refcon} for system \refeq{reference}. Then, with probability one,  for all subinterval $J = [\tau_0, \tau_f] \subset [0, T_f]$, the reference trajectory $\overline X$ is $\gamma$-attractive in $(J, \BallV{c})$ for some $\gamma \in (0,1)$. Furthermore, the value of the attractive factor $\gamma$ depend only on the reference inputs on the interval $[\tau_0, \tau_f]$.
\end{lemma}
\begin{proof}
By Theorem \ref{t1}, the differential polynomial $Q( {\overline u}_k^{(j)}: k \in \{1, \ldots , m\}, j=0, 1, \ldots M)\neq 0$ for some $t^* \in [0, T_f]$  with probability one. Substituting \refeq{refcon} in the differential polynomial $Q$,  nonzero analytical function of $t$ is obtained, and so $Q$ cannot be identically zero on $J$. From the same ideas of the proof of Theorem \ref{t1} applied to the interval $J$, it follows that the reference trajectory restricted to $J$ must be $\gamma$-attractive for some $\gamma \in (0,1)$. The last affirmation can be shown by the same arguments of the proof of Theorem \ref{t1}.
\end{proof}

\section{Proof of Theorem \ref{tGlobal}}

The idea of the proof of Theorem \ref{tGlobal} is similar to the one of Theo. \ref{tLocal}.
It is important to stress that, between two switchings of the goal matrix, the behaviour of RIGA
is the same of the one in the context of Theorem \ref{tLocal}.
 After a switching, the Lyapunov
function is instantaneously increased, but this new value is bounded by $\sqrt{\beta}$. This follows easily from
the defined switching condition in section \ref{sSecondStrategy}, that is, the switching from the goal matrix $X_g^q$ to $X_g^{q+1}$ in a step $\ell$ of RIGA only occurs when
$\pdist (X_g^{q+1}, X_f^{\ell-1}) \leq \beta$.

Now the idea is to divide the interval $[0 , T_f]$ in $p+1$ parts,
where $p$ is the total number of goal matrices, that is the maximum number of possible switchings.
Then
$[0, T_f] =  \bigcup_{q=1}^{p+1} {\overline J}_q$. The intervals ${\overline J}_q$ will
have length $r T_1$ for $q < p+1$ and length greater than or equal to $T_1$ for $q=p+1$. In other words,  for $q < p+1$
then
${\overline J}_q = \bigcup_{i=1}^{r} J_{qi}$, where  all the $J_{qi}$ have length $T_1$.
Assume that no switching has occurred yet.
The interval ${\overline J}_1$ will be responsible for delivering a ``sufficiently small''  error matrix
$\widetilde X(t_2)$ for $t_2$ equal  $\min {\overline J}_{2}$. This will allow to apply Lemma \ref{lMain}
for all subintervals $J_{qi}$ for $q>1$ and for ${\overline J}_{p+1}$.
If only one switching has occurred, then the interval ${\overline J}_2$ will be responsible for delivering a ``sufficiently small''  error matrix
$\widetilde X(t_3)$ for $t_3$ equal to $\min {\overline J}_{3}$. This will allow to apply Lemma \ref{lMain}
for all subintervals $J_{qi}$ for $q>2$ and for ${\overline J}_{p+1}$, and so on. This idea will be explained more deeply. \\
(a) The step one of Algorithm 1  and all  steps $\ell$ for which a switching of the goal matrix occurs will be called
by switching steps. Recall that, between two switching steps, the step $\ell$ obeys the same conditions of the context of Theorem \ref{tLocal}.
Note  that, at a switching step $\ell$,
  one has $\mathcal V({\widetilde X}^{\ell} (0)) \leq \beta.$  \\
(c) Let $T_f \geq (p r + 1) T_1$ where $r$ is to be determined. Assume that $[0, T_f] \subset \bigcup_{q=1}^{p+1} {\overline J}_q$ where
${\overline J}_q = [r (q-1)  T_1, r q T_1] = \bigcup_{j=1}^{r} J_{jq}$
where $J_{jq} = [ (j -1  + r q - r) T_1,  (j + r q - r) T_1]$, for $q < p$ and $j=1, 2, \ldots, r$. Furthermore, ${\overline J}_{p+1} = [ r p T_1, T_f]$.
Note that the intervals $J_{jq}, j=1, \ldots, r, q=1, \ldots, p$  have length $T_1$ and the interval ${\overline J}_{p+1}$ has length greater that $T_1$.  It is clear from the same reasoning (A) of Section \ref{a34} above, that it is possible to assume that
the reference trajectory ${\overline X}^0$ is $\lambda_0$-attractive in
 $(J_{jq}, \BallV{c})$ for all $j=1, \ldots , r$ and for all $q=1, \ldots, p$. Furthermore,
 ${\overline X}^0$ is $\lambda_0$-attractive in $({\overline J}_{p+1}, \BallV{c})$.\\
 (d) To show that a first switching will certainly occur. it will be shown that, if a switching never occurs
 then $\lim_{\ell \rightarrow \infty} {\mathcal V}^\ell (\widetilde X^\ell(0)) = 0$.
 To show that this last convergence implies that a first switching will ocurr,
 without loss of generality, consider that the present goal matrix is $X_{goal}^1$ and consider the possibility of the first switching only.
 Now, to say
 that the first switching never occurs is equivalent to have $\pdist( {\overline X}^\ell (T_f), X_{goal}^2) > \beta, \forall \ell \in \NN$. Since in the beginning
 of a step $\ell$ of RIGA one has $\mathcal V({\widetilde X}^\ell(T_f) = \pdist({\overline X}^{\ell-1}(T_f), {\overline X}^\ell(T_f))^2 =$
 $\pdist(X_g^1, {\overline X}^\ell(T_f))^2$. Now note that the triangular inequality implies that $\pdist(X_g^2, {\overline X}^\ell(T_f)) \leq \pdist(X_g^2, X_g^1) + \pdist(X_g^1, {\overline X}^\ell(T_f)) =$
 $\alpha + \pdist(X_g^1, {\overline X}^\ell(T_f))$. So, as the switching never occurs, Lemma \ref{lMain} implies that $\mathcal V({\widetilde X}^\ell)(\tau_0)$ will
 converges to zero.
  So, by Prop. \ref{pMonotonous}, there exist $\ell^*$ large enough such that  $\pdist(X_g^1, {\overline X}^\ell(T_f)) < \beta - \alpha$.
 This implies that, for $\ell=\ell^*$, $\pdist(X_g^2, {\overline X}^\ell(T_f)) < \alpha + \beta - \alpha = \beta$, which is the switching condition, and then is a contradiction.

 The reasoning for the other switchings is analogous.\\
 (e) Assume that RIGA is being executed between steps $\ell =1$ and $\ell_1$, where $\ell_1$ is the step for which the first switching occurs. As it is not know whether a switching occurs or not, $\ell_1$ may be infinite for the moment.  Fix $\theta \in (\gamma_0, 1)$. Let $\delta_\lambda = (\theta - \lambda_0)/p$.
  Let $\theta_q = \lambda_0 +  q \delta_\gamma, q=1, \ldots, p$. Note that $\theta_p = \theta$.
  Let $L_0 = \mu_0 + p \overline M \phi_\infty(\theta) \beta$. Let $\Lambda(L_0)$ be the value whose existence is assured in Lemma \ref{lMain}.
  Let ${\mathcal L}_0 =  \mu_0  < L_0$. Then Lemma \ref{lMain} will assure
  that the reference inputs will be uniformily bounded by ${\mathcal L}_1 = \mu_0 + \overline M \phi_\infty(\theta) \beta$.

 It is clear that, for $r$ large enough, $\theta^r \beta < \left( \frac{\delta_\lambda}{\Lambda \phi_{\infty}(\theta) } \right)^2 $. Let $t^* = \max {\overline J}_1$.  As $\lambda_0 < \theta$,
 it follows that
  $\mathcal V({\widetilde X}^0(t^*)) \leq \theta^r \mathcal V({\widetilde X}^0(0)) \leq \theta^r \beta  \leq \left( \frac{\delta_\lambda}{\Lambda \phi_{\infty}(\lambda_0) } \right)^2$.
  Lemma \ref{lMain} will imply that ${\overline X}^\ell$  is $\theta_1$-attractive in $(J, \BallV{c})$ for $J = J_{jq}$ for $1< q \leq p$ and $j=1, \ldots , r$. Furthermore this is also true for $J = {\overline J}_{p+1}$. It also assures that reference inputs will be uniformly bounded.
 That Lemma implies also that ${\mathcal V}({\widetilde X}^\ell(0)) \leq \theta^\ell \mathcal V({\widetilde X}^0(0))$, and by (d)
 this implies that a switching will occur for a finite $\ell$, namely, $\ell=\ell_1$.

 Assume that  RIGA is being executed, with $\ell$ between $\ell =\ell_1$ and $\ell_2$, respectively the first and the second steps for which a switching occurs, and $\ell_2$ may be infinite. Let $t^* = \max {\overline J}_2$.
  As $\theta_1 < \theta$,
 it follows that
  $\mathcal V({\widetilde X}^{\ell_1} (t^*))  \leq \theta^r \mathcal V({\widetilde X}^{\ell_1}(0)) \leq \theta^r \beta \left( \frac{\delta_\lambda}{\Lambda \phi_{\infty}(\theta_1) } \right)^2$.
  Lemma \refeq{lMain}  will imply that ${\overline X}^\ell$  is $\theta_2$-attractive in $(J, \BallV{c})$ for $J = J_{jq}$ for $2< q < p+1$ and $j=1, \ldots , r$. Furthermore it is also true for $J = {\overline J}_{p+1}$.
 Again Lemma \ref{lMain} invoked with the same $L_0$, (producing the same $\Lambda (L_0)$) and ${\mathcal L}_0 = \mu_0 + \overline M \phi_\infty(\theta) \beta$
 implies also that ${\mathcal V}({\widetilde X}^\ell(0)) \leq \theta^{\ell-\ell_1} \mathcal V({\widetilde X}^{\ell_1}(0))$ for all $\ell \leq \ell_2$,
 where $\ell_2$ is the step for which the second switching occurs. By (d), this implies that a new switching will occur for a finite $\ell_2$.
  Note that, for the $q$-th switching, the application of Lemma \ref{lMain} will consider
 ${\mathcal L}_0 = \mu_0 + q \overline M \phi_\infty(\theta) \beta$, and so on. As the maximum number of swithings is $p$, the maximum value of ${\mathcal L}_0$ for
 the application of Lemma \ref{lMain}  is given by $\mu_0 + p \overline M \phi_\infty(\theta) \beta = L_0$.

Reasoning in this way, as the maximum number of switchings is $p$,  the convergence of the algorithm for $T_f > (p r + 1) T_1$ is shown, as well as the fact that the generated reference inputs are uniformly bounded
by $L_0$.

\section{Proof of Lemma \ref{lMain}}
\label{alMain}

Consider the same notation of Definition \ref{dFund}.
Some further notations are necessary for writing this proof:
\begin{definition}
\label{dNotations_ell}
Let $J =[\tau_0, \tau_1] \subset [0, T_f]$.
 Assume that ${\overline X}^{\ell-1}(t)$ is $\lambda_{\ell-1}$-attractive in $(J, \BallV{c})$. It not shown yet that $\lambda_{\ell-1} < 1$, and so this fact is not claimed at this moment.
 One denotes $V_{\ell-1} = \mathcal V ({\widetilde X}^{\ell}(\tau_0))$,  $W_\ell = \mathcal V ({\widetilde X}^{\ell}(\tau_1)), \ell \in \NN$,
 and ${\mathcal V}_{\ell-1}= \mathcal V ({\widetilde X}^{\ell}(0))$.
One denotes ${\overline U}_\ell = \max \{ | {\overline u}^\ell_k(t) |~: k\in \{1, \ldots, m\}, t \in [0, T_f]\}$.
\end{definition}
From part (d) of proposition \ref{pMonotonous}, and from
 from the  fact that the Lyapunov function
is nonincreasing, it follows that
\begin{equation}
\label{ineqV}
{\mathcal V}_{\ell} \leq V_{\ell} \leq W_{\ell} \leq \lambda_{\ell-1} V_{\ell-1} \leq V_{\ell-1} \leq  {\mathcal V}_{\ell-1}.
\end{equation}
The following three Lemmas are instrumental:

\begin{lemma} \label{lUell} The following affirmations holds:\\
(a) There exists $\overline M >0$ such that  $|{\overline u}^\ell_k(t) - {\overline u}^{\ell-1}_k(t)| \leq  \overline M \sqrt{\mathcal V({\widetilde X}^\ell (t))}$.
for  all $ k\in \{1, \ldots, m\}$ and all $t \in [0, T_f]$.  In particular
\begin{equation}
\label{eUell}
{\overline U}_{\ell} \leq {\overline U}_{\ell-1} + \overline M \sqrt{{\mathcal V}_{\ell-1}}
\end{equation}
and
\begin{equation}
\label{eUell2}
|{\overline u}^\ell_k(t) - {\overline u}^{\ell-1}_k(t)| \leq  \overline M \sqrt{V_{\ell-1}},
  \forall t \in J=[\tau_0, \tau_1]
\end{equation}
(b) One has $\|\left({\overline X}^\ell(t) - {\overline X}^{\ell-1}(t)\right)E\| \leq 2 \sqrt{{\mathcal V(\widetilde X (t))}}$ for all $t \in [0, T_f]$. In particular,
$\|\left({\overline X}^\ell(t) - {\overline X}^{\ell-1}(t) \right)E\| \leq 2 \sqrt{{ V}_{\ell-1}}$, for all $t \in J=[\tau_0, \tau_1]$.\\
\end{lemma}

\begin{proof}
From equation \refeq{eFeedbackLaw}, from the fact that $E^\dag {\widetilde S}_k E$ is skew hermitian (and so $\Re [\trace(E^\dag {\widetilde S}_k E)]=0$)
from Prop. \ref{pCritical}, and from  Proposition \ref{ineqFrob}, it follows  that
\begin{small}
\begin{eqnarray*}
\| \widetilde{u}_k (t)\| & =& \left\| 2 K \Re \left[ \trace \left(E^\dag {\widetilde S}_k(t) (\widetilde X(t) - I) E \right) \right] \right\|  \\
 & \leq & 2 K \|{\overline X}(t)^\dag S_k X(t) \| \|(\widetilde X(t) - I) E \|\\
 & = & 2  K \|S_k\| \sqrt{\mathcal V(\widetilde X(t))}
\end{eqnarray*}
\end{small}
From this, and from the fact that $u_k(t) = {\overline u}_k (t) + {\widetilde u}_k(t)$, then (a) follows easily.
Now to show (b), note that from Prop. \ref{ineqFrob} that:
$\|({\overline X}^\ell(t) - {\overline X}^{\ell-1}(t))E\| =$ $ \| [{\overline X}^\ell(t) - { X}^{\ell}(t)  + { X}^{\ell}(t) - {\overline X}^{\ell-1}(t)] E\| \leq$
$ \| ({\overline X}^\ell(t) - { X}^{\ell}(t)) E\| + \| ({ X}^{\ell}(t)- {\overline X}^{\ell-1}(t))E\| =$
$ \| ( X^\ell(t) R_{\ell+1}  - { X}^{\ell}(t)) E\| + \|({ X}^{\ell}(t) - {\overline X}^{\ell-1}(t))E\| =$
$ \| (R_{\ell+1} - I) E  \| + \sqrt{\mathcal V(\widetilde X(t))}$. Now, since $R_{\ell+1} = (X_f^{\ell})^\dag X_{goal}$, then by Prop. \refeq{ineqFrob}
$\| (R_{\ell+1} - I) E  \| = \| ((X_f^{\ell})^\dag X_{goal} - I) E\| \leq \|(X_{goal} - X_f^{\ell})E\| =  \|({\overline X}^{\ell-1}(T_f) - X^{\ell}(T_f) )E\| =$
$ \sqrt{\mathcal V(\widetilde X^\ell(T_f))} \leq \sqrt{\mathcal V(\widetilde X^\ell(t))}$, for $t \in [0, T_f]$.

\end{proof}

\begin{lemma}
\label{lCB} Assume that\footnote{The proof of Lemma \ref{lCB} shows that the parameter $\Lambda$ that appears in \refeq{estar} depends on the set $\{(\overline L), c, T_1, K, \|S_k\|: k=1,\ldots , m\}$.} ${\mathcal V}_0 \leq c^2 <4$.  If  ${\overline U}_{\ell-1} \leq \overline L$  for some $\overline L >0$ then
there exists\footnote{The proof of Lemma \ref{lCB} shows that $\Lambda$ depends on the set $\{(\overline L), c, T_1, K, \|S_k\|: k=1,\ldots , m\}$.} $\Lambda(\overline L) >0$ such that
\begin{subequations}
 \label{def_gamma_V}
 \begin{equation}
 \label{estar}
\lambda_\ell  \leq  \lambda_{\ell-1} + \Lambda(\overline L) \sqrt{V_{\ell-1}}.
\end{equation}
 Furthermore:
 \begin{eqnarray}
 \label{estarstarstar} {\mathcal  V}_{\ell} &\leq & \lambda_{\ell-1} {\mathcal V}_{\ell};\\
 \label{estarstar} { V}_\ell &\leq & \lambda_{\ell-1} { V}_{\ell-1}.
 \end{eqnarray}
 \end{subequations}
 \end{lemma}

 \begin{lemma}
  \label{eSequences}
  Assume that ${\overline U}_\ell, {\mathcal V}_\ell$, ${ V}_\ell$ and $\gamma_\ell, \ell \in \NN$ are nonnegative real sequences
  such that \refeq{eUell}, \refeq{estar}, \refeq{estarstarstar},  \refeq{estarstar} holds (\refeq{estar} holds whenever ${\overline U}_{\ell-1} \leq \overline L$).
  Fix a positive real numbers $\delta_\lambda$
  and let $\theta = \lambda_0 + \delta_\lambda$. Assume that $\theta \in (0 ,1)$ and ${\overline U}_{0} \leq L_0$,
  Let $\overline L = L_0 +  \overline M \phi_\infty(\theta) \sqrt{{\mathcal V}_0}$. Let\footnote{The existence of $\Lambda (\overline L)$ is assumed
  by \refeq{estar})} ${\overline \Lambda} = \Lambda (\overline L)$.
 Assume that $\sqrt{V_0} \leq \min \{ \frac{\delta_\lambda}{{\overline \Lambda} \phi_\infty (\theta)}, \sqrt{c}\}$.
 Then, for all $\ell \in \NN$ one has ${\overline U}_\ell \leq \overline L$, $\sqrt{{\mathcal V}_\ell} \leq (\sqrt{\theta})^\ell \sqrt{{\mathcal V}_{\ell-1}}$,
  $\sqrt{{ V}_\ell} (\sqrt{\theta})^\ell \leq \sqrt{{ V}_{\ell-1}}$, and $\lambda_\ell \leq \theta$, for all $\ell \in \NN - {0}$.
 \end{lemma}

\begin{proof} (of Lemma \ref{lMain})
The proof of Lemma \ref{lMain} is then a direct application of Lemmas \ref{lUell}, \ref{lCB} and \ref{eSequences}.
\end{proof}

\subsection{Proof of Lemma \ref{eSequences}}
\label{aSequences}

It will be shown by induction that, for all $\ell\in \NN$:\\
 (i) $\sum_{j=0}^{\ell-1} \sqrt{V_j} \leq \phi_{\ell-1}(\theta) \sqrt{V_0}$,\\
 (ii) $\lambda_\ell \leq \lambda_0 + \Lambda \phi_{\ell-1} (\theta) \sqrt{V_0} \leq \theta$,\\
 (iii) $\sqrt{V_\ell} \leq (\sqrt{\theta})^\ell \sqrt{V_0}$.\\
  (iv) $\sqrt{{\mathcal V}_\ell} \leq (\sqrt{\theta})^\ell \sqrt{{\mathcal V}_0}$.\\
  (v) ${\overline U}_\ell \leq {\overline U}_0 + \overline M \phi_{\ell-1}(\theta) \sqrt{{\mathcal V}_0}$\\

In fact, for $\ell = 1$, as $\phi_0(\theta) =1$, then (i) is trivial. Furtheremore, since $\phi_0(\theta) =1$, \refeq{estar}
for $\ell =1$ reads $\lambda_1 \leq \lambda_0 + \Lambda \phi_0(\theta) \sqrt{V_0}$. Since $\sqrt{V_0} \leq
\delta_\lambda / (\Lambda \phi_\infty(\theta))$, then  $\lambda_1 \leq \lambda_0 + \Lambda \phi_0(\theta) \sqrt{V_0} \leq  \lambda_0 + \frac{\Lambda \phi_0(\theta) \delta_\lambda}{ \Lambda \phi_\infty(\theta)} \leq
\lambda_0 + \delta_\lambda \leq \theta$, showing (ii) for $\ell=1$. Now, since $\lambda_0 \leq \theta$,
from \refeq{estarstar} for $\ell  =1$, then (iii) holds for $\ell=1$. From \refeq{ineqV}, it follows that
$\sqrt{{\mathcal V}_1} \leq \sqrt{\lambda_0} \sqrt{V_0} \leq \sqrt{V_0} \leq \sqrt{{\mathcal V}_0}$. In particular
$\sqrt{{\mathcal V}_1} \leq \sqrt{\lambda_0} \sqrt{{\mathcal V}_0}$, showing (iv) for $\ell=1$. Now note that \refeq{eUell} for $\ell=1$
coincides to (v) for $\ell =1$.

Now assume  that (i), (ii) and (iii) hold for $\ell$. It will be shown that
they hold for $\ell+1$. Note that, from (i) and (iii) and from the fact
that $\phi_{\ell-1}(\theta) + (\sqrt{\theta})^\ell = \phi_{\ell}(\theta)$, then  (i) follows
for $\ell +1$. Now from  \refeq{estar} for $\ell+1$, from (ii) and (iii) that
$
\gamma_{\ell+1}  \leq  \lambda_\ell + \Lambda \sqrt{V_\ell}$
               $ \leq  \lambda_0
 + \Lambda \left\{ \phi_{\ell-1}(\theta) \sqrt{V_{0}} +   \sqrt{V_{\ell}} \right\}$
 $ \leq   \lambda_0
 + \Lambda \left\{ \phi_{\ell-1}(\theta)  +   (\sqrt{\theta})^\ell \right\} \sqrt{V_{0}} $
 $ = \lambda_0
 + \Lambda  \phi_{\ell}(\theta) \sqrt{V_{0}}
 $
Since $\sqrt{V_0} \leq
\delta_\lambda / (\Lambda \phi_\infty(\theta))$, then  $\lambda_{\ell+1} \leq \lambda_0 + \Lambda \phi_\ell(\theta) \sqrt{V_0} \leq  \lambda_0 + \frac{\Lambda \phi_\ell(\theta) \delta_\lambda}{ \Lambda \phi_\infty(\theta)} \leq
\lambda_0 + \delta_\lambda \leq \theta$, showing (ii) for $\ell+1$.

Now, from \refeq{estarstar} for $\ell+1$, from (ii) and from (iii), then the fact that $\sqrt{V_{\ell+1}} \leq \lambda_\ell \sqrt{V_{\ell}} \leq \sqrt{\theta} \sqrt{V_{\ell}} \leq (\sqrt{\theta})^{\ell+1} \sqrt{V_0}$ follows,
showing (iii) for $\ell+1$. Now, from \refeq{ineqV}, it follows easily that (iv) holds fo $\ell+1$. Now assume that
(v) holds. From (v), (iv), from \refeq{eUell},
and from the fact that $\phi_{\ell} (\theta) = \phi_{\ell-1} (\theta) + \sqrt{\theta})^\ell$,
it follows that ${\overline U}_{\ell+1} \leq $ ${\overline U}_{\ell} + \overline M \sqrt{ {\mathcal V}_\ell} \leq$
${\overline U}_{0} + \overline M \phi_{\ell-1} (\theta) \sqrt{ {\mathcal V}_0} + \overline M \sqrt{ {\mathcal V}_\ell} \leq$
${\overline U}_{0} + \overline M \phi_{\ell-1} (\theta) \sqrt{ {\mathcal V}_0} + \overline M (\sqrt{\theta})^\ell \sqrt{ {\mathcal V}_0}$
${\overline U}_{0} + \overline M \phi_{\ell} (\theta) \sqrt{ {\mathcal V}_0}$.

\subsection{Proof of Lemma \ref{lCB}}

After right-multiplication of \refeq{cqs} and \refeq{reference} by $E$, one obtains
\begin{eqnarray*}
\dot{\overline Y}(t) & = & \left( S_0 + \sum_{k=1}^{m}{\overline u}_k(t) S_k \right) \overline Y(t)\\
\dot{ Y}(t) & = & \left( S_0 + \sum_{k=1}^{m}{ u}_k(t) S_k \right) Y(t)
\end{eqnarray*}
where $Y(t) = X(t) E$ and $\overline Y(t) = \overline X (t) E$ are $n \times \nbar$ matrices.
Note that the feedback ${\widetilde u}_ (t) = u_k(t) - {\overline u}_k (t)$ given by \refeq{eFeedbackLaw} can be rewritten as:
\begin{subequations}
\begin{eqnarray}
 \nonumber {\widetilde u}_ (t)  & = &  2 K \Re \left[ \trace \left( E^\dag  {\widetilde S}_k {\overline X(t)}^\dag X(t) E \right) \right]\\
 \nonumber & = & 2 K \Re \left[ \trace \left( E^\dag {\overline X(t)}^\dag {S}_k {\overline X}(t) {\overline X(t)}^\dag X(t) E \right) \right]\\
\nonumber  & = & 2 K \Re \left[ \trace \left( {\overline Y(t)}^\dag {S}_k Y(t) \right) \right]\\
\nonumber  & = & 2 K  \Re \left[ \trace \left( {\overline Y(t)}^\dag {S}_k (Y(t) - \overline Y(t)) \right) \right]\\
\label{eFeedZ} & = & 2 K  \Re \left[ \trace \left( {\overline Y(t)}^\dag {S}_k Z(t) \right) \right]
\end{eqnarray}
\end{subequations}
where the last equality follows from the fact that $S_k$ is anti-hermitian, and so
 $\Re \left[ \trace \left( {\overline Y(t)}^\dag {S}_k \overline Y(t)) \right) \right] =0$.
Let $Z(t) = (Y(t) - \overline Y(t))$. After some simple manipulations, it follows that
$
\dot Z(t) = f({\overline u}(t), {\overline Y(t)}, Z(t))
$
where
\begin{equation}
\begin{array}{c}
\label{eFuYZ}
f({\overline u}(t), \overline Y, Z(t))  = \left( S_0 + \sum_{k=1}^{m}{\overline u}_k(t) S_k \right) Z(t) \\
 +  \sum_{k=1}^{m}  {\widetilde u}_k(t) S_k \left( Z(t) + \overline Y(t) \right)
\end{array}
\end{equation}
with  ${\widetilde u}_k(t)$ given by \refeq{eFeedZ}.
Note that the initial condition $X_0$ in RIGA is always the identity, this means
that $Z_0 = Y(0) - {\overline Y}(0) = (I - {\overline X}_0) E$. Since $\|Z\|^2 = \dist (Y, \overline Y)^2$, this
justifies the following notation:
\[
\begin{array}{c}
\BallZ{c} = \{ Z_0 \in \CC^{n\times \nbar}~|~\|Z_0 \| \leq c, \\
Z_0 = (I - {\overline X}_0) E, {\overline X}_0 \in \Un\}
\end{array}
\]

The proof of Lemma \ref{lCB} is based on part (b) of the following lemma:
\begin{lemma}
\label{lJJ}
Let $J=[\tau_0, \tau_f] \subset [0, T_f]$ with $\tau_f - \tau_0 = T_1$.
Assume that  two sets of continuous reference inputs ${\overline u}^1 = \{{\overline u}_k^1(t) : J \rightarrow \RR, k=1, \ldots, m\}$,
 ${\overline u}^2 = \{{\overline u}_k^2(t) : J \rightarrow \RR, k =1 , \ldots, m$ are given.
Assume that ${\overline Y}_i(t), i=1,2$
are the corresponding solutions of
\begin{equation}
 \label{eREF}
   \dot{\overline Y}^i(t)  =  \left( S_0 + \sum_{k=1}^{m}{\overline u}_k^i(t) S_k \right) {\overline Y}^i(t), \quad  {\overline Y}^i(\tau_0) =  {\overline Y}^i_0\\
\end{equation}
Let $f_i(t, Z(t)) = f({\overline u}^i(t), {\overline Y}^i(t), Z(t))$, where $f(u(t), \overline Y, Z(t))$ is defined by \refeq{eFuYZ}.
 Consider the corresponding solution ${Z}_i : J \rightarrow \Un$ of the  closed-loop system
\begin{equation}
\label{eCLOSED}
  \dot Z (t)= f_i(t, Z(t)),\quad Z(\tau_0) = Z_0, i=1,2
\end{equation}
both with the same initial condition
$Z(\tau_0) = Z_0 \in \BallZ{c}$, but with different reference trajectories ${\overline Y}_i(t)$, $i=1, 2$, which are the corresponding solution of \refeq{eREF} with initial condition ${\overline Y}^i_0$, respectively.
 Assume that there exists $\Delta_1 >0$ such that  $\|{\overline Y}_2(t) -  {\overline Y}_1(t)\| \leq \Delta_1, \forall t \in J$. Assume that $\Delta_1 < 2$.
 Assume that $\max \{ \|{\overline u}_k^2(t) - {\overline u}_k^1(t)\|~|~t \in J, k=1, \ldots, m\} \leq \Delta_2$. Assume that  $\max \{ \|{\overline u}_k^1(t)\|~|~t \in J, k=1, \ldots, m\} \leq \overline L$.\\
(a)  There exist positive real numbers ${\mathcal M}_1, {\mathcal M}_2$, depending on $\overline L$, such that $\|{Z}_1(t) - {Z}_2(t)\| \leq \left( {\mathcal M}_1 \Delta_1  + {\mathcal M}_2 \Delta_2\right) \|Z_0\|$.\\
(b) Let  $\alpha_i (Z_0) =  \sum_{k=1}^{m} \int_{\tau_0}^{\tau_1}  4 K \trace [  {\overline Y}_i(t)^\dag S_k Z_i(t)]^2 dt$, for $i=1, 2$. There exist $N_1, N_2  > 0$, depending on $\overline L$, such that $\delta_\alpha= | \alpha_1 - \alpha_2 | \leq T_1 (N_1 \Delta_1 + N_2 \Delta_2) \|Z_0\|^2$.
\end{lemma}

\begin{remark} \label{rZ0X0} If $Z_0$ is fixed, then the solution $Z_i(t)$ of \refeq{eCLOSED} is well defined,
and hence $\alpha_i$ is a function of $Z_0$.
Now initial conditions $X_0=I$ an ${\overline X}_0$ such that $Z_0 = X_0 E - {\overline X}_0$ determine
${\widetilde X}_0$, and so determine the solution $\widetilde X (t)$ of the closed loop system  \refeq{eClosedLoopComplete}.
Furthermore it is clear from \refeq{eFeedZ} that,
 $\alpha_i =
\int_{0}^{T_f} \sum_{k=1}^m 4 K \left\{\Re \left[ \trace \left( E^\dag  {\widetilde S}_k \widetilde X(t) E \right) \right]\right\}^2 dt$,
and so, $\alpha_i$ also a function of ${\widetilde X}_0$ when considering that the reference trajectory $\overline X(t)$ is shuch that
${\overline Y}_i = {\overline X} (t) E$.
\end{remark}

  Lemma \ref{lCB} will be shown before proving Lemma \ref{lJJ}.

\begin{proof} (Of Lemma \ref{lCB})
 Assuming that ${\overline X}^{\ell-1}$ is
$\lambda_{\ell-1}$-attractive in $(J, \BallV{c})$, it is sufficient to  show that \refeq{estar} holds,
since the other inequations follows from \refeq{ineqV}.

Let $\Gamma(t) = {\overline Y}_1(t) - {\overline Y}_2(t)$ By part (b) of Lemma \ref{lUell}, then $\Delta_1 = \max_{t \in [\tau_0, \tau+1]} \{ \|\Gamma (t)\| \} \leq 2 \sqrt{V_{\ell-1}}$.
By part (a) of Lemma \ref{lUell}, $\Delta_2 = \max_{t \in [\tau_0, \tau+1]} \| u_k^2(t) - u_k^1(t)\| \leq \overline M  \sqrt{V_{\ell-1}}$. Hence, part (b) of Lemma \ref{lJJ} implies that
$|\alpha_2 - \alpha_1| \leq T_1 (N_1 \Delta_1 + N_2 \Delta_2) \| Z_0\|^2$
$\leq T_1 ( 2 N_1  + N_2 \overline M) \sqrt{V_{\ell-1}} \| Z_0\|^2$.
Now define $\Lambda = T_1 ( 2 N_1  + N_2 \overline M)$. By Prop. \ref{pCritical}, it is clear
that $\|Z_0\|^2 = \mathcal V({\widetilde X}_0)$.

  As in the proof of Prop. \ref{pEquivalence}, it follows that, for all ${\widetilde X}_0   \in \BallV{c}$ one has
  $\mathcal V({\widetilde X}_i(\tau_1)) -   \mathcal V({\widetilde X}_i(\tau_0)) = \mathcal V({\widetilde X}_i(\tau_1))- \mathcal V({\widetilde X}_0))
  = - \alpha_i, i=1, 2$. By definition of $\lambda_i$, then $\mathcal V({\widetilde X}_i(\tau_1)) \leq \lambda_i \mathcal V({\widetilde X}_0)$.
  Let $A_1 = 1 - \lambda-1$. Then $-\alpha_1 = \mathcal V({\widetilde X}_1(\tau_1)) -   \mathcal V({\widetilde X}_1(\tau_0)) \leq
  $ $ -A_1 \mathcal V({\widetilde X}_1(\tau_0))$. Hence $-\alpha_1 \leq -A_1 \mathcal V({\widetilde X}_1(\tau_0))$. Assuming that
  ${\widetilde X}_1(\tau_0))= {\widetilde X}_2(\tau_0)) = {\widetilde X}_0$, then, as $\alpha_2  \geq \alpha_1 - |\alpha_1-\alpha_2| \geq  A_1 \mathcal V({\widetilde X}_0) + |\alpha_1 - \alpha_2|$
  it follows that $\alpha_2 \geq (1-\lambda_1)  \mathcal V({\widetilde X}_0) - \Lambda  \sqrt{V_{\ell-1}} \mathcal V({\widetilde X}_0)$. So
  $\mathcal V({\widetilde X}_2(\tau_1)) -   \mathcal V({\widetilde X}_0))  = -\alpha_2  \leq - (1 - \lambda_1 - \Lambda  \sqrt{V_{\ell-1}} ) \mathcal V({\widetilde X}_0))$ and
  so  $\mathcal V({\widetilde X}_2(\tau_1))   \leq  (\lambda_1 + \Lambda  \sqrt{V_{\ell-1}} ) \mathcal V({\widetilde X}_0))$. As the last inequality holds
  for all ${\widetilde X}_0 \in \BallV{c}$, it follows that ${\overline X}^\ell (t)$ is $\lambda_\ell$-attractive, and $\lambda_\ell$ is such that $\lambda_ell \leq
  (\lambda_{\ell-1} + \Lambda  \sqrt{V_{\ell-1}}$,
  showing Lemma \ref{lCB}.
 \end{proof}

\begin{proof} (Of Lemma \ref{lJJ})
The proof of part (a) of  this Lemma is based  on the following ideas:\\
(i) It will be shown first that $f_1(t, Z)$ is Lipschitz with respect to the second variable, that is
there exists $\mathcal L>0$ such that
\begin{equation}
\label{eLf}
\|f_1 (t, Z_1) - f_1(t, Z_2) \| \leq \mathcal L \| {Z}_1 - {Z}_2\|,
\end{equation}
 for
all ${Z}_1, {Z}_2 \in \BallZ{c}$ and $t \in J$.\\
(ii) Then it will be shown that $\| f_1(t, Z) - f_2(t, Z) \| \leq ({\mathcal N}_1 \Delta_1 + {\mathcal N}_2 \Delta_2) \| Z \|$ for all $t \in J$ and $Z \in  \BallZ{c}$.
all $\widetilde X$ such that $\mathcal V(\widetilde X) \leq  \mathcal V({\widetilde X}_0)$ for some $\mu >0$.\\
Now, using (i) and (ii), and using the fact that the Lyapunov function  of the closed-loop system
is always nonincreasing, it follows from  \cite[Theorem 2.5, p.79]{Kha02} that $\| { Z}_1(t) - {\widetilde Z}_2(t) \| \leq$ $\frac{({\mathcal N}_1 \Delta_1 + {\mathcal N}_2 \Delta_2) \| Z_0 \|} {\mathcal L} \left\{ \exp(\mathcal L T_1) - 1 \right\}$, showing  part (a) the Lemma.\\
Now, the proofs of (i) and (ii) shall be presented.
To show (i) note that:\\
$ f_1(t,Z_1) - f_1(t,Z_2) = (S_0 + \sum_{i=1}^{m}   {\overline u}_k^1 S_k) (Z_1 - Z_2) +$ $\sum_{i=1}^{m}  \left( {\widetilde u}^1_k S_k Z_1 - {\widetilde u}^2_k S_k Z_2 \right) +$
$\sum_{i=1}^{m}  ({\widetilde u}^1_k - {\widetilde u}^2_k) S_k {\overline Y}_1$, where ${\widetilde u}_k^i = 2  K \Re [ \trace( {\overline Y}_1^\dag S_k Z_i)]$.
Since $|{\overline u}_k^1 | \leq \overline L$, it follows that $\|(S_0 + \sum_{i=1}^{m}   {\overline u}_k^1 S_k) (Z_1 - Z_2) \| \leq [\|S_0\| + \overline L \sum_{i=1}^{m} \|S_k\|] \|Z_1 - Z_2\|$.
As in the proof of Lemma \ref{lUell}, one has
\begin{equation}
\label{inequtilde}
|{\widetilde u}^i_k| \leq 2 K \|S_k\| \|Z_i\| \leq 2 K c \|S_k\|,
\end{equation}
 hence $\|\sum_{i=1}^{m}  \left( {\widetilde u}^1_k S_k Z_1 - {\widetilde u}^2_k S_k Z_2 \right) \| \leq$
$\|\sum_{i=1}^{m}  \left( {\widetilde u}^1_k S_k Z_1 -  {\widetilde u}^1_k S_k Z_2 +  {\widetilde u}^1_k S_k Z_2 - {\widetilde u}^2_k S_k Z_2 \right)\| \leq$
$\sum_{i=1}^{m} \| \left( {\widetilde u}^1_k S_k Z_1 -  {\widetilde u}^1_k S_k Z_2 \| +  \|{\widetilde u}^1_k S_k Z_2 - {\widetilde u}^2_k S_k Z_2 \| \right)$.
Then, $\|  {\widetilde u}^1_k S_k Z_1 -  {\widetilde u}^1_k S_k Z_2 \| \leq |u^1_k| \|S_k\| \| Z_1 - Z-2\| \leq 2 K c \|S_k\|^2 \|Z_1- Z_2\|$.
Furthermore,  $\|{\widetilde u}^1_k S_k Z_2 - {\widetilde u}^2_k S_k Z_2 \| \leq \|{\widetilde u}^1_k - {\widetilde u}^2_k\| \|S_k\| \|Z_2\|$.
By Proposition \ref{ineqFrob}, $|{\widetilde u}^1_k  - {\widetilde u}^2_k | = $ $2 K \Re[ \trace ({\overline Y}_1 S_k (Z_1 - Z_2) )]$
$\leq 2 K \|S_k\| \|Z_1 - Z_2\|$.  As $\|Z_2\| \leq c$, then $\|{\widetilde u}^1_k S_k Z_2 - {\widetilde u}^2_k S_k Z_2 \| \leq 2 K c \|S_k\|^2 \|Z_1 -Z_2\|$.
it follows that
$\|\sum_{i=1}^{m}  \left( {\widetilde u}^1_k S_k Z_1 - {\widetilde u}^2_k S_k Z_2 \right) \| \leq$
$ 4 K c [\sum_{i=1}^{m} \|S_k\|^2 ] \| Z_1 - Z_2\|$. It follows that
$\| f_1(t,Z_1) - f_1(t,Z_2)\| \leq \mathcal L \| Z_1 - Z_2\|$, where
 $\mathcal L =  [ \|S_0\| + \overline L (\sum_{i=1}^{m}\|S_k \|) + 4 K c (\sum_{i=1}^{m}\|S_k \|^2)]$, showing (i).

Now, to show (ii) note that $\| f_1 (t, Z) - f_2(t, Z) \| = \|\sum_{i=1}^{m} [ ({\overline u}_k^2 - {\overline u}_k^1) S_k Z] +$  $\sum_{i=1}^{m}  [{\widetilde u}_k^2 - {\widetilde u}_k^1)] S_k Z +$ $({\widetilde u}_k^1 S_k {\overline Y}_1 - ({\widetilde u}_k^2 S_k {\overline Y}_2)\|$.
Note that $\| ({\overline u}_k^2 - {\overline u}_k^1) S_k Z]  \leq  \|S_k\| \Delta_2 \|Z\|$.
Since $\| [{\widetilde u}_k^2 - {\widetilde u}_k^1)] \| \leq 2 K \|\Re [ \trace({\overline Y}_1 S_k Z - {\overline Y}_2 S_k Z)]\| \leq 2 K \Delta_1 \|S_k\| \|Z\|$, then. as $\|Z \| \leq c$,
 $\| [{\widetilde u}_k^2 - {\widetilde u}_k^1)] S_k Z\| \leq $ $ 2 K \Delta_1  \|S_k\|^2 ) \|Z\| \|Z\| \leq  2 K c \|S_k\|^2 \Delta_1 \|Z\|$.
 Furthermore $\|{\widetilde u}_k^1 S_k {\overline Y}_1 - {\widetilde u}_k^2 S_k {\overline Y}_2\| = \|{\widetilde u}_k^1 S_k {\overline Y}_1 - {\widetilde u}_k^1 S_k {\overline Y}_2 +
 {\widetilde u}_k^1 S_k {\overline Y}_2 - ({\widetilde u}_k^2 S_k {\overline Y}_2\| \leq$ $\|{\widetilde u}_k^1 S_k {\overline Y}_1 - {\widetilde u}_k^1 S_k {\overline Y}_2\| +  $  $\|{\widetilde u}_k^1 S_k {\overline Y}_2 - ({\widetilde u}_k^2 S_k {\overline Y}_2\|$. Now  $\|{\widetilde u}_k^1 S_k {\overline Y}_1 - {\widetilde u}_k^1 S_k {\overline Y}_2\| \leq$ $|{\widetilde u}_k^1| \|S_k\| \| Y_1 - Y_2\| \leq$ $2 K \|S_k\|^2 \Delta_1 \|Z\|$.
 Note that,   $\|{\widetilde u}_k^1 S_k {\overline Y}_2 - ({\widetilde u}_k^2 S_k {\overline Y}_2\| \leq$ $|{\widetilde u}_k^1 - {\widetilde u}_k^2| \|S_k\| \leq $ $ 2 K c\|S_k\|^2  \Delta_1 \|Z\|$.
  Hence, as $\|Z\|$ is non-increasing it follows that
 $\| f_1 (t, Z) - f_2(t, Z) \| \leq ({\mathcal N}_1 \Delta_1) + {\mathcal N}_2 \Delta_2) \|Z_0\|$, where
 ${\mathcal N}_1 =  2 K (2 c +2) (\sum_{k=1}^m \|S_k\|^2)$
 ${\mathcal N}_2 = (\sum_{k=1}^m \|S_k\|)^2$.

 Now, to show (b), note that $\alpha_i = \sum_{i=1}^m \int_{\tau_0}^{\tau_1} \beta_{i_k}(t)^2 dt$, where $\beta_{i_k} = \Re [ \trace ( Y_i^\dag(t) S_k Z_i(t))]$. So,
 $\|\alpha_1 - \alpha_2\| \leq T_1 \max_{t \in J} \sum_{k=1}^m \|\beta_{1_k}^2 - \beta_{2_k}^2\| \leq T_1 \max_{t \in J} \sum_{k=1}^m \| \beta_{1_k} + \beta_{2_k} \| \| \beta_{1_k} - \beta_{2_k}\|$.
 Note that $\| \beta_{1_k} - \beta_{2_k}\| \leq \| 4 K \Re [ \trace ( Y_1^\dag(t) S_k Z_1(t) -   Y_2^\dag(t) S_k Z_1(t) +  Y_2^\dag(t) S_k Z_1(t) -Y_2^\dag(t) S_k Z_2(t))] \leq$
 $ 4 K \|(Y_1^\dag(t)-Y_2^\dag(t)) S_k Z_1(t) +$ $\|Y_2^\dag(t) S_k (Z_1(t) - Z_2)\| \leq 4 K \left\{\Delta_1 \|S_k\| \|Z_0\| + ({\mathcal M}_1 \Delta_1 + {\mathcal M}_2 \Delta_2) \|S_k \|\right\} \|Z_0\|$.
 Furthermore $\| \beta_1 + \beta_2 \| \leq 4 K \Re [ \trace ( Y_1^\dag(t) S_k Z_1(t) +  Y_2^\dag(t) S_k Z_2(t))] \leq$ $ \sum_{k=1}^{m} 8 K \|S_k\| \|Z_0\|$.
 In particular, $|\alpha_1 - \alpha _2\| \leq \{ {\mathcal N}_1 \Delta_1 + {\mathcal N}_2 \Delta_2\} \sqrt{V_{\ell-1}} \|Z_0\|^2$,
 where ${\mathcal N}_2 = 2 K (\sum_{k=1}^{m} \|S_k\| )$ and ${\mathcal N}_1 = 4 K (\sum_{k=1}^{m} \|S_k\| )^2 + 2 K  (\sum_{k=1}^{m} \|S_k\| )$,
showing (ii).
\end{proof}

\def\cprime{$'$}


\begin{thebibliography}{10}
	
	\bibitem{PalaoK2002PRL}
	J.~P. Palao and R.~Kosloff.
	\newblock ``Quantum computing by an optimal control algorithm for unitary
	transformations''.
	\newblock
	\href{https://dx.doi.org/https://doi.org/10.1103/PhysRevLett.89.188301}{Phys.
		Rev. Lett. {\bf 89}, 188301--}~(2002).
	
	\bibitem{PalaoK2003PRA}
	J.~P. Palao and R.~Kosloff.
	\newblock ``Optimal control theory for unitary transformations''.
	\newblock
	\href{https://dx.doi.org/https://doi.org/10.1103/PhysRevA.68.062308}{Phys.
		Rev. A {\bf 68}, 062308--}~(2003).
	
	\bibitem{SchirF2011NJoP}
	S.~G. Schirmer and P.~de~Fouquieres.
	\newblock ``Efficient algorithms for optimal control of quantum dynamics: the
	krotov method unencumbered''.
	\newblock \href{https://dx.doi.org/10.1088/1367-2630/13/7/073029}{New Journal
		of Physics {\bf 13}, 073029--}~(2011).
	
	\bibitem{L1}
	N.~Yamamoto, K.~Tsumura, and S.~Hara.
	\newblock ``Feedback control of quantum entanglement in a two-spin system''.
	\newblock
	\href{https://dx.doi.org/http://dx.doi.org/10.1016/j.automatica.2006.12.008}{Automatica
		{\bf 43}, 981 -- 992}~(2007).
	
	\bibitem{L2}
	M.~Mirrahimi.
	\newblock ``Lyapunov control of a quantum particle in a decaying potential''.
	\newblock In Annales de l'Institut Henri Poincare (C) Non Linear Analysis.
	\newblock
	\href{https://dx.doi.org/https://doi.org/10.1016/j.anihpc.2008.09.006}{Volume~26,
		pages 1743--1765}.
	\newblock Elsevier~(2009).
	
	\bibitem{L3}
	M.~Mirrahimi, P.~Rouchon, and G.~Turinici.
	\newblock ``Lyapunov control of bilinear schrödinger equations''.
	\newblock
	\href{https://dx.doi.org/https://doi.org/10.1016/j.automatica.2005.05.018}{Automatica
		{\bf 41}, 1987--1994}~(2005).
	
	\bibitem{L4}
	S.~Grivopoulos and B.~Bamieh.
	\newblock ``Lyapunov-based control of quantum systems''.
	\newblock In Decision and Control, 2003. Proceedings. 42nd IEEE Conference on.
	\newblock
	\href{https://dx.doi.org/https://doi.org/10.1109/CDC.2003.1272601}{Volume~1,
		pages 434--438}.
	\newblock IEEE~(2003).
	
	\bibitem{L5}
	J.~Zhang, Y.-X. Liu, R.-B. Wu, K.~Jacobs, and F.~Nori.
	\newblock ``Quantum feedback: Theory, experiments, and applications''.
	\newblock
	\href{https://dx.doi.org/https://doi.org/10.1016/j.physrep.2017.02.003}{Physics
		Reports {\bf 679}, 1--60}~(2017).
	
	\bibitem{L6}
	Yu~Pan, V.~Ugrinovskii, and M.~R. James.
	\newblock ``Lyapunov analysis for coherent control of quantum systems by
	dissipation''.
	\newblock In 2015 American Control Conference (ACC).
	\newblock \href{https://dx.doi.org/10.1109/ACC.2015.7170718}{Pages 98--103}.
	\newblock ~(2015).
	
	\bibitem{LS1}
	D.~Dong and I.~R. Petersen.
	\newblock ``Quantum control theory and applications: a survey''.
	\newblock \href{https://dx.doi.org/10.1049/iet-cta.2009.0508}{IET Control
		Theory Applications {\bf 4}, 2651--2671}~(2010).
	
	\bibitem{SilPerRou14}
	H.~B. Silveira, P.~S. Pereira~da Silva, and P.~Rouchon.
	\newblock ``Quantum gate generation by {T}-sampling stabilization''.
	\newblock
	\href{https://dx.doi.org/https://doi.org/10.1080/00207179.2013.873951}{International
		Journal of Control {\bf 87}, 1227--1242}~(2014).
	
	\bibitem{SilPerRou16}
	H.~B. Silveira, P.~S. Pereira~da Silva, and P.~Rouchon.
	\newblock ``Quantum gate generation for systems with drift in {U}(n) using
	{L}yapunov-{L}asalle techniques''.
	\newblock
	\href{https://dx.doi.org/https://doi.org/10.1080/00207179.2016.1161830}{International
		Journal of Control {\bf 89}, 2466--2481}~(2016).
	
	\bibitem{KHANEJA2005}
	N.~Khaneja, T.~Reiss, C.~Kehlet, T.~Schulte-Herbrüggen, and S.~J. Glaser.
	\newblock ``Optimal control of coupled spin dynamics: design of nmr pulse
	sequences by gradient ascent algorithms''.
	\newblock
	\href{https://dx.doi.org/http://dx.doi.org/10.1016/j.jmr.2004.11.004}{Journal
		of Magnetic Resonance {\bf 172}, 296 -- 305}~(2005).
	
	\bibitem{SecondGRAPE}
	P.~de~Fouquieres, S.G. Schirmer, S.J. Glaser, and Ilya Kuprov.
	\newblock ``Second order gradient ascent pulse engineering''.
	\newblock
	\href{https://dx.doi.org/https://doi.org/10.1016/j.jmr.2011.07.023}{Journal
		of Magnetic Resonance {\bf 212}, 412 -- 417}~(2011).
	
	\bibitem{PerSilRou19}
	P.~S. Pereira~da Silva, H.~B. Silveira, and P.~Rouchon.
	\newblock ``Fast and virtually exact quantum gate generation in {U}(n) via
	iterative lyapunov methods''.
	\newblock
	\href{https://dx.doi.org/https://doi.org/10.1080/00207179.2019.1626023}{Int.
		J. of Control {\bf 94}, 84--99}~(2021).
	
	\bibitem{CODE_OCEAN_CONSTANT}
	P.~S. Pereira~da Silva, H.~B Silveira, and P.~Rouchon.
	\newblock ``{RIGA}, a fast algorithm for quantum gate generation [source
	code]''.
	\newblock \href{https://dx.doi.org/https://doi.org/10.24433/CO.0909026.v1}{Code
		Ocean}~(2019).
	
	\bibitem{CRAB}
	N.~Rach, M.~M. M\"uller, T.~Calarco, and S.~Montangero.
	\newblock ``Dressing the chopped-random-basis optimization: A bandwidth-limited
	access to the trap-free landscape''.
	\newblock \href{https://dx.doi.org/10.1103/PhysRevA.92.062343}{Phys. Rev. A
		{\bf 92}, 062343}~(2015).
	
	\bibitem{MacShaTanFra15}
	Shai Machnes, Elie Ass\'emat, David Tannor, and Frank~K. Wilhelm.
	\newblock ``Tunable, flexible, and efficient optimization of control pulses for
	practical qubits''.
	\newblock \href{https://dx.doi.org/10.1103/PhysRevLett.120.150401}{Phys. Rev.
		Lett. {\bf 120}, 150401}~(2018).
	
	\bibitem{CODE_OCEAN_SMOOTH}
	P.~S. Pereira~da Silva and P.~Rouchon.
	\newblock ``{RIGA} and {FPA}, quantum control with smooth control pulses
	[source code]''.
	\newblock \href{https://dx.doi.org/https://doi.org/10.24433/CO.3293651.v2}{Code
		Ocean}~(2020).
	
	\bibitem{CRAB_GOAT_GRAPE}
	B.~Riaz, C.~Shuang, and S.~Qamar.
	\newblock ``Optimal control methods for quantum gate preparation: a comparative
	study''.
	\newblock \href{https://dx.doi.org/10.1007/s11128-019-2190-0}{Quantum Inf
		Process{\bf 18}}~(2019).
	
	\bibitem{LeuAbdKocSch17}
	Nelson Leung, Mohamed Abdelhafez, Jens Koch, and David Schuster.
	\newblock ``Speedup for quantum optimal control from automatic differentiation
	based on graphics processing units''.
	\newblock \href{https://dx.doi.org/10.1103/PhysRevA.95.042318}{Phys. Rev. A
		{\bf 95}, 042318}~(2017).
	
	\bibitem{HeeEtAl15}
	Reinier~W. Heeres, Brian Vlastakis, Eric Holland, Stefan Krastanov, Victor~V.
	Albert, Luigi Frunzio, Liang Jiang, and Robert~J. Schoelkopf.
	\newblock ``Cavity state manipulation using photon-number selective phase
	gates''.
	\newblock \href{https://dx.doi.org/10.1103/PhysRevLett.115.137002}{Phys. Rev.
		Lett. {\bf 115}, 137002}~(2015).
	
	\bibitem{CavityTransmonGrape}
	R.~W. Heeres, P.~Reinhold, N.~Ofek, L.~Frunzio, L.~Jiang, M.~H. Devoret, and
	R.~J. Schoelkopf.
	\newblock ``Implementing a universal gate set on a logical qubit encoded in an
	oscillator''.
	\newblock \href{https://dx.doi.org/10.1038/s41467-017-00045-1}{Nature
		Communications{\bf 8}}~(2017).
	
	\bibitem{Maz14}
	Mazyar Mirrahimi, Zaki Leghtas, Victor~V Albert, Steven Touzard, Robert~J
	Schoelkopf, Liang Jiang, and Michel~H Devoret.
	\newblock ``Dynamically protected cat-qubits: a new paradigm for universal
	quantum computation''.
	\newblock New Journal of Physics {\bf 16}, 045014~(2014).
	\newblock
	url:~\href{http://stacks.iop.org/1367-2630/16/i=4/a=045014}{http://stacks.iop.org/1367-2630/16/i=4/a=045014}.
	
	\bibitem{Jin12}
	Jingfu Zhang, Raymond Laflamme, and Dieter Suter.
	\newblock ``Experimental implementation of encoded logical qubit operations in
	a perfect quantum error correcting code''.
	\newblock \href{https://dx.doi.org/10.1103/PhysRevLett.109.100503}{Phys. Rev.
		Lett. {\bf 109}, 100503}~(2012).
	
	\bibitem{Nig14}
	D.~Nigg, M.~Müller, E.~A. Martinez, P.~Schindler, M.~Hennrich, T.~Monz,
	M.~A. Martin-Delgado, and R.~Blatt.
	\newblock ``Quantum computations on a topologically encoded qubit''.
	\newblock \href{https://dx.doi.org/10.1126/science.1253742}{Science {\bf 345},
		302--305}~(2014).
	
	\bibitem{Got01}
	Daniel Gottesman, Alexei Kitaev, and John Preskill.
	\newblock ``Encoding a qubit in an oscillator''.
	\newblock \href{https://dx.doi.org/10.1103/PhysRevA.64.012310}{Phys. Rev. A
		{\bf 64}, 012310}~(2001).
	
	\bibitem{Mic16}
	Marios~H. Michael, Matti Silveri, R.~T. Brierley, Victor~V. Albert, Juha
	Salmilehto, Liang Jiang, and S.~M. Girvin.
	\newblock ``New class of quantum error-correcting codes for a bosonic mode''.
	\newblock \href{https://dx.doi.org/10.1103/PhysRevX.6.031006}{Phys. Rev. X {\bf
			6}, 031006}~(2016).
	
	\bibitem{Chi04}
	J.~Chiaverini, D.~Leibfried, T.~Schaetz, M.~D. Barrett, R.~B. Blakestad,
	J.~Britton, W.~M. Itano, J.~D. Jost, E.~Knill, R.~Ozeri, and D.~J. Wineland.
	\newblock ``Realization of quantum error correction''.
	\newblock \href{https://dx.doi.org/10.1038/nature03074}{Nature {\bf 432},
		602--605}~(2004).
	
	\bibitem{Cra16}
	J~Cramer, N~Kalb, MA~Rol, B~Hensen, MS~Blok, M~Markham, DJ~Twitchen, R~Hanson,
	and TH~Taminiau.
	\newblock ``Repeated quantum error correction on a continuously encoded qubit
	by real-time feedback''.
	\newblock \href{https://dx.doi.org/10.1038/ncomms11526}{Nature communications
		{\bf 7}, 11526}~(2016).
	
	\bibitem{Fu17}
	X.~Fu, M.~A. Rol, C.~C. Bultink, J.~van Someren, N.~Khammassi, I.~Ashraf,
	R.~F.~L. Vermeulen, J.~C. de~Sterke, W.~J. Vlothuizen, R.~N. Schouten, C.~G.
	Almudever, L.~DiCarlo, and K.~Bertels.
	\newblock ``An experimental microarchitecture for a superconducting quantum
	processor''.
	\newblock In Proceedings of the 50th Annual IEEE/ACM International Symposium on
	Microarchitecture.
	\newblock \href{https://dx.doi.org/10.1145/3123939.3123952}{Pages 813--825}.
	\newblock MICRO-50 '17New York, NY, USA~(2017). ACM.
	
	\bibitem{Ris15}
	D.~Ristè, Poletto S., M.-Z. Huang, A.~Bruno, V.~Vesterinen, O.-P. Saira,
	and L.~DiCarlo.
	\newblock ``Detecting bit-flip errors in a logical qubit using stabilizer
	measurements''.
	\newblock \href{https://dx.doi.org/10.1038/ncomms11526}{Nature communications
		{\bf 6}, 6983}~(2016).
	
	\bibitem{Mag15}
	E.~Magesan, S.~J. Srinivasan, A.~W. Cross, M.~Steffen, Jay~M. Gambetta, and
	J.~M. Chow.
	\newblock ``Demonstration of a quantum error detection code using a square
	lattice of four superconducting qubits''.
	\newblock \href{https://dx.doi.org/10.1038/ncomms7979}{Nature communications
		{\bf 6}, 6979}~(2015).
	
	\bibitem{Cor07}
	J.M. Coron.
	\newblock ``Control and nonlinearity''.
	\newblock
	\href{https://dx.doi.org/https://doi.org/http://dx.doi.org/10.1090/surv/136}{Mathematical
		surveys and monographs}. American Mathematical Society. ~(2007).
	
	\bibitem{Die98}
	F.~Diele, L.~Lopez, and R.~Peluso.
	\newblock ``The {C}ayley transform in the numerical solution of unitary
	differential systems''.
	\newblock Advances in Computational Mathematics {\bf 8}, 317--334~(1998).
	\newblock
	url:~\href{https://doi.org/10.1023/A:1018908700358}{doi.org/10.1023/A:1018908700358}.
	
	\bibitem{BreCamPet95}
	K.~E. Brenan, S.~L. Campbell, and L.~R. Petzold.
	\newblock ``Numerical solution of initial-value problems in
	differential-algebraic equations''.
	\newblock
	\href{https://dx.doi.org/https://doi.org/10.1137/1.9781611971224}{Springer-Verlag}.
	New York~(1995).
	
	\bibitem{Kel75}
	J.~B. Keller.
	\newblock ``Closest unitary, orthogonal and hermitian operators to a given
	operator''.
	\newblock \href{https://dx.doi.org/https://doi.org/10.2307/2690338}{Mathematics
		Magazine {\bf 48}, 192--197}~(1975).
	
	\bibitem{Hig08}
	N.~J. Higham.
	\newblock ``Functions of matrices''.
	\newblock \href{https://dx.doi.org/10.1137/1.9780898717778}{Society for
		Industrial and Applied Mathematics ({SIAM})}. ~(2008).
	
	\bibitem{Kha02}
	H.K. Khalil.
	\newblock ``Nonlinear systems''.
	\newblock Pearson Education. Prentice Hall. ~(2002).
	\newblock
	url:~\href{https://www.pearson.com/en-us/subject-catalog/p/nonlinear-systems/P200000003306/9780130673893}{www.pearson.com/en-us/subject-catalog/p/nonlinear-systems/P200000003306/9780130673893}.
	
\end{thebibliography}
\end{document}